\pdfoutput=1
\documentclass[a4paper,12pt,twoside,openright]{book}

\usepackage{microtype}
\usepackage{fullpage}
\usepackage{geometry}
\usepackage{setspace}
\usepackage[utf8]{inputenc}
\usepackage{hyperref}
\usepackage{booktabs}
\usepackage{diagbox}
\usepackage{multirow}
\usepackage{adjustbox}
\usepackage[table]{xcolor}
\usepackage{ntheorem}
\usepackage{algorithm}
\usepackage{algpseudocode}

\theoremseparator{:}
\newtheorem{hyp}{Hypothesis}

\usepackage{fancyhdr}
\fancypagestyle{plain}{%
\fancyfoot[LE,RO]{\hfill\thepage\hfill}
}

\usepackage{makecell}

\usepackage{caption}
\usepackage{subcaption}

\usepackage{tikz}
\usetikzlibrary{arrows.meta, shapes, calc}
\usetikzlibrary{quantikz}

\usepackage{dsfont}
\usepackage{cite}
\usepackage{mathtools}
\usepackage{amsmath}
\usepackage{amssymb}
\usepackage[polish,english]{babel}
\usepackage[utf8]{inputenc}
\usepackage[T1]{fontenc}

\usepackage{graphics}

\usepackage{enumerate}
\usepackage{xspace}

\usepackage{listings}
\usepackage{xcolor}

\usepackage{algorithm}
\usepackage{algpseudocode}

\usetikzlibrary{quantikz}
\usetikzlibrary{shapes}
\usetikzlibrary{arrows}
\usetikzlibrary{shadows}
\usetikzlibrary{calc}
\usetikzlibrary{positioning}

\newsavebox\CPU
\sbox\CPU{
	\begin{tikzpicture}[line cap=round,line join=round,scale=1]
		
		\draw[thick,rounded corners,fill=white] (0,0) rectangle ++ (1.2,1.2);
		\draw[thick,rounded corners] (0.15,0.15) rectangle ++ (0.9,0.9);
		\draw[draw,thick] (0.6,0) -- (0.6,-0.2);
		\draw[draw,thick] (0.4,0) -- (0.4,-0.2);
		\draw[draw,thick] (0.8,0) -- (0.8,-0.2);
		
		\draw[draw,thick] (0.6,1.2) -- (0.6,1.4);
		\draw[draw,thick] (0.4,1.2) -- (0.4,1.4);
		\draw[draw,thick] (0.8,1.2) -- (0.8,1.4);
		
		\draw[draw,thick] (0,0.4) -- (-0.2,0.4);
		\draw[draw,thick] (0,0.6) -- (-0.2,0.6);
		\draw[draw,thick] (0,0.8) -- (-0.2,0.8);
		
		\draw[draw,thick] (1.2,0.4) -- (1.4,0.4);
		\draw[draw,thick] (1.2,0.6) -- (1.4,0.6);
		\draw[draw,thick] (1.2,0.8) -- (1.4,0.8);
		
		\node[] at (0.6,0.6) {\small{\textsf{CPU}}};
\end{tikzpicture}}

\tikzstyle{qdev} = [draw, rectangle, 
minimum height=3em, fill=none, minimum width=2.5em, rounded corners, node distance=1cm]

\tikzstyle{anch} = [fill=none, minimum width=3em, minimum height=0.1cm, node distance=1cm]

\tikzstyle{jam1} = [draw, rectangle, dashed, line width=0.5pt, fill=blue!35,
minimum width=2.5em, minimum height=1.25cm, rounded corners]

\tikzstyle{jam2} = [draw, rectangle, dashed, line width=0.5pt, fill=blue!35,
minimum width=3em, minimum height=3cm, rounded corners]

\tikzstyle{jam3} = [draw, rectangle, dashed, line width=0.5pt, fill=blue!35,  minimum width=3em, minimum height=4.5cm, rounded corners]

\tikzstyle{ssfb} = [draw, rectangle, dotted, line width=0.75pt, fill=black!10, minimum width=3em, minimum height=3cm, rounded corners]

\tikzstyle{vqfe} = [draw, rectangle, dashdotted, line width=0.5pt, fill=yellow!75, minimum width=3em, minimum height=1in, rounded corners]

\tikzstyle{vqfe-short} = [draw, rectangle, dashdotted, line width = 0.5pt,    fill=yellow!75, minimum width=3em, minimum height=0.5in, rounded corners]

\tikzstyle{cpu} = [node contents=\usebox{\CPU},scale=0.75,]

\newcommand{\tr}{\ensuremath{\mathrm{tr}}}
\newcommand{\Tr}{\tr}
\newcommand{\id}{\mathds{1}}

\newcommand{\ie}{\emph{ie.}}
\newcommand{\Jam}[1]{\ensuremath{\mathcal{J}(#1)}}

\newcommand{\figref}[1]{Fig.~(\ref{#1})}
\renewcommand{\eqref}[1]{Eq.~(\ref{#1})}
\newcommand{\DOI}[1]{DOI:\href{https://doi.org/#1}{#1}}
\newcommand{\arXiv}[1]{arXiv:\href{https://arXiv.org/abs/#1}{#1}}

\newtheorem{prop}{Proposition}
\newtheorem{definition}{Definition}[section]

\renewcommand{\ket}[1]{\ensuremath{|#1\rangle}}
\renewcommand{\bra}[1]{\ensuremath{\langle#1|}}

\newcommand{\ii}{\ensuremath{\mathrm{i}}}

\newenvironment{proof}[1][Proof]{\noindent\textbf{#1.} }{\hfill 
    \rule{0.5em}{0.5em}}

\definecolor{codegreen}{rgb}{0,0.6,0}
\definecolor{codegray}{rgb}{0.5,0.5,0.5}
\definecolor{codepurple}{rgb}{0.58,0,0.82}
\definecolor{backcolour}{rgb}{0.95,0.95,0.92}

\lstdefinestyle{mystyle}{
    backgroundcolor=\color{backcolour},   
    commentstyle=\color{codegreen},
    keywordstyle=\color{magenta},
    numberstyle=\tiny\color{codegray},
    stringstyle=\color{codepurple},
    basicstyle=\ttfamily\footnotesize,
    breakatwhitespace=false,         
    breaklines=true,                 
    captionpos=b,                    
    keepspaces=true,                 
    numbers=left,                    
    numbersep=5pt,                  
    showspaces=false,                
    showstringspaces=false,
    showtabs=false,                  
    tabsize=2
}

\usepackage[subfigure]{tocloft}
\newcommand{\listnotename}{List of TODOs}
\newlistof[chapter]{note}{todos}{\listnotename}

\begin{document}
\singlespace

\frontmatter

\pagenumbering{arabic}

\begin{titlepage}
	\begin{center}
		\hrule
		\hrule
		\vspace{.5em}
		\textsc{\Large Reinforcement learning-assisted quantum architecture search for variational quantum algorithms}
		\vspace{.5em}
		\hrule
		\hrule
		\vspace{1.5em}
		\textsc{\large Doctoral dissertation}\\
		\vspace*{1em}
		{\large mgr Akash \textsc{Kundu}}
		{\centering
			\includegraphics[width=0.9\textwidth]{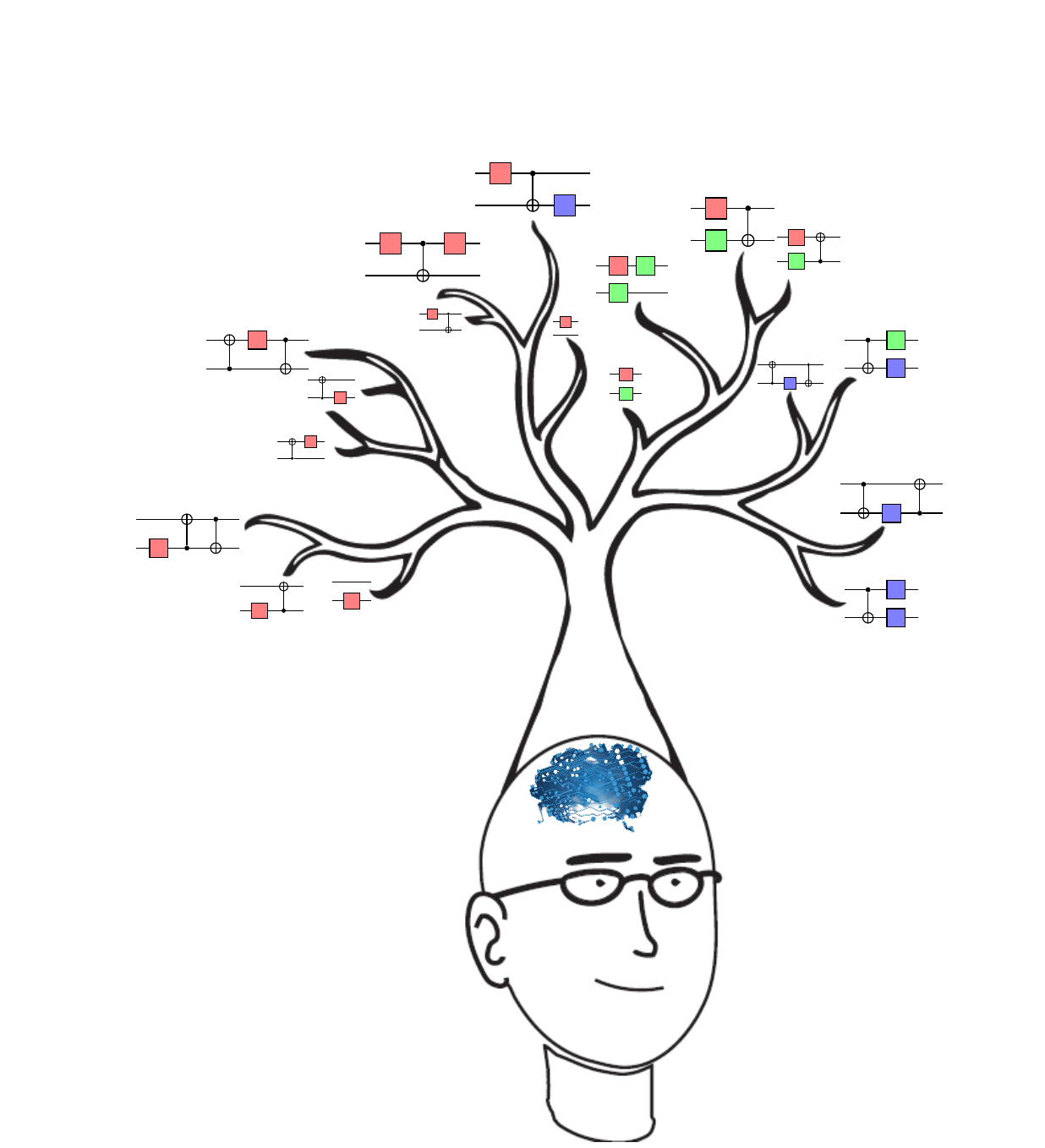}}
		\vfill
		
	\end{center}
{\textbf{Idea of illustration:} Akash Kundu,\\ \textbf{Illustrator:} Ludmila Botelho.}
\end{titlepage}

\begin{titlepage}
	\begin{center}
		\includegraphics[width=0.4\textwidth]{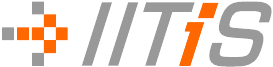}\\
		\vspace{0.5em}
		\textsc{\large Institute of Theoretical and Applied Informatics, Polish 
			Academy of Sciences}
		\vspace*{1in}
		\hrule
		\vspace*{0.5em}
		\textsc{\Large Reinforcement learning-assisted quantum architecture search for variational quantum algorithms}
		\vspace*{0.5em}
		\hrule
		\vspace*{1em}
		\textsc{\large Doctoral dissertation}
		\par
		\vspace{1.5in}
		{\large mgr Akash \textsc{Kundu}}\\
		Supervisor:\\ dr hab. Jaros\l{}aw \textsc{Miszczak}\\
		\vspace{3em}
		\textsc{In partial fulfillment of the requirements for the degree of Doctor of Philosophy}
		\vfill
		{Gliwice, \today}
	\end{center}
\end{titlepage}

\setcounter{page}{4}

\cleardoublepage

\tableofcontents

\chapter*{List of publications}\addcontentsline{toc}{chapter}{List of publications}
\markboth{\MakeUppercase{List of publications}}{\MakeUppercase{List of publications}}

Publications relevant to this dissertation are highlighted using \textbf{bold font}.


\vspace{0.2cm}
\begin{spacing}{0.9}
	\begin{enumerate}
		
		{\bf \item  \textbf{Akash Kundu}, Przemys\l{}aw Bede\l{}ek, Mateusz Ostaszewski, Onur Danaci, Yash J. Patel, Vedran Dunjko, Jaros\l{}aw A. Miszczak; 
			\emph{Enhancing variational quantum state diagonalization using reinforcement learning techniques},
			New Journal of Physics, Vol. 26, pp.~013034 (2024)},
		\arXiv{2306.11086}, 
		\DOI{10.1088/1367-2630/ad1b7f}
		\newline
		{Code}: \url{https://github.com/iitis/RL_for_VQSD_ansatz_optimization}
		
		{\bf \item Yash J. Patel,  \textbf{Akash Kundu}, Mateusz Ostaszewski, Xavier Bonet-Monroig, Vedran Dunjko, Onur Danaci; 
			\emph{Curriculum reinforcement learning for quantum architecture search under hardware errors}; In The Twelfth International Conference on Learning Representations, 2024, \url{https://openreview.net/forum?id=rINBD8jPoP}.} \arXiv{2402.03500}
		\newline
		{Code}: \url{https://anonymous.4open.science/r/CRLQAS}
		
		{\bf \item  \textbf{Akash Kundu}, Jaros{\l}aw A. Miszczak, 
			\emph{Variational certification of quantum devices}; Quantum Science and Technology, Vol.~7, No.~4, pp~045017 (2022)}, \DOI{10.1088/2058-9565/ac8572}, \arXiv{2011.01879}.
		\newline
		{Code}: \url{https://github.com/iitis/variational_channel_fidelity}
		
		\item  \textbf{Akash Kundu}, Ludmila Botelho, Adam Glos, 
		\emph{Hamiltonian-oriented homotopy quantum approximate optimization algorithm}; Physical Review A, Vol. 109, pp.~022611 (2024), arXiv:{2301.13170} \DOI{10.1103/PhysRevA.109.022611}.
		\newline
		{Code}: \url{https://doi.org/10.5281/zenodo.7585691}
		
		\item Ludmila Botelho, Adam Glos,  \textbf{Akash Kundu}, Jaros{\l}aw Adam Miszczak, {\"O}zlem Salehi, Zolt{\'a}n Zimbor{\'a}s, 
		\emph{Error mitigation for variational quantum algorithms through mid-circuit measurements}; Physical Review A, Vol. 105, No. 2, pp.~022441 (2022), arXiv:{2108.10927} \DOI{10.1103/PhysRevA.105.022441}.
		\newline
		{Code}: \url{https://github.com/iitis/method-of-continuation-qaoa}

		\item Krzysztof Domino,  \textbf{Akash Kundu}, {\"O}zlem Salehi, Krzysztof Krawiec,
		\emph{Quadratic and higher-order unconstrained binary optimization of railway rescheduling for quantum computing}; Quantum Information Processing, vol. 21, No. 9, pp. 337 (2022), \DOI{10.1007/s11128-022-03670-y}.
		\newline
		{Code}: \url{https://github.com/iitis/railways_HOBO}
		
		\item Adam Glos,  \textbf{Akash Kundu}, {\"O}zlem Salehi,
		\emph{Optimizing the Production of Test Vehicles Using Hybrid Constrained Quantum Annealing}, SN Computer Science, Vol.~4, 609 (2023), \DOI{10.1007/s42979-023-02071-x}.
		\newline
		{Code}: \url{https://github.com/iitis/bmw-qchallenge}

		\item  \textbf{Akash Kundu}, Jaros\l{}aw A. Miszczak,
		\emph{Transparency and Enhancement in Fast and Slow Light in q-Deformed Optomechanical System}, Annalen der Physik, Vol.~534, No.~8, pp.~2200026 (2022), \DOI{10.1002/andp.202200026}, \arXiv{2205.15800}.
		
		\item Jia-Xin Peng,  \textbf{Akash Kundu}, Zeng-Xing Liu, Atta ur Rahman, Naeem Akhtar, and M. Asjad,
		\emph{Vector photon-magnon-phonon coherence in a polarized microwave driven cavity magnomechanical system
		}, Physical Review B, Vol.~109, pp~064412 (2024), \DOI{10.1103/PhysRevB.109.064412}.
		
		\item  \textbf{Akash Kundu}, Chao Jin, Jia-Xin Peng,
		\emph{Study of the optical response and coherence of a quadratically coupled optomechanical system}, Physica Scripta, Vol.~96, No.~6, pp~065102 (2021), \DOI{10.1088/1402-4896/abee4f}.
		
		\item Hao-Jie Cheng, Shu-Jie Zhou, Jia-Xin Peng,  \textbf{Akash Kundu}, Hong-Xue Li, Lei Jin, Xun-Li Feng,
		\emph{Tripartite entanglement in a Laguerre--Gaussian rotational-cavity system with an yttrium iron garnet sphere},
		Journal of the Optical Society of America B, Vol.~38, No.~2, pp.~285--293 (2021), \DOI{10.1364/JOSAB.405097}.
		
		\item  \textbf{Akash Kundu}, Chao Jin, Jia-Xin Peng,
		\emph{Study of the optical response and coherence of a quadratically coupled optomechanical system}, Physica Scripta, Vol.~96, No.~6, pp~065102 (2021), \DOI{10.1088/1402-4896/abee4f}.
		
		\item  \textbf{Akash Kundu}, Chao Jin, Jia-Xin Peng, 
		\emph{Optical response of a dual membrane active--passive optomechanical cavity}; Annals of Physics, vol. 249, pp.~168465 (2021), \DOI{10.1016/j.aop.2021.168465}.
		
		\item  \textbf{Akash Kundu}, SD Pathak, VK Ojha; 
		\emph{Interacting tachyonic scalar field}, Communications in Theoretical Physics, Vol.~38, No.~2, pp~285--293 (2021), \DOI{10.1088/1572-9494/abcfb1}.
		
		\item Turbasu Chatterjee, Shah Ishmam Mohtashim,  \textbf{Akash Kundu}; 
		\emph{On the variational perspectives to the graph isomorphism problem}, \arXiv{2111.09821 2021}.
		
		\item  \textbf{Akash Kundu}, Tamal Acharya, Aritra Sarkar;
		\emph{A scalable quantum gate-based implementation for causal hypothesis testing}, 2023, \DOI{10.48550/arXiv.2209.02016}.
		\newline
		{Code}: \url{https://github.com/Advanced-Research-Centre/QaCHT}
		
	\end{enumerate}
	
	
\end{spacing}




\newpage

\chapter*{Abstract in Polish}\addcontentsline{toc}{chapter}{Abstract in Polish}
\markboth{\MakeUppercase{Polish abstract}}{\MakeUppercase{Polish abstract}}

\begin{otherlanguage}{polish}
	
	
	Istotną przeszkodą w erze zaszumionych komputerów kwantowych średniej skali (ang. NISQ -- Noisy Intermediate-Scale Quantum) jest konstrukcja obwodów kwantowych, które pozwolą na wykonanie użytecznych algorytmów kwantowych i są zgodne z ograniczeniami narzuconymi przez obecne ograniczenia sprzętu kwantowego. Aby sprostać tym wyzwaniom w obecnie dostępnych urządzeniach kwantowych, opracowano wariacyjne algorytmy kwantowe (ang. VQA -- Variational Quantum Algorithms), które stanowią klasę hybrydowych algorytmów kwantowo-klasycznej dla problemów optymalizacji. Jednakże ogólna wydajność wariacyjnych algorytmów kwantowe zależy od (1) strategii inicjalizacji obwodu wariacyjnego, (2) struktury obwodu (znanej również jako ansatz) oraz (3) konfiguracji funkcji kosztu. Koncentrując się na (2), w tej pracy zaproponowane są metody poprawy wydajność wariacyjnych algorytmów kwantowych poprzez automatyzację wyszukiwania optymalnej struktury obwodów wariacyjnych za pomocą uczenia się ze wzmocnieniem (ang. RL -- Reinforcement Learning).
	W ramach pracy skupiamy się na określeniu optymalności obwodu poprzez ocenę jego głębokości, całkowitą liczbę bramek i parametrów oraz dokładności w rozwiązaniu zadanego problemu.
	Zadanie automatyzacji wyszukiwania optymalnych obwodów kwantowych znane jest jako wyszukiwanie architektury kwantowej (ang. QAS -- Quantum Architecture Search). Większość badań w zakresie wyszukiwania architektury kwantowej koncentruje się na scenariuszu bezszumowym.
	W związku z tym wpływ szumu na proces wyszukiwania architektury pozostaje niewystarczająco zbadany. W tej pracy zajmujemy się tym problemem poprzez wprowadzenie techniki łączącej kodowanie obwodów kwantowych opartego na tensorach, ograniczenie dynamiki środowiska w celu efektywnego badania przestrzeni poszukiwań możliwych obwodów, schemat zatrzymywania epizodów w celu nakierowania agenta na znajdź krótsze obwody, oraz poprzez wykorzystanie podwójnie głęboką sieć Q (DDQN) z polityką $\epsilon$ dla lepszej stabilności. Eksperymenty numeryczne na bezszumowym i zaszumionym sprzęcie kwantowym pokazują, że w radzeniu sobie z wybranymi algorytmami wariacyjnymi, zaproponowana metoda wyszukiwania architektury przewyższa istniejący metody. Dodatkowo metody, które proponujemy w pracy, można zostać łatwo zaadoptowane do szerokiego zakresu innych VQA.
\end{otherlanguage}

\newpage
\chapter*{Abstract in English}\addcontentsline{toc}{chapter}{Abstract in 
	English}
\markboth{\MakeUppercase{English abstract}}{\MakeUppercase{English abstract}}



A significant hurdle in the noisy intermediate-scale quantum (NISQ) era is identifying functional quantum circuits. These circuits must also adhere to the constraints imposed by current quantum hardware limitations. Variational quantum algorithms (VQAs), a class of quantum-classical optimization algorithms, were developed to address these challenges in the currently available quantum devices. However, the overall performance of VQAs depends on the initialization strategy of the variational circuit, the structure of the circuit (also known as ansatz) and the configuration of the cost function. Focusing on the structure of the circuit, in this thesis, we improve the performance of VQAs by automating the search for an optimal structure for the variational circuits using reinforcement learning (RL). 
Within the thesis, the optimality of a circuit is determined by evaluating its depth, the overall count of gates and parameters, and its accuracy in solving the given problem.
The task of automating the search for optimal quantum circuits is known as quantum architecture search (QAS). The majority of research in QAS is primarily focused on a noiseless scenario.
Yet, the impact of noise on the QAS remains inadequately explored. In this thesis, we tackle the issue by introducing a tensor-based quantum circuit encoding, restrictions on environment dynamics to explore the search space of possible circuits efficiently, an episode halting scheme to steer the agent to find shorter circuits, a double deep Q-network (DDQN) with an $\epsilon$-greedy policy for better stability. The numerical experiments on noiseless and noisy quantum hardware show that in dealing with various VQAs, our RL-based QAS outperforms existing QAS. Meanwhile, the methods we propose in the thesis can be readily adapted to address a wide range of other VQAs.

\mainmatter


\setcounter{page}{13}

\chapter*{Abbreviations used in thesis}
\addcontentsline{toc}{chapter}{Abbreviations used in thesis}
\begin{table}[h!]
	\centering
	\scalebox{0.8}{
	\begin{tabular}[scale=0.8]{c|c||c|c}
		Abbreviation    & Full form & Abbreviation & Full form \\
		\hline
		RL  & Reinforcement learning  &   RHS  &  Right hand side  \\
		\hline
		NISQ & Noisy intermediate-scale quantum  &  LHS  &  Left hand side  \\
		\hline
		VQA & Variational quantum algorithm  & BFGS  &  Broyden–Fletcher–Goldfarb–Shanno
		\\
		\hline
		PQC & Parametric quantum circuit  & ADAM  &  Adaptive moment estimation
		\\
		\hline
		QAS  &  Quantum architecture search  & \makecell{DNN\\(in Fig.~\ref{fig:rl_vqsd_diagram})}  &  Deep neural network  \\
		\hline
		DQN  &  Deep Q-network  & QPE  &  Quantum phase estimation  \\
		\hline
		DDQN  &  Double Deep Q-network  & IQFT  &  \makecell{Inverse quantum fourier\\transform}  \\
		\hline
		CNOT/CX  &  Controlled NOT  & $H_2$  &  Hydrogen  \\
		\hline
		CRL  &  Curriculum reinforcement learning  & $LiH$  &  Lithium hydride \\
		\hline
		VQE  &  Variational quantum eigensolver  & $H_2O$  &  Hydrogen oxide  \\
		\hline
		CRLVQE  &  \makecell{Curriculum reinforcement learning \\ for variational quantum eigensolver}  &GPU  &  Graphical processing unit  \\
		\hline
		VQSD  &  \makecell{Variational quantum state\\diagonalization}   &CPU  &  Central processing unit  \\
		\hline
		ML  &  Machine learning  &TBE  &  Tensor based encoding  \\
		\hline
		NN  &  Neural network  &IE  &  Integer encoding  \\
		\hline
		UCC  &  Unitary coupled cluster  &cq/ctrl  & Control  \\
		\hline
		UCCSD  &  \makecell{Unitary coupled cluster\\single and double}  &CPTP  &  \makecell{Completely positive trace\\preserving}  \\
		\hline
		HEA  &  Hardware efficient ansatz  & RH  &  Random halting  \\
		\hline
		LHEA  &  \makecell{Layered hardware efficient \\ansatz}  & wo-RH  &  Without random halting  \\
		\hline
		QAOA  &  \makecell{Quantum approximate optimization\\ algorithm}  & SPSA  &  \makecell{Simultaneous perturbation stochastic\\ approximation
		}  \\
		\hline
		MDP  &  Markov decision process  &VQFE  &  \makecell{Variational quantum fidelity\\estimation}  \\
		\hline
		TD  &  Temporal difference  &PTM  &  Pauli-transfer matrix  \\
		\hline
		FQI  &  Fitted Q-iteration  & SSFB  &  Sub and super fidelity bounds  \\
		\hline
		qPCA  &  \makecell{Quantum principle component\\ analysis}  &tq/targ  &  Target \\
		\hline
		COBYLA  &  \makecell{Constrained Optimization by\\ Linear Approximation}  &FB  &  Fidelity bound  \\
		\hline
		TFB  &  Truncated fidelity bound  & CZ  &  Controlled Z\\
		\hline
		RL-VQE  &  \makecell{Reinforcement learning assisted\\variational quantum eigensolver}  & JIT  & Just in time
	\end{tabular}}
\end{table}

\chapter{Introduction}\label{ch:intro}
Quantum computing leverages the principles of quantum mechanics to gain a distinct advantage in information processing. Ongoing worldwide endeavours are actively striving to materialize a sufficiently large, controllable, and programmable quantum computer. Corporations such as  Google~\cite{castelvecchi2017quantum}, IBM~\cite{IBMcompute}, Rigetti~\cite{Rigetticompute}, and Intel~\cite{Intelcompute} are utilizing superconducting qubits where the quantum processing unit utilizes a superconducting architecture. Whereas, Honeywell~\cite{Honeywellcompute} and IonQ~\cite{IonQcompute} utilize ion traps as quantum processors where charged atoms, i.e. ions, are used as qubits due to the fact that ions can be trapped in one precise location with the help of electric fields. Meanwhile, D-Wave~\cite{DWavecompute} quantum computers are based on quantum annealing~\cite{das2008colloquium}. The qubits are made from tiny superconducting loops.


In the early and late 90s, pure quantum algorithms were introduced, such as Shor's~\cite{shor1994algorithms}, which is used to find the prime factor of integers, and Grover's~\cite{grover1996fast} algorithm, which is used to search a unique input from an unstructured dataset. To reveal the true potential of these algorithms and achieve quantum advantage, we require quantum hardware with thousands to millions of qubits. Unfortunately, the current quantum devices are of small scale with 5-200 qubits, noise prone, and have constrained connectivity between qubits. These devices are called Noisy intermediate-scale quantum (NISQ) computers~\cite{preskill2018quantum}. At the time of this thesis (2021-2023), no quantum devices exist that can execute quantum algorithms demonstrating a provable quantum advantage for real-world use cases.

To deal with these limitations and exploit the NISQ devices, variational quantum~algorithms (VQAs)~\cite{mcclean2016theory} was introduced, where the task of solving a quantum problem is distributed into quantum and classical computers. The VQAs are fundamentally comprised of three essential components: a parametrized quantum circuit (PQC) or ansatz, a cost function that encodes the problem, and a classical optimization procedure responsible for adjusting the PQC's parameters in order to minimize the cost function. Ongoing research efforts are dedicated to exploring and comprehending the potential of each of these building blocks within the realm of VQAs~\cite{cerezo2020variational}.

The conventional approach for creating an ansatz involves predefining its structure before getting underway with the algorithm. Based on the user's ambition, the structure of the ansatz can be driven by physical considerations~\cite{peruzzo2014variational} or hardware constraints~\cite{kandala2017hardware}.

Nonetheless, fixing the ansatz's structure imposes a significant restriction on exploring the cost function landscape and can prevent us from reaching the true solution. Sophisticated methods~\cite{botelho2022error, kundu2023hamiltonian, grimsley2019adaptive, bako2022near, glos2022space} have been introduced to enhance the performance of VQAs. 
Meanwhile, to avoid these limitations, recent attention has been directed towards automating the construction of ansatz~\cite{grimsley2019adaptive, tang2021qubit, tang2021qubit, anastasiou2022tetris, zhang2022differentiable}, known as quantum architecture search (QAS). The QAS eliminates the need for domain-specific expertise, and it has the ability to produce an ansatz tailored for specific VQAs. Given a finite set of quantum gates, the objective of QAS is to discover the optimal arrangement of quantum operators in the form of an ansatz that minimizes the cost function. The readers are encouraged to check the section~\ref{app:overview_of_architecture_search_methods} for an overview of recently proposed methods for QAS.

A solution to address the challenges in QAS involves the application of reinforcement learning (RL) techniques, as proposed in~\cite{ostaszewski2021reinforcement, fosel2021quantum, kuo2021quantum, ye2021quantum}. In the RL-based QAS methods, the cost function is optimized independently using a classical optimizer, providing an intermittent signal contributing to the cumulative reward function. This reward function, in turn, updates a policy that aims at maximizing expected returns. Based on the return, the RL-agent selects an optimal action for subsequent steps.

Meanwhile, to minimize the adverse effects of gate errors, constrained connectivity, and decoherence, it is crucial to design an ansatz that uses as few quantum gates and computational steps as possible and is as shallow as possible in terms of their depth. Gate errors, restricted connectivity, and decoherence are common challenges in NISQ devices that cause inaccuracies and lead to the loss of quantum information. By keeping the circuits gate-frugal (using fewer gates) and shallow (reducing the depth), quantum computations can be more resilient and less susceptible to the negative impacts of these quantum computing issues. This approach enhances the stability and reliability of quantum algorithms and makes them more suitable for practical applications.

In this thesis, we propose and analyze an RL-based QAS method whose agent operates on a double deep-Q network (DDQN) and an $\epsilon$-greedy policy. The proposed QAS method improves the performance of the RL-agent to not only propose a very compact ansatz with fewer gates, depth, and number of parameters but also return a very low error while solving the problem. To achieve these results, we utilize:
\begin{enumerate}
	\item \textbf{A tensor-based encoding scheme} to encode the ansatz as an observable for the RL-agent.
	\item \textbf{A one-hot encoding scheme} to construct the action space. Each action is represented by a parameterized rotation in either X, Y, or Z direction or a controlled-NOT (\texttt{CNOT}) gate.
	\item \textbf{A dense and a sparse reward function} to encode the cost function.
	\item \textbf{An illegal action scheme} to impose restrictions on the environment dynamics to explore the search space of possible ansatz efficiently.
	\item \textbf{A random halting scheme} is an episode halting scheme to encourage the agent to find a shorter gate and depth ansatz.
\end{enumerate}

These components are briefly discussed in upcoming sections. Using these, we automate the search for an optimal ansatz that finds the ground state of molecules using a curriculum reinforcement learning-based variational quantum eigensolver (CRLVQE). We also utilize a reinforcement learning enhanced variational quantum state diagonalization (VQSD) algorithm that finds the optimal structure of an ansatz that diagonalizes an arbitrary quantum state.

Our results correspond to achieving an error $10^{3}$ times lower than that of the previously proposed ansatz using less than half of the gates, indicating that the RL-agent enhances the exploration of the optimization of the cost function landscape leading to the design of smaller ansatz. These developments are summarized in the following hypothesis:
\begin{hyp}\label{hyp:1st_hypothesis}
	Utilizing reinforcement learning techniques enhances the exploration of the optimization landscape and leads to ansatz designs for variational quantum algorithms (VQAs) with minimal gate count, depth, and parameters while consistently achieving a low error level in cost function evaluation.
\end{hyp}

The \textbf{Hypothesis~\ref{hyp:1st_hypothesis}} deals with the performance of our RL-based algorithm under an idealistic scenario in the absence of any kind of device noise. Meanwhile, most algorithms for QAS have been formulated under the assumption of a noiseless scenario, free from physical noise and under consideration of all-to-all qubit connectivity, and there has been very little progress in automating the QAS problem~\cite{du2022quantum}. However, in order to make the algorithm physically realizable, it is important to show that the QAS algorithm can solve the problem and provide us with a compact ansatz in the presence of constraints imposed by current NISQ devices, characterized by limited qubit connectivity and susceptibility to noise.

Noise imposes severe effects on the cost function landscape, and noise can cause the optimization process to get stuck in regions of the landscape where the gradients are vanishing to guide further progress, thereby hindering the overall optimization effort~\cite{wang2021noise}. The phenomena of having a flat or extremely shallow cost function landscape are known as \textit{barren plateau}~\cite{mcclean2018barren}. Meanwhile, non-unitary noisy channels such as amplitude damping, decoherence\footnote{A coherence is a phase damping channel that belongs to a class of doubly stochastic channels~\cite{helm2009quantum}.}, and thermal relaxation transform the cost function landscape by exponentially converting the global minima into local optima. This increases the complexity of the optimization process~\cite{fontana2022non}. Hence, performing QAS in the presence of noise is a critical step toward understanding and ultimately overcoming the challenges posed by real-world quantum hardware and advancing the field of quantum architecture search on NISQ devices.


In this thesis, we utilize a curriculum-based RL method for QAS and show that our algorithm can efficiently solve quantum chemistry problems by finding the ground state of molecules in a realistic noisy scenario with very few gates and lower energy error in estimating the ground state.

These results can be summarized through the following hypothesis:

\begin{hyp}\label{hyp:2nd_hypothesis}
	Combining binary encoding for a quantum circuit, action space pruning, a variable episode length, and curriculum-based reinforcement learning methods can enhance the agent's learning, facilitating the quantum chemistry problem solution irrespective of quantum hardware noise and connectivity constraints.
\end{hyp}

The majority of research in QAS
focuses on a noiseless scenario. Yet, the impact of noise on the QAS
remains inadequately explored. Through \textbf{Hypothesis~\ref{hyp:2nd_hypothesis}}, we stress filling this gap by solving the ground state finding problem of difference molecules in realistic noisy noise scenarios. These realistic scenarios are imported from IBM quantum devices such as \texttt{ibmq\_mumbai} and \texttt{ibmq\_ourense}.

The rest of the thesis is organized as follows. Chapter~\ref{ch:preliminaries} initiates by giving a brief introduction to variational quantum algorithms, including their components (such as ansatz and cost function) and the basics of reinforcement learning with suitable examples. While introducing RL, we primarily focus on model-free learning and algorithms. At the end of this Chapter, we discuss various approaches to QAS and provide an overview of our RL-based QAS approach. In Chapter~\ref{ch:vqsd_using_rl}, we discuss an RL-assisted variational quantum state diagonalization, namely the RL-VQSD algorithm. The main task of RL-VQSD is to find a compact structure of an ansatz that diagonalizes an arbitrary quantum state with very low error in eigenvalue and eigenvector estimation. The Chapter starts by benchmarking the existing quantum state diagonalization algorithm. Afterwards, we show that the RL-VQSD proposed a more compact ansatz containing a smaller gate and lower depth compared to state-of-the-art diagonalization algorithms in a noiseless scenario. The Chapter ends by providing a pointwise summary of the algorithm and results. In Chapter~\ref{ch:crlvqe}, we utilize a curriculum-based reinforcement learning (CRL) approach to find the ground state of molecules under realistic noisy scenarios. We show that the ansatz proposed by the CRL-agent is more compact in terms of the number of gates and depth as well as provides better accuracy compared to existing variational quantum eigensolver. This is followed by a thorough discussion of how different novel schemes, RL-state encoding, and the choice of reward function impact the performance of the CRLVQE algorithm in various realistic noisy scenarios. The Chapter ends by giving a pointwise summary of the algorithm and the results. In the following Chapter~\ref{ch:vqcd_application_rl_vqsd}, we discuss a quantum state diagonalization-based quantum channel certification method. The method primarily depends on the Choi-Jamio\l{}kowski isomorphism and variational quantum fidelity estimation~\cite{cerezo2020variational} algorithm. As the variational quantum state diagonalization is an indispensable part of our quantum channel certification process, we discuss the possible application of RL-VQSD. This includes how it can improve the channel certification method and make it currently available quantum device friendly by deciding the number of gates and depth of diagonalizing quantum circuits. In Chapter~\ref{ch:discussions}, we summarise the thesis. Finally, in Chapter~\ref{ch:conclusions}, we evaluate the merits and drawbacks of our methodology and conclude the results in the thesis. Furthermore, we delve into potential avenues for future exploration and development.




\chapter{Preliminaries}\label{ch:preliminaries}


This chapter introduces the basic concepts and notation used in the thesis. We start the chapter by introducing variational quantum algorithms (VQAs). This introduction includes a detailed discussion of various components of VQAs relevant to understanding the later concepts of the thesis. For example, we briefly define the construction of the \textit{cost function} and its importance in finding the optimal solution for a problem. This discussion is followed by elaborating on various \textit{ansatz} (also known as a variational quantum circuit) constructions. Through this discussion, we will see that the number of two-qubit gates and depth for the problem-inspired ansatz increases exponentially with the number of qubits. Meanwhile, for the problem agnostic ansatz, the constrained connectivity and the elevation in the number of parameters with increasing problem size cause trainability issues.

In the following section, we provide a detailed discussion of the basics of reinforcement learning with proper examples. For the sake of the thesis, we primarily focus on model-free learning methods. The model-free learning discussion includes Q-learning, the deep and double-deep Q-network with proper code blocks for implementation purposes. In the same section, we stress the discussion of the exploration and exploitation scenario.

We conclude the chapter by introducing the basic concepts and infrastructure of our RL-based quantum architecture search algorithm.

We encourage the reader to consult {Appendix~\ref{app:intro-quantum-computing}} for more details on quantum states, gates, and channels.

\section{Variational quantum algorithms}\label{sec:intro_vqa}


Over the past few decades, multiple scientific fields have collaborated in the exploration and advancement of quantum algorithms and their experimental implementation. However, while numerous quantum algorithms were initially proposed, the majority require millions of physical qubits to operate on quantum hardware. Unfortunately, contemporary quantum hardware is only capable of accommodating a few hundred physical qubits, which are classified as noisy intermediate-scale quantum (NISQ)~\cite{bharti2022noisy,preskill2018quantum} devices. 
\begin{figure}
	\includegraphics[width=\linewidth]{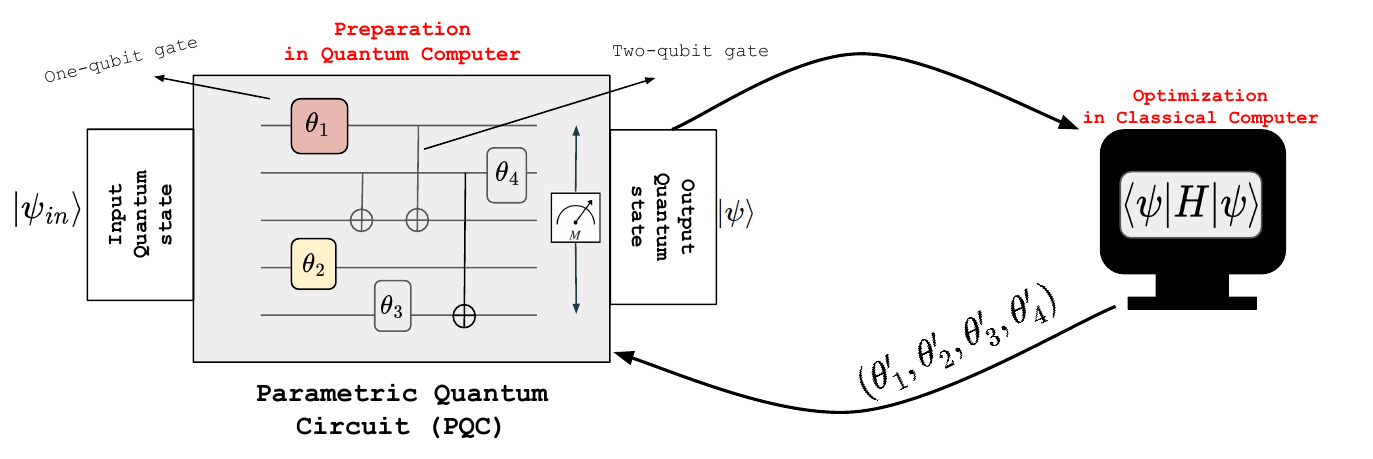}
	\caption{{The workflow of variational quantum algorithms (VQAs). The process starts with an input quantum state, which is processed through a series of one- and two-qubit gates within the parametric quantum circuit (PQC). The output is the evolved quantum state from the PQC. This contrasts with a classical subroutine that optimizes the PQC parameters using classical optimization methods to minimize a cost function encoding a problem. The optimized parameters are then fed back to update the PQC iteratively until the problem is solved}.}
	\label{fig:VQAs}
\end{figure}

In light of this, variational quantum algorithms (VQAs)~\cite{mcclean2016theory,cerezo2021variational}, illustrated in Figure~\ref{fig:VQAs}, represent a specific class of NISQ algorithms that can run on these devices, as they were specifically designed with such limitations in mind. VQAs utilize quantum hardware to run a parametric quantum circuit (PQC). {The PQC consists of a series of quantum gates whose operations depend on a set of tunable parameters, typically involving rotations and entangling gates like controlled-Z (\texttt{CZ}) or controlled-NOT (\texttt{CNOT} or \texttt{CX}) gates}. {The PQC generates a quantum state whose parameters are optimized using a classical optimization method}. If not optimized up to an expected threshold, the parameters are fed back to the PQC. This quantum-classical hybrid process helps us keep the depth of the quantum circuit low and mitigate the noise. The two basic components of a VQA are (1) PQC, which is famously known as an ansatz, and (2) the cost or loss or objective function that encodes the problem.

In the remaining subsection, we briefly describe the two basic aspects of VQA.

In the subsection \ref{sec:cost_function_introduced}, we briefly introduce the cost function by mentioning the main criterion one should consider while choosing a cost function. Then, in section~\ref{sec:ansatz_introduced}, we define the problem-inspired and agnostic structure of ansatzes and their advantages and disadvantages.

\subsection{Cost function}\label{sec:cost_function_introduced}
In classical machine learning (ML) methods, a cost function is introduced to evaluate the model's performance. The main objective of a model, then, is to determine the optimal set of model parameters, which indicates a global minimum of the cost function. Hence, the cost function is sometimes called an objective function. For the optimization procedure, either gradient-based or gradient-free optimization algorithms are frequently used.

VQAs serve as quantum analogues of machine learning techniques. A crucial operational aspect of these algorithms depends on the ability to encode problems effectively into cost functions. 
From a geometric perspective, a cost function is defined by a hyper-surface over the parameters—a cost landscape—whereas the optimizer's role involves traversing this landscape -- as depicted in~\figref{fig:cost_landscape} -- to locate the global minima. The~\figref{fig:cost_landscape} illustrates the optimization process within VQAs, highlighting the quest to navigate this parameterized space to achieve optimal solutions efficiently.

\begin{figure}[h]
	\centering   \includegraphics{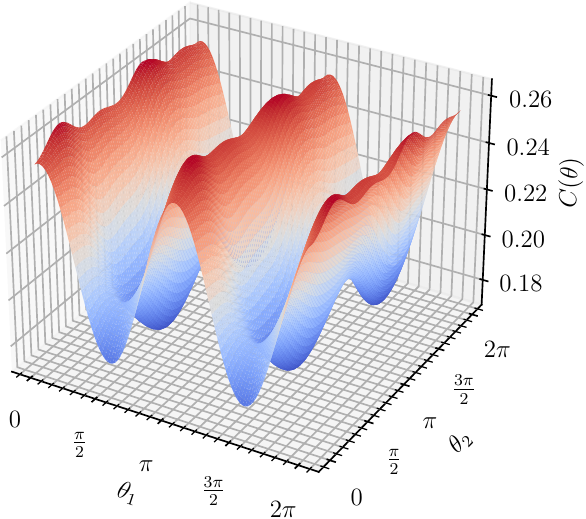}
	\caption{An illustration of the cost function landscape for the variational quantum state diagonalization algorithm~\cite{larose2019variational} with a two-qubit mixed quantum state. As an ansatz, we choose a layer of rotation \texttt{RY} and \texttt{RZ} on both the qubits, followed by a \texttt{CNOT} with control on the first qubit.}
	\label{fig:cost_landscape}
\end{figure}

One can express the cost function in the form
\begin{equation}
	C(\vec{\theta}) = \sum_j c_j\left( \textrm{Tr}\left( O_j\rho_j^\prime(\Vec{\theta}) \right)\right),
	\label{eq:general_cost_function}
\end{equation}
{where $\{c_j\}$ are set of functions. The choice of $\{c_j\}$ solely depends on the task in hand}. $\rho_j^\prime(\Vec{\theta}) = U( \vec{\theta})\rho_jU^\dagger(\vec{\theta})$, where $U(\vec{\theta})$ is the PQC with discrete or continuous vector of parameters $\vec{\theta}$. $\{\rho_j\}$ {represent the input quantum states, initialized based on the specific problem. For the variational quantum eigensolver (VQE), $\rho_j=|0\rangle^{\otimes n}$, where $n$ is the number of qubits. In contrast, for variational quantum state diagonalization (VQSD), $\rho_j$ is the state that needs to be diagonalized.} {$\{O_j\}$ is the set of observable,  which, for example, can be defined by a chemical Hamiltonian when using the VQE to find the ground state of a molecule}. An ideal cost function should follow a list of criteria discussed below.

\paragraph{Faithfulness}

If we have a problem under consideration, we need to formulate a cost function to solve it using a variational method. The \textit{faithfulness} of the cost function implies that its minimum must correspond to the solution of the problem. If we consider a problem $p$ and we define the cost function corresponds to the problem $C(\vec{\theta})$ then the solution to the problem ($p^*$) is given by
\begin{equation}
	p^* = \min_{\theta_j} C(\theta_j),
\end{equation}
if $C(\vec{\theta})$ is faithful cost function.
\paragraph{Efficiency}
From the term efficiency of a cost function, we indicate that one must be able to estimate it by performing measurements on a quantum device.
To maintain the validity of the quantum advantage, it is important to devise a cost function that proves computationally challenging for classical computers to compute.


For an example, for a quantum state, $\rho$, $\textrm{Tr}(\rho^2)$ defines the \textit{purity} of the state. Minimization of purity is a very useful primitive method to solve a range of problems that are relevant to quantum physics problems such as quantum state rank estimation~\cite{o2016efficient}, quantum state learning~\cite{lee2018learning, chen2021variational}, quantum device certification~\cite{kundu2022variational} and many more. It is also well-known that a quantum computer can find the purity of an $n$-qubit state with complexity scaling linearly in $n$, which gives exponential speed-up over classical computers~\cite{buhrman2001quantum,gottesman2001quantum}. So, a cost function of the form 
\begin{equation}
	C = \sum_jc_j(\textrm{Tr}(\rho_j^2)),
\end{equation}
is efficiently computable and provides a quantum advantage.

\paragraph{Trainability} Through trainability, we emphasize the fact that the cost function must be trainable, i.e. it should be possible to efficiently optimize the parameters of the cost function $\vec{\theta}$. One of the main hindrances of trainability is the occurrence of barren plateaus with increasing depth and number of qubits. Technically, a barren plateau is defined by the exponentially vanishing average partial derivatives of the cost function with the size of a quantum system~\cite{mcclean2018barren}. 
\begin{figure}[h]
	\centering
	\includegraphics[scale=0.8]{ 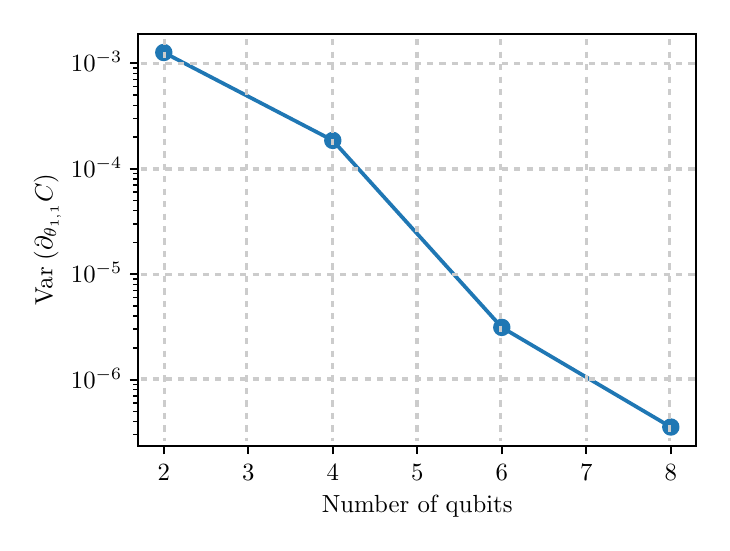}
	\caption{Exponential decay of the variance of the cost function gradient for quantum state diagonalization problem for three-qubit of depth 2 ansatz.}
	\label{fig:var-grad-diagonalization}
\end{figure}
This makes the landscape of the cost function essentially flat, so it requires exponentially enhanced precision to tackle the finite sampling noise and to determine the direction of the global minima. This is irrespective of the fact that one uses a gradient-based~\cite{cerezo2021higher} or gradient-free optimization method~\cite{arrasmith2021effect}.
In a recent work~\cite{cerezo2021cost}, the authors claim that the issue of barren plateau could be tackled by carefully reconstructing the cost function and by making it local. They propose two variants of cost functions in~\eqref{eq:general_cost_function} as
\begin{equation}
	C^\textrm{global}(\vec{\theta}) = \sum_j c_j \textrm{Tr}\left( O_j^\textrm{global}\rho_j^\prime(\Vec{\theta}) \right), \; \textrm{where} \;\; O_j^\textrm{global} =  \id_j - \ket{0}\bra{0}_j,
	\label{eq:general_local_cost}
\end{equation}
or as
\begin{equation}
	C^\textrm{local}(\vec{\theta}) = \sum_j c_j \textrm{Tr}\left( O_j^\textrm{local}\rho_j^\prime(\Vec{\theta}) \right), \; \textrm{where} \;\; O_j^\textrm{local} =  \id - \dfrac{1}{n}\sum_{j^\prime=1}^n \left(|0\rangle\langle0|_j\otimes\id_{\Bar{j^\prime}} \right).
	\label{eq:general_global_cost}
\end{equation}
Here, $\id$ is the identity operation and $\id_{\bar{j^\prime}}$ defines $\id$ over all qubits except $j^\prime$. Irrespective of the PQC, the variance of the gradient for the local cost function in~\eqref{eq:general_local_cost} at worst vanishes polynomially, which makes it trainable up to a depth of order $\mathcal{O}(\textrm{log}n)$, which is not the case for the global cost function~\eqref{eq:general_global_cost}, whose variance of gradient decays exponentially.

\subsection{Ansätze}\label{sec:ansatz_introduced}
In the previous section, we stress the point that VQAs are analogous to classical ML methods. One of the fundamental aspects of the ML model is the neural network (NN). The quantum version of NN is a PQC that is famously known as an ansatz.
Inspired by the success of classical NNs, some ansatz architectures such as quantum convolutional NN~\cite{cong2019quantum}, recurrent quantum NN~\cite{bausch2020recurrent}, quantum long short-term memory~\cite{chen2022quantum}, and quantum graph NN~\cite{verdon2019quantum} has been introduced.

As shown in~\eqref{eq:general_cost_function}, the ansatz $U(\Vec{\theta})$ contains trainable parameters $\Vec{\theta}$, these can be trained to minimize the cost function. Now, the $U(\Vec{\theta})$ does not have a specific structure, but it depends on the problem under consideration, for example, unitary coupled cluster (UCC), quantum alternating operator, variational Hamiltonian ansatz. These types of ansatz fall under the group \textit{problem inspired ansatz}. One can generically express $U(\vec{\theta})$ as a product of $L$ succeeding unitaries where one part of the unitary is parametrized by parameters $\vec{\theta}_j$ and another part is non-parametrized i.e.
\begin{equation}
	U(\Vec{\theta}) = \prod_{l=1}^L V(\Vec{\theta}_l)W_l,
	\label{eq:ansatz-main-equation}
\end{equation}
where $V(\Vec{\theta}_l)$ is the parametrized part of the ansatz and is of the form $\exp(-i\Vec{\theta}_lH_l)$ and $H_l$ is a Hermitian operator and $W_l$ is the non-parametrized unitary. Each layer $l$ contains a vector of parameters $\Vec{\theta}$. Depending on the problem under consideration, the parametrized and the non-parametrized part takes different forms. This leads to a class of sophisticated structures of ansatz that we briefly discuss in the following.

\paragraph{Unitary coupled cluster (UCC)} The Unitary coupled cluster is a \textit{problem-inspired ansatz} that is widely utilized in quantum chemistry problems. In this case, the problem statement is to find the ground state energy of a molecule represented through a fermionic Hamiltonian $H$. 

According to the Born-Oppenheimer approximation~\cite{born1927oppen}, one can describe the interaction of a system of electrons with its nucleus in a second quantized form where single-particle orbitals can either be filled or empty. And any interaction between electrons can be represented using annihilation ($\hat{a}$) and creation ($\hat{a}^\dagger$) following an anti-commutation relationship. Hence, a non-relativistic molecule Hamiltonian can be written in the form
\begin{equation}
	H_\textrm{mol} = H_\textrm{nuc} + \underbrace{\sum_{pq}h_{pq}\hat{a}_p^\dagger\hat{a}_q}_{\textrm{single excitations}} + \underbrace{\frac{1}{2}\sum_{pqrs}h_{pqrs}\hat{a}_p^\dagger\hat{a}_q^\dagger\hat{a}_r\hat{a}_s}_{\textrm{double excitations}}+\ldots
	\label{eq:fermionic_hamiltonian}
\end{equation}

To obtain UCC ansatz, we replace the traditional Hamiltonian cluster operator in~\eqref{eq:fermionic_hamiltonian} in terms of coupled cluster theory with its anti-Hermitian as follows~\cite{taube2006new,peruzzo2014variational}
\begin{align}
	&\ket{\psi}_\textrm{UCC} = U_\textrm{cc}\ket{\psi} = e^{T_1 + T_2+\ldots}\ket{\psi},\label{eq:ucc_ansatz}\nonumber\\
	&T_1 = \sum_{\substack{v\in \textrm{vacant},\\ o\in\textrm{occupied}}}\theta_{vo}\left(\hat{a}_v^\dagger\hat{a}_o - \hat{a}_o^\dagger\hat{a}_v\right)\\
	&T_2 = \sum_{\substack{v,v'\in \textrm{vacant},\\ o,o'\in\textrm{occupied}}}\theta_{vv'oo'}\left(\hat{a}_v^\dagger\hat{a}_{v'}^\dagger\hat{a}_o\hat{a}_{o'} - \hat{a}_o^\dagger\hat{a}_{o'}^\dagger\hat{a}_v\hat{a}_{v'}\right),\nonumber\\
	&\vdots\nonumber
\end{align}
where $\ket{\psi}$ is an uncorrelated reference state, usually a Hartree-Fock state. The ansatz shown in~\eqref{eq:ucc_ansatz} is the UCC ansatz. One can obtain a popular variant of UCC -- which is UCCSD, where SD stands for single and double  -- just by considering $T_1$ and $T_2$ in the coupled cluster representation. To implement the ansatz in a quantum computer, one can use the fermion to spin mappings such as Jordan-Wigner, Parity, and Bravyi-Kiteav transformations~\cite{tranter2018comparison,steudtner2018fermion}.
\begin{figure}[H]
	\centering  \includegraphics[scale=0.8]{ 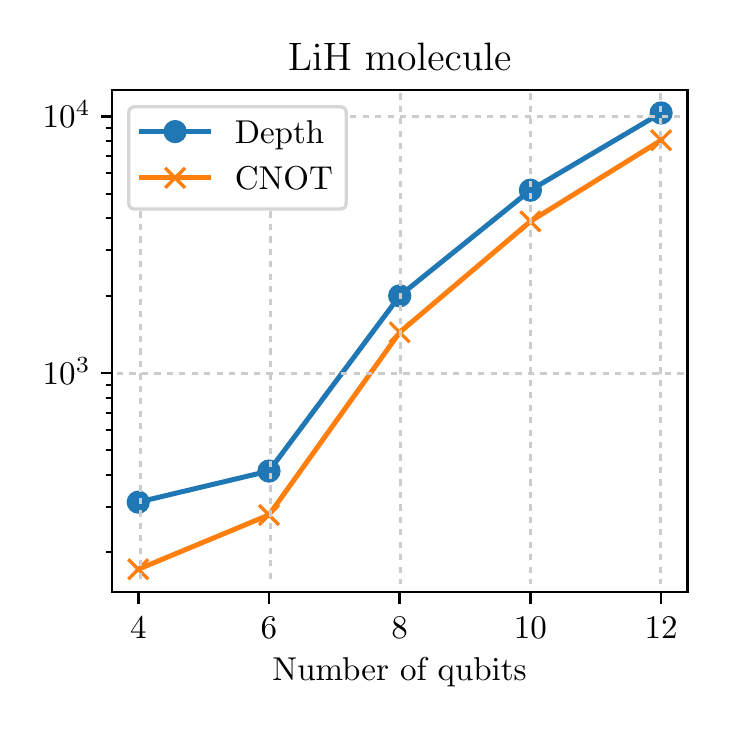}
	\caption{Illustration of the variation of the number of \texttt{CNOT} gates and \textit{depth of the circuit} with increasing spin orbitals in LiH molecule with UCCSD ansatz.}\label{fig:LiH_resource_requirement}
\end{figure}

Even though the unitarity of UCCSD suggests ease of implementation on quantum hardware, the current gate-based quantum computing asks for a decomposition in terms of one and two-qubit gates. However, the number of one and two-qubit gates and the depth of the circuit proliferates with the number of qubits as shown in~\figref{fig:LiH_resource_requirement}.

One can utilize the Suzuki-Trotter approximation of the $T_i$ operators to deal with the issue, keeping in mind the correct operator ordering of trotterized UCC ansatz~\cite{grimsley2019trotterized}.

\begin{figure}[ht!]
	\centering
	\includegraphics{ 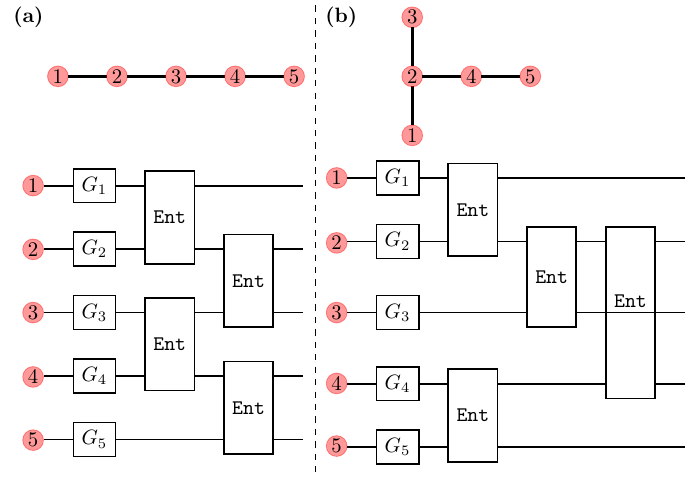}
	\caption{An example of HEA depending on the topology of \texttt{IBM} quantum devices. In (a), we present HEA that follows the topology of \texttt{ibmq\_manila} whereas in (b) the HEA follows the topology of \texttt{ibm\_quito}, \texttt{ibm\_belem} and \texttt{ibm\_lima}. It should be noted that the $G_i$ are parametrized, and \texttt{Ent} is the entangling unitary. The quantum wire through the last \texttt{Ent} gate means the gate does not apply on that qubit.}
	\label{fig:hea}
\end{figure}

\paragraph{Hardware efficient} 
As the name suggests, the hardware efficient ansatz (HEA) aims to reduce the circuit depth and gate count in $U(\Vec{\theta})$ and is efficiently implementable in
currently available quantum hardware. The generic form of HEA follows the form
\begin{equation}
	U_\textrm{HEA} = \prod_{q=1}^N\prod_{d=D}^1 \left(G^{q,d}(\Vec{\theta})\times U_\textrm{Ent}\right).
	\label{eq:hea-equation}
\end{equation}
Here $q$ is the number of qubits up to $N$ and $d$ is the depth of the ansatz up to $D$. The product on the number of qubits is over parameterized gates $G(\Vec{\theta})$. $U_\textrm{Ent}$ defines entangled gates.

One determines the $U_\textrm{Ent}$ and $G(\Vec{\theta})$ from a predefined quantum gate set, and the placement of the gates is determined by the topology of real quantum hardware. This, in turn, helps avoid circuit depth overhead while translating an arbitrary ansatz into a sequence of gates, as shown in~\figref{fig:hea}. This makes HEA very applicable to Hamiltonian, which has interactions similar to the quantum hardware~\cite{kokail2019self}. One of the primary advantages of HEA is that while implemented, it can incorporate encoding symmetries~\cite{gard2020efficient,otten2019noise} and depth reduction~\cite{tkachenko2021correlation}. 
\begin{figure}[ht!]
	\centering
	\includegraphics[width = 0.5\linewidth]{ 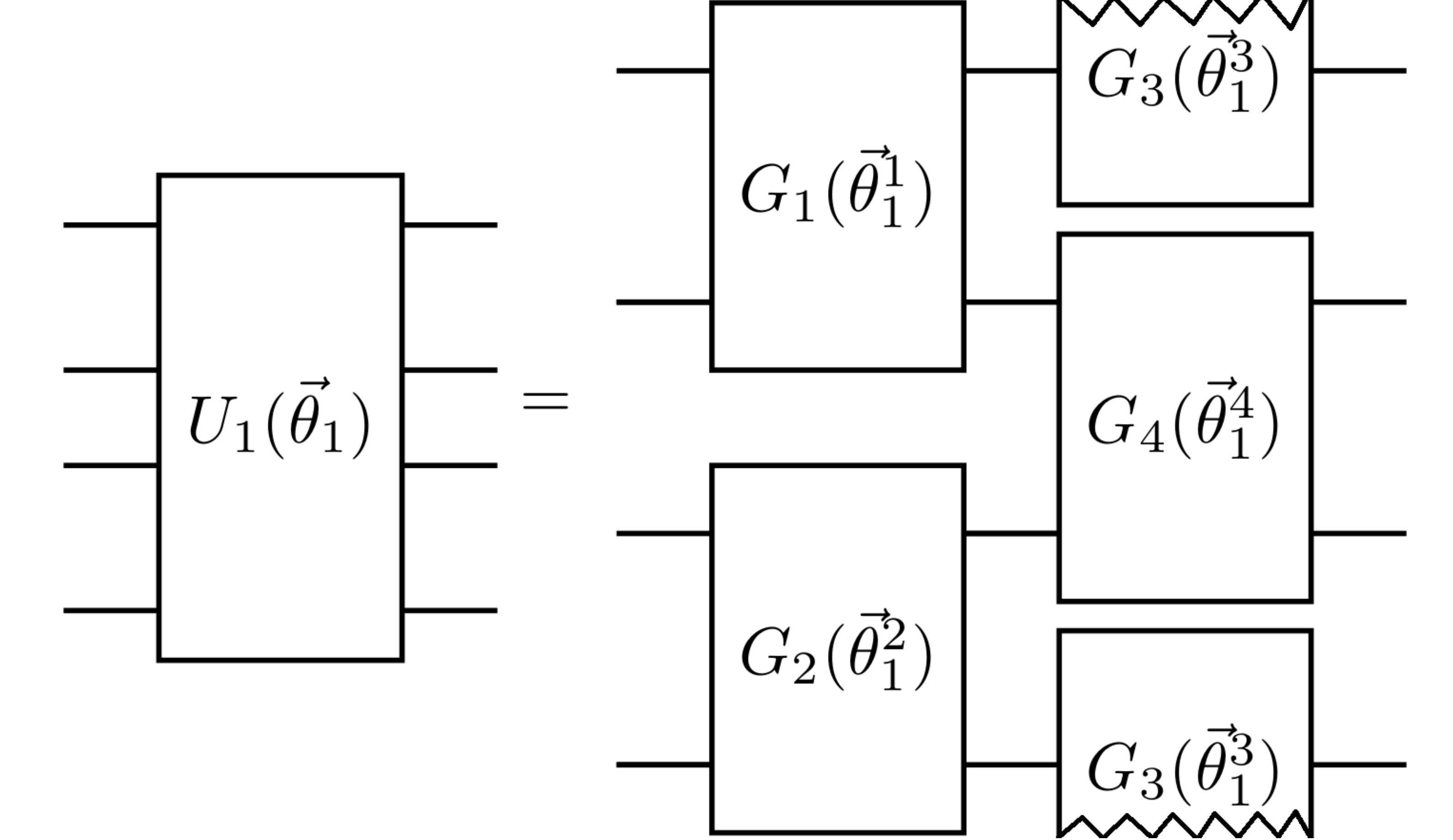}
	\caption{Structure of a layered ansatz, where the ansatz $U_l(\vec{\theta})$ is decomposed into layer-wise unitaries $U_l(\vec{\theta}_l)$ for $l = 1,2,\ldots, l$. Each $U_l(\vec{\theta}_l)$ is further decomposed into two-qubit rotations for $\vec{\theta}_i^j$ the $i$ denotes the layer and $j$ is the parameter count.}
	\label{fig:layered-ansatz}
\end{figure}
A well-known variant of HEA is layered HEA (LHEA), where gates act on alternating pairs of qubits in a brick-like structure as shown in~\figref{fig:layered-ansatz}. But this has trainability issues. In~\cite{leone2022practical}, it has been shown that the trainability of HEA depends on the entanglement in input states, i.e. HEAs are trainable if input states follow the area law of entanglement~\cite{eisert2010colloquium} whereas it becomes untrainable if the input follows volume law of entanglement~\cite{bianchi2022volume}.

\paragraph{Quantum approximate operator}
The quantum approximate optimization algorithm (QAOA) ansatz was first introduced in~\cite{farhi2014quantum} to obtain solutions for combinatorial optimization problems. The ansatz provides an approximation to a Hamiltonian $H$ by constructing a specific variational ansatz through first-order Suzuki-Trotter decomposition approximating the adiabatic evolution. The operators $\exp\left(\ii \beta_j H_\textrm{mix}\right)$ and $\exp\left(\ii \gamma_j H_\textrm{obj}\right)$ are applied in alternating manner, resulting in the ansatz
\begin{equation}
	U_\textrm{QAOA} = \prod_{j=1}^l\exp\left(\ii \beta_j H_\textrm{mix}\right)\exp\left(\ii \gamma_j H_\textrm{obj}\right),
	\label{eq:QAOA-ansatz-equation}
\end{equation}
where $H_\textrm{mix} = -\sum_i\sigma_x^i$, $\sigma_x$ is the Pauli $X$ operator. The computational power and reachability of the QAOA ansatz are rigorously discussed in~\cite{morales2020universality,lloyd2018quantum,hastings2019classical}.

While it's true that breaking down~\eqref{eq:QAOA-ansatz-equation} into native gates may result in a long circuit, this is often due to the presence of many-body terms in $H$ and limited device connectivity. However, one notable advantage of this ansatz is that for certain problems, the feasible subspace is smaller than the full Hilbert space. This restriction can possibly lead to better algorithmic performance, making this approach highly effective in certain scenarios.

\paragraph{Variable structure}
Although the constancy in the structure of the ansatz while solving a problem enables control over the overall ansatz complexity, it fails to harness the refinement that is attained by optimizing the circuit. The refinement can be realized through the addition or removal of unnecessary quantum operators. This novel approach was first introduced in ADAPT-VQE~\cite{grimsley2019adaptive}, which adaptively inserts a quantum fermionic operator to provide a desired level of accuracy while maintaining a minimal number of operators. The operator is chosen in such a way that it affects the minimization of energy the most. The operators are chosen from a pool of fermionic operators, and it has been shown that the ADAPT-VQE substantially outperforms UCCSD in terms of both the number of variational parameters and accuracy. Following this trail in~\cite{tang2021qubit}, the authors introduce a hardware-efficient variant of ADAPT-VQE, i.e. qubit ADAPT-VQE, which drastically reduces the operator pool. Other more efficient variants of ADAPT-VQE are introduced in~\cite{yordanov2021qubit}.

Another way of variable ansatz construction is introduced in~\cite{larose2019variational}, where the ansatz is allowed to grow. And if the algorithm cannot minimize the cost function for a specified number of iterations, then one adds an identity gate spanned by new variational parameters that are randomly added to the ansatz.


\section{Basics of reinforcement learning}\label{sec:reinforcement_learning_introdced}
One of the most prominent sub-fields of machine learning is reinforcement learning (RL), which is concerned with how an \textit{agent} can learn to make decisions in an \textit{environment} to maximize a cumulative \textit{reward function}. The primary goal of the \textit{agent} is to learn a \textit{policy}, which a mapping from the states of the environment to a decision (i.e. an \textit{action}), that in turn maximizes cumulative reward.
\begin{figure}[ht]
	\centering
	\includegraphics[width = 0.8\textwidth]{ 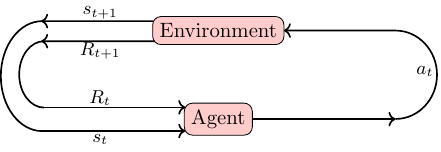}
	\caption{Illustration of the agent and the environment interaction.}
	\label{fig:agent-environment-interaction}
\end{figure}

RL is intrinsically distinct from \textit{supervised learning}~\cite{cunningham2008supervised,hastie2009overview,learned2014introduction} because, in supervised learning, the learning process is conducted from a training set of labeled exemplary data, which is in turn provided by a knowledgeable external supervisor. Each exemplary data corresponds to a situation, and the label denotes the action the system should take in that situation. In RL, the agent's learning process is mediated through interaction with the environment. 

Meanwhile, RL differs from \textit{unsupervised learning}~\cite{barlow1989unsupervised,dayan1999unsupervised,ghahramani2003unsupervised} in the sense that in USL, the main ambition is to find structure hidden in connections of unlabeled data. Whereas in RL all one tries to achieve is the maximization of the reward function signal instead of trying to find a hidden structure. In the following, we will briefly specify RL problems in terms of optimal control of the Markov decision process (MDP) and the challenges in RL.

The summary of notations and definitions used in this thesis in the context of reinforcement learning is provided in Table~\ref{tab:rl-table-of-notations}.


\begin{table}
	\begin{tabular}{c|c|c}
		Name   & Notation & Definition \\
		\hline
		State  & $s$    & \makecell{Refers to the current situation of an agent \\ in an environment, which includes all the \\ relevant information necessary to make \\ decisions about what action to take next} \\
		\hline
		Action & $a$    & \makecell{Refers to the choices an agent can make \\ based on its current state to interact \\ with the environment. The agent \\ chooses the action by using a policy.}  \\
		\hline
		Policy & $\pi$  & \makecell{ Refers to a mapping from perceived states \\ of the environment to actions to be taken \\ when in those states. }  \\
		\hline
		Reward & $R$    &  \makecell{Refers to the goal in a problem. On each \\ time step, the environment sends to the \\ agent a single number, a reward. The \\ agent’s sole objective is to maximize the \\ total reward it receives over the long run.}  \\
		\hline
		Value Function  &  $v$  & \makecell{It refers to the value of a state as the total \\ amount of reward an agent can expect to \\ accumulate over the future, starting from \\ that state. }  \\
		\hline
		Transition probability & $p[s'|s,a]$ & \makecell{It refers to the probability of transition \\ to state $s'$ if action $a$ is taken on \\ the state of the environment $s$.}
	\end{tabular}
	\caption{The necessary notations and their definitions utilized to provide an introduction to RL and in the upcoming sections. Additionally we define $\mathcal{S}$ for the set of states, and $\mathcal{A}$ for the set of actions.}
	\label{tab:rl-table-of-notations}
\end{table}

\subsection{Finite Markov decision process}

In this section, we discuss the mathematical form of reinforcement learning problems, which includes key elements of the problems' structure, such as the value function and Bellman equation.

\paragraph{Environment-agent interaction}
In RL, the \textit{agent}, which is the learner and decision-maker, interacts with everything outside it called the \textit{environment}. 

The whole RL process can be characterized by time steps $t = 0,1,\ldots, T$ where at each time step, the agent receives the state of the environment $s_t\in \mathcal{S}$ where $\mathcal{S}$ is the set of all possible states. After interacting with the environment, the agent gives out an action $a_t\in \mathcal{A}(s_t)$, where $\mathcal{A}(s_t)$ defines the space of all actions available in the state $s_t$. In the next time step, the environment gives out (1) a new state: $s_t\rightarrow s_{t+1}$ and (2) a reward $R_{t+1}\in \mathcal{R}\in \mathbb{R}$, which is depicted in the~\figref{fig:agent-environment-interaction}.

At each time step $t$, the agent forms a mapping between the states to the probabilities of selecting each possible action. This mapping is termed as \textit{Policy} ($\pi_t$), where $\pi_t(a|s)$ tells us the probability of selecting the action $a_t = a$ if the environment is at state $s_t=s$.

The above-mentioned framework can be used for many different problems. As a fundamental principle in the study of reinforcement learning, we adhere to the notion that an agent's environment comprises all elements that are beyond its arbitrary control, forming a context in which the agent operates. The boundary between the agent and the environment the limit of the agent's \textit{absolute control}, not its knowledge. For example, the agent might know almost everything about the interacting environment, but still, it faces difficulty in solving the task just as we know how Rubik's cube works and still might not be able to solve it.

At its core, the RL framework a powerful abstraction of the challenge of learning to achieve goals through interaction with an environment. Any problem with learning goal-directed behaviour can be efficiently reduced to three sequences of processes and repeating back and forth. (1) the choice, that is made by the agent, which we call \textit{action} (2) the basis upon which choices are determined i.e. the \textit{states} and (3) the definition of the agent's goal in the form of \textit{reward}. In the next, we briefly describe the \textit{rewards} with a simple example.

\paragraph{Rewards and returns}
At each time step $t$, the reward is defined by a real number $R_t\in \mathbb{R}$ which the agent receives from the environment as a signal. The main purpose of the agent is to maximize the cumulative reward in the long run. Formulation of a goal corresponding to a problem in terms of reward is a hard problem. Based on a problem one can define either a (1) sparse or (2) dense reward, where the sparse reward is easier to formulate than the dense one. 

For example in the case of a robot in a maze problem, the goal is to teach the robot to escape from the maze. If we formulate the goal in terms of a dense reward where for each time step $t$, for each movement of the robot in the maze, the agent receives a reward signal $R_t=-1$. This encourages the robot to find a solution to escape the maze quickly as each step of the robot is penalized ($-1$) and the robot will focus on minimizing the penalty.

However, the same problem can be formulated using a sparse reward where the agent receives a reward when it escapes the maze, but it slows down learning because the agent needs to take many actions before getting any reward. In problems like chess or backgammon, the only way to formulate the reward is by giving the player a big reward if wins. This is known as the \textit{credit assignment problem}.

If we consider that after each time step the agent receives a reward $R_{t+1}, R_{t+2},\ldots$ then we seek to maximize the expected return $G_t$. The $G_t$ is defined by
\begin{equation}
	G_t = R_{t+1}+R_{t+2}+R_{t+3}+\ldots+R_T,
	\label{eq:return}
\end{equation}
where $T$ is the final time step. Considering the notion of the final time step, the agent–
environment interaction can be broken into subsequences, which we call \textit{episodes}. An episode might be represented by a game of chess or a complete trip in the maze by a robot. Meanwhile, in many problems, it might not be possible to break the interaction between the agent and the environment, which can not be broken into identifiable episodes but goes on continually without limit. For these tasks, the definition of return in~\eqref{eq:return} is not valid, and it can be infinite.

An important additional concept to consider while formulating a reward function is called \textit{discounting}. With this approach, the agent aims to choose actions that maximize the total of discounted rewards it will receive in the future as follows
\begin{equation}
	G_t = R_{t+1}+\gamma R_{t+2}+\gamma^2 R_{t+3}+\ldots = \sum_{j = 0}^\infty \gamma^j R_{t+j+1},
	\label{eq:discounted_return}
\end{equation}
where $\gamma$ is the discount rate parameter that ranges as $0\leq \gamma \leq 1$. It determines the present value of the future rewards, i.e. a reward received at $j$-th time in the future is worth only $\gamma^{j-1}$ times what it would be
worth if it were received immediately.

If $\gamma<1$, the infinite sum has a finite
value as long as the reward sequence $\{R_j\}$ is bounded. If $\gamma$ is close to 0, the agent is \textit{myopic} i.e. the agent is concerned only with maximizing immediate rewards. If $\gamma$ is close to 1, the agent is \textit{far-sighted}. In the following, we provide an elaborate example of how a \textit{myopic} and \textit{far-sighted} agent impacts the cumulative return.
\begin{figure}[ht]
	\centering
	\includegraphics[width=0.5\textwidth]{ 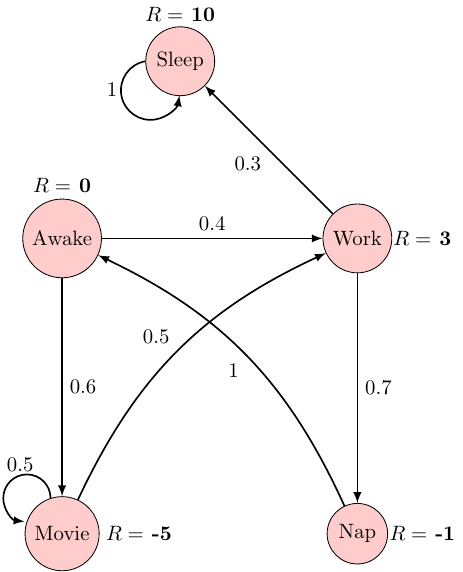}
	\caption{A toy Markov decision process (MDP) of the everyday schedule of a PhD student.}
	\label{fig:simple_mdp}
\end{figure}

\paragraph{Toy Example} Before we move on further, let us take a simple example and run through the process of formulating the $G$ step by step. In~\figref{fig:simple_mdp}, we give an illustration of a work-sleep schedule of a PhD student as an MDP. In the problem, we consider \textit{sleep} as the terminal step. There are many possible ways to reach the terminal step, such as 
\begin{itemize}
	\item$\text{Awake}\rightarrow\text{Work}\rightarrow\text{Sleep}$ 
	\item$\text{Awake}\rightarrow\text{Work}\rightarrow\text{Nap}\rightarrow\text{Awake}\rightarrow\text{Movie}\rightarrow\text{Work}\rightarrow\text{Sleep}$ etc.
\end{itemize}

The transition matrix corresponding to the MDP is defined by
\begin{equation}
	p[s'=s_{t+1}|s_t] = 
	\begin{tabular}{l|lllll}
		& Awake                  & Movie                    & Nap                      & Sleep                    & Work                     \\ \cline{2-6} 
		\multicolumn{1}{l|}{Awake} & \multicolumn{1}{l}{}  & \multicolumn{1}{l}{0.6} & \multicolumn{1}{l}{}    & \multicolumn{1}{l}{}    & \multicolumn{1}{l}{0.4} \\[0.4cm] 
		\multicolumn{1}{l|}{Movie} & \multicolumn{1}{l}{}  & \multicolumn{1}{l}{0.5} & \multicolumn{1}{l}{}    & \multicolumn{1}{l}{}    & \multicolumn{1}{l}{0.5} \\[0.4cm] 
		\multicolumn{1}{l|}{Nap}   & \multicolumn{1}{l}{1} & \multicolumn{1}{l}{}    & \multicolumn{1}{l}{}    & \multicolumn{1}{l}{}    & \multicolumn{1}{l}{}    \\[0.4cm] 
		\multicolumn{1}{l|}{Sleep} & \multicolumn{1}{l}{}  & \multicolumn{1}{l}{}    & \multicolumn{1}{l}{}    & \multicolumn{1}{l}{1}   & \multicolumn{1}{l}{}    \\[0.4cm] 
		\multicolumn{1}{l|}{Work}  & \multicolumn{1}{l}{}  & \multicolumn{1}{l}{}    & \multicolumn{1}{l}{0.7} & \multicolumn{1}{l}{0.3} & \multicolumn{1}{l}{}    \\[0.4cm] 
	\end{tabular}.
\end{equation}
As all the states are associated with a particular reward, we can calculate the discounted return using~\eqref{eq:discounted_return}. For an example for the path: $\text{Awake}\rightarrow\text{Movie}\rightarrow\text{Work}\rightarrow\text{Sleep}$, if we are currently at the state ``Awake" then the discounted reward for the "Awake" state becomes
\begin{equation}
	G_\text{Awake} = R_\text{Movie}+\gamma R_\text{Work} + \gamma^2 R_\text{Sleep} = -5 + 3\gamma + 10\gamma^2,
\end{equation}
\begin{itemize}
	\item \textbf{Myopic agent:} If we consider $\gamma$ is close to $0$ say $\gamma=0.1$ then we get $G_\text{Awake}=-4.6,$
	\item \textbf{Far-sighted agent:} On the other hand, if we consider $\gamma=0.9$, which is close to $1$ then we get $G_\text{Awake}=+5.8$.
\end{itemize}
From this, we can say that for the \textit{far-sighted case}, the agent might prefer to take the route $\text{Awake}\rightarrow\text{Movie}\rightarrow\text{Work}\rightarrow\text{Sleep}$. But a \textit{myopic agent} is more physical because animals have a preference for immediate reward. In the following part, we elaborate on Markov decision processes (MDPs) and the importance of value function in RL.

\paragraph{Markov decision process}
The Markov decision processes (MDPs) provide a modelling framework for sequential decision problems~\cite{littman1994markov}, as shown in the example of the previous section. MDPs have something called \textit{Markov property}, which is defined by the dependency of the next state of the environment on the current state and the action. To express it more elaborately, we consider an environment that might respond at the time $t+1$ corresponding to the action taken at $t$. In the most general case, the nature at $t+1$ may depend on everything that happened in the past events from $t=0$ to $t=t$ i.e. the probability of the state appearing on state $s' = s_{t+1}$ with reward $r = R_{t+1}$ can be written as
\begin{equation}
	\textrm{P}\left[R_{t+1}=r, s_{t+1} = s'|s_0,a_0,R_1, \ldots s_{t-1}, a_{t-1}, R_t, s_t, a_t\right].
\end{equation}
On the contrary, if the state has \textit{Markov property} then the same probability can be rewritten as
\begin{equation}
	p(r, s'|s_t,a_t)=\textrm{P}\left[R_{t+1}=r, s_{t+1} = s'| s_t, a_t\right],
	\label{eq:mdp-dynamics}
\end{equation}
which says that \textit{the future state of the system depends solely on the current state and the action taken, rather than on any prior history}~\cite{howard1960dynamic}. This means instead of memorizing all the information about all the past events and defining the nature of the environment at $t+1$ the agent now can characterize $t+1$ by just inquiring about its previous event at time $t$.

One can define an MDP as a 5-tuple process ($\mathcal{S}$, $\mathcal{A}$, $p(s'|s_t,a_t)$, $R$, $\gamma$) depending on:
\begin{itemize}
	\item $\mathcal{S}$, the finite set of all states of the environment;
	\item $\mathcal{A}$, the finite set of all legal actions that can be executed at state $s_t\in \mathcal{S}$;
	\item $p[s_{t+1}=s'|s_t,a_t]$, the probability of transition to state $s'$ at the time $t+1$ if $a_t$ action is taken on the state of the environment $s_t$;
	\item $R\in\mathcal{R}$, the reward received when the environment transitioning from $s_t$ to $s'$ after action $a_t$;
	\item $\gamma\in\left[0,1\right]$, is the discount factor that two-qubit the difference in the future and present reward.
\end{itemize}

If the dynamics are specific by~\eqref{eq:mdp-dynamics}, one can compute the expected reward $r(s, a)$ as 
\begin{equation}
	\sum_{r\in \mathcal{R}}r \sum_{s'\in \mathcal{S}}p(r,s'|s_t,a_t),
\end{equation}
from the state-action pair. The state transition probability $p(s^\prime|s_t,a_t)$ is expressed as 
\begin{equation}
	p(s^\prime|s_t,a_t)= \sum_{r\in\mathcal{R}}p(r,s'|s_t,a_t),
\end{equation}
and the expected rewards for the state–action–next-state $r(s, a, s^\prime)$ as 
\begin{equation}
	\sum_{r\in\mathcal{R}}\dfrac{rp(r,s'|s_t,a_t)}{p(s'|s_t, a_t)}.
\end{equation}


\paragraph{Value function}
In RL the value function provides a measure of how good it is for an RL-agent to take a specific action at a particular state. It helps the agent to make decisions in an uncertain environment by quantifying the expected future reward that an agent can get for that particular state. The value functions are defined with respect to specific policies. Recalling that a policy, $\pi$, at time step $t$, is defined by a mapping from each state $s_t\in\mathcal{S}$ and action $a_t\in\mathcal{A}$ to the probability $\pi(s_t,a_t)$ of taking the action $a_t$ in state $s_t$. This inevitably gives rise to two main types of value functions in RL based on the state and the action as follows:
\begin{itemize}
	\item \textbf{State-value function}, $v_\pi(s)$ is represented by the expected return starting from a particular state $s$ and by following a particular policy $\pi$. The state value function encapsulates the inherent value associated with each $s$, irrespective of the action taken. For MDPs, one can define the state-value function as
	\begin{equation}
		v_\pi(s_t) = \mathbb{E}_\pi\left[ \sum_{j=0}^\infty \gamma^j R_{t+j+1} \Bigg\vert s_t=s \right],
		\label{eq:state_value_func}
	\end{equation}
	where $\mathbb{E}_\pi[.]$ gives the expected value given that the agent follows policy $\pi$.
	
	\item \textbf{Action-value function}, $q_\pi(s, a)$ is represented by the expected return starting from a particular state $s$, taking an action $a\in\mathcal{A}$ and by following a particular policy $\pi$. It estimates the long-term value of taking a particular action in a given state. For MDPs, one can define the action-value function as
	\begin{equation}
		q_\pi(s_t,a_t) = \mathbb{E}_\pi\left[ \sum_{j=0}^\infty \gamma^j R_{t+j+1} \Bigg\vert s_t=s, a_t=a \right].
		\label{eq:action_value_func}
	\end{equation}
\end{itemize}

\noindent The value function in~\eqref{eq:state_value_func} can be decomposed into two parts: (1) the immediate reward and (2) the discounted value of the successor rate in the following way:
\begin{align}
	v_\pi(s_t) & = \mathbb{E}_\pi\left[ R_{t+1} + \gamma R_{t+2} +\ldots \Bigg\vert s_t=s \right]\nonumber\\
	& = \mathbb{E}_\pi\left[ R_{t+1} + \gamma \left(R_{t+2} + \gamma R_{t+3}+\ldots\right) \Bigg\vert s_t=s \right]\nonumber\\
	& = \mathbb{E}_\pi\left[ \underbrace{R_{t+1}}_{\textrm{\shortstack{Immediate\\[-1ex] reward}}} + \underbrace{\gamma\sum_{j=0}^\infty\gamma^k R_{t+j+2}}_{\textrm{\shortstack{Discounted\\[-0.5ex] successor rate}}}\Bigg\vert s_t=s \right]\nonumber\\
	& = \sum_{a}\pi(a|s)\sum_{s'=s_{t+1}}\sum_{r=R_{t+1}}p(s',r|s,a)\left[ r+\gamma\mathbb{E}_\pi\left( \sum_{j=0}^\infty R_{t+k+2}\Bigg\vert s' \right) \right]
	\nonumber\\
	& = \sum_{a}\pi(a|s)\sum_{s',r}p_{ss'}\left[ r + \gamma v_\pi(s') \right].
	\label{eq:bellman_equation_state_val_func}
\end{align}
The~\eqref{eq:bellman_equation_state_val_func} is known as the \textit{Bellman expectation equation for the state-value function}. This is basically a sum of all values of the action, the current and the successor state of the environment. For each value of the variables, we compute the $\pi(s,a)p(s',r|s,a)$, multiply it by the linear weighted term $\left[ r + \gamma v_\pi(s') \right]$ and sum over all possible values of the three variables to get expected value. In the same way using~\eqref{eq:action_value_func} we can find the \textit{Bellman equation for the action-value function} as follows:
\begin{equation}
	q_\pi(s_t,a_t)= \mathbb{E}_\pi\left[ R_{t+1} + 
	\gamma v(s_{t+1},a_{t+1})\Bigg\vert s_t=s, a_t=a \right].
	\label{eq:bellman_equation_action_val_func}
\end{equation}
Here we define the \textit{Bellman operator} $T^\pi$ as 
\begin{equation}
	(T^\pi q)(s_t, a_t) = R_{t+1} +\gamma\mathbb{E}\left[ q(s_{t+1},a_{t+1}) \Big\vert s_{t+1}, a_{t+1} \right],
\end{equation}
where $s_{t+1}$ is samples with probability $p(s_{t+1}|s_t,a_t)$ and $a_{t+1}$ is sampled from the policy $\pi(a|s)$. In the same manner, the \textit{Bellman optimality operator} is defined by
\begin{equation}
	(T^* q)(s_t, a_t) = R_{t+1}+\gamma\mathbb{E}\left[ q^*(s_{t+1},a_{t+1}) \Big\vert s_{t+1}\sim p(s_{t+1} | s_t,a_t) \right].\label{eq:bellman_optimality_operator}
\end{equation}
The operator defined in~\eqref{eq:bellman_optimality_operator} will be useful while we discuss the Deep Q-Network (DQN).  

There are many possible ways to approximate, compute, and learn the value function for a particular policy, and all these methods are diagrammatically represented using \textit{backup diagrams}. These diagrams visualize how the value function is updated based on new information received from the environment. In other words, backup diagrams depict the flow of information and the update process involved in processing the value function. This two-qubit how the estimated values of state or state-action pair are modified after receiving feedback from the environment. In each \textit{backup diagram} the states (\tikz\draw[draw] (0,0) circle (.5ex);) and actions (\tikz\draw[draw, fill = black] (0,0) circle (.5ex);) are represented by vertices and the transaction among them are edges $\left(\begin{tikzpicture}
	\path (0,0) node[draw, circle, scale=0.6] (A) {};
	\path (0.4,0.4) node[circle, fill = black, scale=0.6] (B) {};
	\draw (A) -- (B);
\end{tikzpicture}\; \text{or} \; \begin{tikzpicture}
	\path (0,0) node[circle, fill = black, scale=0.6] (A) {};
	\path (.4,.4) node[draw, circle, scale=0.6] (B) {};
	\draw (A) -- (B);
\end{tikzpicture}\right) 
$

For example, in~\figref{fig:backup_diagrams} we illustrate the \textit{backup diagrams} for $v_\pi$ and $q_\pi$.

\begin{figure}
	\centering
	\includegraphics[width = 0.8\textwidth]{ 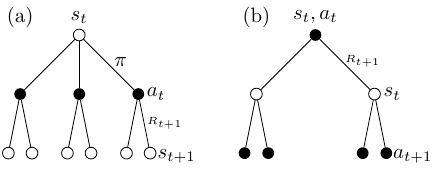}
	\caption{Simple illustration of the backup diagrams for (a) state-value function ($v_\pi$) and (b) action-value function ($q_\pi$). The time flows downwards.}
	\label{fig:backup_diagrams}
\end{figure}

\subsection{Optimal value function} 
The optimal value function is defined by the maximum expected return that an agent can achieve by following an optimal policy. A policy $\pi$ is defined to be better or equivalent to another policy $\pi^\prime$ i.e. $\pi\geq\pi^\prime$ iff $v_\pi(s_t)\geq v_{\pi^\prime}(s_t) \;\forall \; s_t\in \mathcal{S}$. There can be more than one \textit{optimal policy}, which we denote by $\pi^*$. These policies share the same state-value function, which is called by \textit{optimal state-value function}, $v^*$ and is defined by
\begin{equation}
	v^*(s_t) = \underset{\pi}{\textrm{max}}\; v_\pi(s_t),
\end{equation}
and in the same manner the \textit{optimal action-value function}, $q^*$, is given by
\begin{equation}
	v^*(s_t) = \underset{\pi}{\textrm{max}}\; q_\pi(s_t, a_t),
\end{equation}
where $a_t\in \mathcal{A}(s_t)$. Using~\eqref{eq:bellman_equation_action_val_func} we can rewrite the above equation in terms of $v^*$ as
\begin{equation}
	q_*(s_t,a_t) = \mathbb{E}
	\left[ R_{t+1} +\gamma v^*(s_{t+1}, a_{t+1})\Bigg| s_t=s, a_t=a \right].
\end{equation}

As we saw in the previous section, one can obtain \textit{Bellman equation} directly from either the state-value function or the action-value function. In the same way, from the optimal state-value and action-value function, we get \textit{Bellman optimality equation}. The intuition behind the Bellman optimality equation lies in the principle of optimality, which states that an optimal policy can be divided into sub-problems and solved independently for each state. In other words, it decomposes the problem of finding the optimal value function into sub-problems for each state, allowing for a dynamic programming approach to solving MDPs. The Bellman optimality equation for the state-value function is given by
\begin{align}
	v_*(s_t) & = \underset{a_t\in\mathcal{A}}{\textrm{max}} \; \mathbb{E}_{\pi^*}\left[ R_{t+1} + \gamma R_{t+2} +\ldots \Bigg\vert s_t=s, a_t=a \right]\nonumber\\
	& = \underset{a_t\in\mathcal{A}}{\textrm{max}} \; \mathbb{E}_{\pi^*}\left[ R_{t+1} + \gamma \left(R_{t+2} + \gamma R_{t+3}+\ldots\right) \Bigg\vert s_t=s, a_t=a \right]\nonumber\\
	& = \underset{a_t\in\mathcal{A}}{\textrm{max}} \; \mathbb{E}_{\pi^*}\left[ R_{t+1} + \gamma\sum_{j=0}^\infty\gamma^k R_{t+j+2}\Bigg\vert s_t=s, a_t=a \right]\nonumber\\
	& = \underset{a_t\in\mathcal{A}}{\textrm{max}} \; \mathbb{E} \left[ R_{t+1} +\gamma v_*(s') \Bigg\vert s_t=s, a_t=a \right]\label{eq:bellman_optimality_equation_state_val_func1}\\
	& =\underset{a_t\in\mathcal{A}}{\textrm{max}} \sum_{s',r}p(s',r|s,a)\left( r+\gamma v_*(s') \right).
	\label{eq:bellman_optimality_equation_state_val_func2}
\end{align}


In the \eqref{eq:bellman_optimality_equation_state_val_func1} and \eqref{eq:bellman_optimality_equation_state_val_func2} is the Bellman optimality equation for $v_*$. Following the same steps, we can obtain the Bellman optimality equation for $q_*$ as
\begin{align}
	q_*(s_t,a_t) & = \underset{a_t\in\mathcal{A}}{\textrm{max}} \; \mathbb{E} \left[ R_{t+1} +\gamma v_*(s') \Bigg\vert s_t=s, a_t=a \right] ,\label{eq:bellman_optimality_equation_state_val_func3}\\
	& =\underset{a_t\in\mathcal{A}}{\textrm{max}} \sum_{s',r}p(s',r|s,a)\left( r+\gamma v_*(s') \right). \label{eq:bellman_optimality_equation_state_val_func4}
\end{align}

Having the (optimal) value functions at our disposal, we can now talk about how to solve a reinforcement learning problem. There are fundamentally two approaches used in solving reinforcement learning (RL) problems: (1) Model-Based RL~\cite{polydoros2017survey,moerland2023model,kaiser2019model} and (2) Model-Free RL. In the case of Model-Based RL, we have complete knowledge of the dynamics of the Markov decision process, which includes precise knowledge of the transition probabilities and rewards for each transition. Hence, to solve the Model-Free RL problem, \textit{value iteration algorithm}~\cite{alpaydin2020introduction,bellman1966dynamic}can be utilized. It iteratively evaluates the optimal value of each state by taking the expected intermediate reward and the values for the successor states into account. For the sake of coherence regarding the thesis topic, we end the discussion on the \textit{Model-Based RL} here. In the following, we briefly discuss the \textit{Model-Free reinforcement learning} approach.
\begin{figure}
	\centering \includegraphics{ 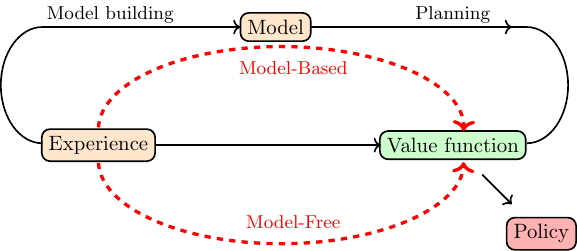}
	\caption{Illustration of model-based and model-free RL. }
	\label{fig:model_based_vs_model_free}
\end{figure}


\subsection{Model-Free learning}\label{sec:model_free_rl}
As previously discussed, the \textit{value iteration algorithm} requires complete knowledge of the model, such as transition probabilities and rewards. In many scenarios, we encounter situations where the agent lacks knowledge of transition probabilities, necessitating alternative approaches to compute the policy function. The model-free approach is developed by keeping these encounters in mind. We compare the value- and policy-based approaches in {Tab.~\ref{tab:value-and-policy-based}}. Meanwhile, a comparison of the model-based and model-free reinforcement learning is illustrated in~\figref{fig:model_based_vs_model_free}. Model-based reinforcement learning involves a subroutine of creating an internal representation of how the environment behaves, allowing the agent to predict future states. Meanwhile, instead of learning from a predefined model, model-free reinforcement learning focuses on learning the best actions in different situations by estimating the value or policy directly from observed experiences. For the sake of this thesis, in the remaining sections, we will focus on value-based model-free algorithms.
\begin{center}
	\begin{table}
		\begin{tabular}{c|c}
			Value-based & Policy-Based   \\
			\hline
			\makecell{ Aims to learn the optimal\\ 
				state-value function or\\ 
				action-value function, i.e.\\ 
				Q-function directly. } & 
			\makecell{ Aims to learn the optimal\\
				policy directly, without\\
				explicitly estimating value\\ functions.} \\ 
			\hline
			\makecell{ These algorithms depends on\\ 
				the balance between exploration\\
				and exploitation. It explores\\ 
				new states and actions while\\
				also exploiting the value\\
				estimates.} & \makecell{
				These algorithms optimize the\\
				policy $\pi(s_t,a_t|\vec{\theta})$ by modifying\\
				the $\vec{\theta}$ to maximize the\\ expected cumulative reward.} \\ 
			\hline
			\makecell{ A popular value-based algorithm\\
				that learns the Q-value through\\
				an iterative updating method\\
				is called~\textit{Q-learning}~\cite{watkins1989learning,watkins1992q}.\\
				This method updates the value\\
				based on observed rewards and\\
				transition probability.} & \makecell{
				In a policy-based algorithm\\
				gradient ascent method~\cite{zinkevich2003online,sorg2010reward}\\ 
				is used to update policy\\ 
				parameters.}
		\end{tabular}
		\caption{Comparison table for value-based and policy-based methods.}
		\label{tab:value-and-policy-based}
	\end{table}
\end{center}

\paragraph{Temporal difference learning} In the value iteration process, the state-value function is calculated recursively using the~\eqref{eq:bellman_equation_state_val_func}. In model-free learning, the transition function is replaced by a sequence of the sample from the environment, and we can use Bellman's recursive computation to estimate the new updates to the value function based on the previous estimates.

\textit{Temporal difference learning} (TD learning)~\cite{sutton1988learning} is a bootstrapping method that can be used to process and refine the samples to achieve an approximate final value of the state. As the name suggests, it refers to the difference in the values of the states at two different time steps. This information is then used to calculate the value at the new time step. TD learning method is given by
\begin{equation}
	v(s_t) \leftarrow v(s_t) + \alpha\left[ R_{t+1} + \gamma v(s_{t+1}) -v(s_t) \right],
	\label{eq:temp_diff_v1}
\end{equation}
recalling $s_t$ is the state of the environment at time step $t$, and $s_{t+1}$ is the new state at the time $t+1$. The reward we receive for the transition from $s_t$ to $s_{t+1}$ is $R_{t+1}$. As we described before, $\gamma$ is the discount factor, $\gamma$ set closer to zero, two-qubit a myopic agent, and if it is closer to $1$, we get a far-sighted agent. The new variable $\alpha$ is the learning rate, which determines how fast the algorithm learns. The essence of temporal difference lies in the last term of~\eqref{eq:temp_diff_v1}, where we subtract $``-v(s_t)"$ from the current state value to compute the temporal difference. The TD learning method revolutionized the use of model-free approaches in different RL scenarios. One remarkable achievement was TD-Gammon~\cite{tesauro1995temporal}, a program that defeated human world champions in the game of Backgammon during the early 90s.

\begin{figure}[H]
	\centering
	\includegraphics[width=0.4\textwidth,angle=90]{ 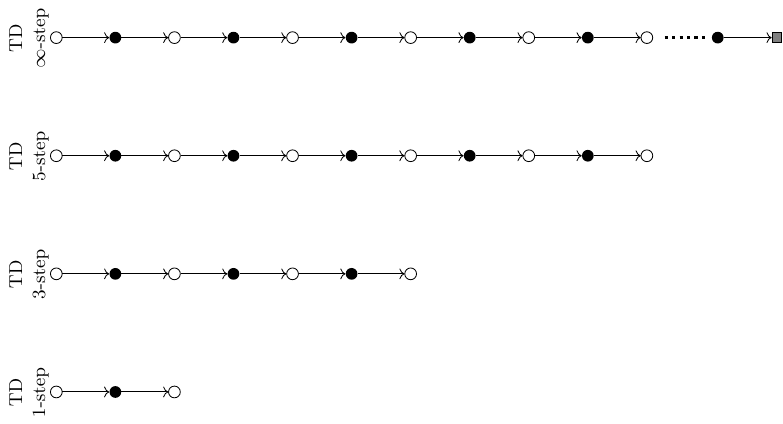}
	\caption{Illustration of $1$, $5$ and $\infty$-step TD learning. The $\infty$ step learning is equivalent to the Monte Carlo learning method.}
	\label{fig:TDinfinity}
\end{figure}

Unlike the Monte Carlo learning method, which performs a full episode with various random actions before it utilizes the reward, the temporal difference bootstraps the Q-function with the help of the values that it gathers from the previous time steps. This helps to refine the value function with the cumulative rewards after each time step (1-step is defined by, \begin{tikzpicture}
	\path (0,0) node[draw, circle, scale=0.6] (A) {};
	\path (1,0) node[circle, fill = black, scale=0.6] (B) {};
	\path (2,0) node[draw, circle, scale=0.6] (C) {};
	\draw[->] (A) -- (B);
	\draw[->] (B) -- (C);
\end{tikzpicture}, the arrow dented the flow of time). Based on this description, we can think of a middle ground with \textit{n-steps}, which fundamentally points to the fact that we neither sample a single step like TD learning nor do we sample a full episode like Monte Carlo, but we sample a few steps (say, $n$ steps) at a time before utilizing the reward values, which we illustrate this in~\figref{fig:TDinfinity}. This strategy allows for a more granular assessment of state values compared to TD learning, which may lead to faster convergence and reduced variance in the estimated values. Moreover, one of the advantages of the n-step approach is that it does not require the entire episode to be completed before updating, making the n-step TD learning computationally efficient over Monte Carlo methods, particularly in environments with long episodes. This middle-ground approach can thus be particularly advantageous in RL tasks where balancing the trade-offs between bias and variance is crucial for efficient learning and decision-making.

\paragraph{Exploration and exploitation}

\begin{figure}[H]
	\centering
	\includegraphics[height = 0.6\textheight]{ 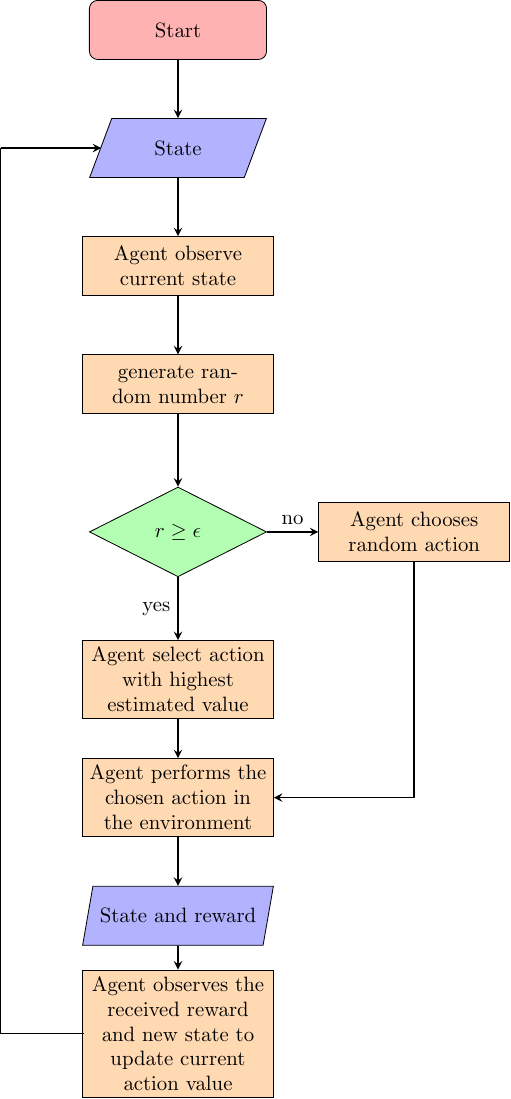}
	\caption{Block diagram of an $\epsilon$-greedy algorithm. In this algorithm, a parameter $\epsilon$ is used to determine the probability of choosing a random exploration action versus selecting the action with the highest expected reward. When $\epsilon$ is set to a small value, the algorithm tends to exploit the current best action, while a higher value encourages more exploration.}
	\label{fig:epsilon-greedy}
\end{figure}

In the context of the model-free RL exploration/exploitation are two primary concepts that are most frequently used in designing RL algorithms. The \textit{exploration} part indicates the process of collecting environment information of the agent by taking random actions. The goal of this routine is to discover new states of the environment, actions, and their corresponding rewards. Meanwhile, \textit{exploitation} is utilized to leverage the information that is gathered by maximizing the immediate rewards, and it focuses on the selection of actions that will provide the highest reward based on the existing knowledge of the environment. Hence, the \textit{exploration} and \textit{exploitation} complement each other.

One of the fundamental challenges in RL is to balance exploitation and exploitation because if the agent only focuses on exploration it never gets the chance to leverage the existing information whereas more exploitation of the immediate knowledge and the algorithm may get stuck in a non-optimal policy. This means the algorithm has failed to find a better strategy that leads to an optimal policy.

\paragraph{$\epsilon$-greedy exploration}

It is a well-known strategy in RL to maintain a balance between the exploration and exploitation aspect. The $\epsilon$-greedy method works by assigning a parameter $\epsilon$, which encodes the information regarding the probability of exploration as 
\begin{equation}
\textrm{Selection of action by agent} = \epsilon a_\textrm{r} + (1-\epsilon) a_\textrm{ev},\label{eq:epsilon-greedy}
\end{equation}
where $a_\textrm{r}$ is the random action and $a_\textrm{ev}$ is the action with highest estimate value. This algorithmic choice presented in~\eqref{eq:epsilon-greedy} is called the \textit{exploration/exploitation trade-off}. An illustration of the $epsilon$-greedy algorithm is presented in~\figref{fig:epsilon-greedy}.
There are other methods to define the exploration aspect, either by using the Thompson sampling~\cite{russo2017tutorial,sickles2019measurement} or adding the Dirichlet-noise~\cite{kotz2004continuous}.

\subsection{On-policy and off-policy learning}

The RL deals with learning an action and a policy from the received rewards. That is, the agent selects an action to perform on the environment and learns from the reward that it receives after taking the action. The agent learns to select an action for the next time-step. Now, the dilemma arises if the agent either updates from its most recent action (\textit{on-policy}) or learns from all the available information gathered (\textit{off-policy}). A tabular representation of the difference between the two kinds of learning is presented in {Tab.~\ref{tab:on-and-off-policy}}.

It should be noted that the on- and off-policy learning shows different behaviour in convergence to the optimal policy. Such as the off-policy method promises to converge to the optimal policy after sampling a sufficient number of states, so it uses a greedy reward approach. At the same time, the on-policy method for a fixed value of $\epsilon$ can not converge to the optimal policy as they keep on selecting the actions that are based on the current action. Meanwhile, when we use a policy where we keep on varying $\epsilon$ towards zero, then the on-policy method promises to converge.

In the following, we will briefly discuss two famous algorithms: (1) on-policy SARSA and (2) off-policy Q-learning.

\begin{table}[ht]
\begin{center}
	\begin{tabular}{c|c}
		On-policy & Off-policy   \\
		\hline
		\makecell{ The current policy determines\\
			the action to take and the\\
			value of that action is used\\
			to update the policy\\
			function} & 
		\makecell{ The determination of action\\
			takes place by backing up\\
			values of other action which\\
			is not necessarily selected\\
			by the behavior policy} \\ 
		\hline
		\makecell{ This learning process is stable\\
			since the agent updates the\\
			policy based on experiences\\
			from the current policy, making\\
			the learning process more\\
			consistent} & \makecell{
			Off-policy learning although less\\ stable compared to on-policy\\
			learning but here, the agent is\\
			open to learning from a diverse\\
			set of experiences, which has\\ 
			the potential to improve the\\
			exploration and discover\\
			better policies.} \\
		\hline
		\makecell{ The agent's exploration might be\\
			biased, which causes difficulty\\
			in exploring actions that are\\
			not properly defined on policy,\\ leading to sub-optimal policy.} & \makecell{
			It requires more data and\\ computational resources\\
			compared to on-policy\\
			learning as it learns\\ 
			from all the available \\
			information gathered\\ 
			before.}
	\end{tabular}
	\caption{A comparison of on- and off-policy learning.}
	\label{tab:on-and-off-policy}
\end{center}
\end{table}


\paragraph{SARSA} Is an on-policy algorithm that was first proposed by Rummery and Niranjan~\cite{rummery1994line}. The on-policy methods update the current policy by utilizing the action value of the policy itself. Hence, the update rule is given by
\begin{equation}
Q(s_t, a_t) \leftarrow Q(s_t, a_t) + \alpha\left[ r_{t+1} + \gamma Q(s_{t+1},a_{t+1}) - Q(s_t,a_t) \right].
\label{eq:sarsa}
\end{equation}
It should be noted that SARSA follows the same updating method as the TD method as presented through~\eqref{eq:temp_diff_v1}. The only difference is that the state-action value function replaces the state-value function.

SARSA updates the Q-values of the current state-action pair ($s_t,a_t$) by using the Q-values of the next state-action pair ($s_{t+1}, a_{t+1}$). To say it in a more elaborate manner, in the SARSA algorithm, we select an action and apply it to the environment, and then it follows the action that is guided by the behaviour policy. The behaviour policy is defined by either the $\epsilon$-greedy approach. Hence, in this learning process, the state space sampling is done by following the behavioural policy and updating the current policy by updating the values of the actions based on the sampling. 

\paragraph{Q-learning}
The Q-learning algorithm was first proposed by Watkins~\cite{cjc1989learning}. Unlike SARSA, Q-learning learns the Q-values using a different policy than the one it followed during the action selection process. The updated formula for Q-learning is given by 
\begin{equation}
Q(s_t, a_t) \leftarrow Q(s_t, a_t) + \alpha\left[ r_{t+1} + \gamma\max_a Q(s_{t+1},a) - Q(s_t,a_t) \right].
\label{eq:q-learning}
\end{equation}
The primary difference between SARSA and Q-learning update rule is the term $Q(s_{t+1},a_{t+1})$ is replaced by $\max_a Q(s_{t+1},a)$. This indicates that it now learns from the stored values of the best action instead of the action that was actually evaluated. The main reason Q-learning is called off-policy is that it updates the Q-values using the Q-value of the next state $s_{t+1}$ and a greedy action which is not necessarily the action of the behaviour policy.

One of the drawbacks of Q-learning is that as the size of the Q-table grows exponentially with the increase in the number of states and actions, it becomes infeasible. To tackle this hindrance, we can combine Q-learning with the deep neural network, which is known as a DQN~\cite{montavon2018methods}.

\paragraph{Deep Q-Network} In this algorithm, a deep neural network in the form $Q_\beta:\mathcal{S}\times \mathcal{A}\rightarrow\mathbb{R}$ is utilized to approximate the optimal Q function $Q^*$, where $\beta$ is the parameter of the neural network. There are two essential tricks to be noted to achieve success with DQN, and those are~\cite{mnih2015human}:
\begin{itemize}
\item The usage of \textit{experience reply}~\cite{lin1992self} in DQN helps to obtain uncorrelated samples since the trajectory of MDP has a strong temporal correlation. Specifically, for each time step $t$, we tend to store the tuple $(s_t, a_t, r_t, s_{t+1})$ into the reply memory, which is followed by the sampling of a mini-batch of independent samples from the replay memory. This is used to train the deep neural network using a stochastic gradient descent method.

\item The utilization of a target network $Q_{\theta^*}$ with parameter $\theta^*$ is another trick which we use in DQN. After collecting the independent samples $\left( s_i, a_i,r_i,s_{i+1} \right)$ from the replay memory, to update the parameters of the Q-network, we evaluate the target network as follows:
\begin{equation}
	E_\textrm{targ}^i = r_i + \gamma\times \textrm{max}_{a\in\mathcal{A}} Q_{\theta^*}(s_{i+1}, a),
\end{equation}
and then updating $\theta$ using the gradient of
\begin{equation}
	L(\theta) = \frac{1}{n}\sum_{i=1}^n \left[ E_\textrm{targ} - Q_\theta\left(s_i,a_i \right) \right]^2.\label{eq:network-param-update-target}
\end{equation}
Parameter $\theta^*$ is updated once in every $T$ i.e. the target network is held fixed for $T$ steps, and then we update it by the current weight of the~Q-network.
\end{itemize}

In~\cite{mnih2015human} the authors utilize a replay memory of size $10^{6}$ on the other hand, in~\cite{ostaszewski2021reinforcement} the authors use a replay memory of $2\times10^{4}$. This indicates that the replay memory size is usually very large. Additionally, in DQN, we use $\epsilon$-greedy policy to enable exploration over the state and action. In this scenario, when the replay memory is large, the experience replay is prone to sample independent transitions from an exploration-driven policy (as it is $\epsilon$-greedy), which decreases the variance of the term $\Delta L(\theta)$. As the $L(\theta)$ plays a vital role in updating the parameters of a neural network, so having a low variance in $L(\theta)$ stabilizes the training of DQN.

Meanwhile, to understand the necessity of a target network, we first set $\theta^* = \theta$. The bias-variance decomposition gives us the expected value of the $L(\theta)$ as
\begin{equation}
\mathbb{E}[L(\theta)] =\underbrace{ || Q_\theta - T^*Q_\theta||^2_\sigma}_{\text{\shortstack{mean-squared\\[-0.5ex] Bellman equation}}} + \underbrace{\mathbb{E}[E_\textrm{targ}^1 - (TQ_\theta) (s_1,a_1)]^2}_{\text{\shortstack{$\theta$ dependent \\[-0.5ex] variance of\\[-0.5ex] $E_\textrm{targ}^1$}}}\label{eq:network-param-update-wo-target},
\end{equation}
where $T^*$ is the \textit{Bellman optimality operator} defined in~\eqref{eq:bellman_optimality_operator}. In~\eqref{eq:network-param-update-wo-target} along with the mean-squared Bellman error, we have an additional bias of the form $\mathbb{E}[E_\textrm{targ}^1 - (TQ_\theta) (s_1,a_1)]^2$, that is strictly dependent on the neural network parameter. This indicates that minimizing the $L(\theta)$ can drastically differ from minimizing the mean-squared Bellman equation. This issue is by using the target network in~\eqref{eq:network-param-update-target} what has an expectation value
\begin{equation}
\mathbb{E}[L(\theta)] =\underbrace{ || Q_\theta - T^*Q_\theta||^2_\sigma}_{\text{\shortstack{mean-squared\\[-0.5ex] Bellman equation}}} + \underbrace{\mathbb{E}[E_\textrm{targ}^1 - (TQ_{\theta^*}) (s_1,a_1)]^2}_{\text{\shortstack{$\theta$ independent \\[-0.5ex] variance of\\[-0.5ex] $E_\textrm{targ}^1$}}},
\end{equation}
as the second term is independent of $\theta$ so minimizing $L(\theta)$ is approximately equivalent to solving
\begin{equation}
\underset{\theta \in \Theta}{\textrm{minimize}}|| Q_\theta - T^*Q_\theta||^2_\sigma,\label{eq:DQN-as-minimization-problem}
\end{equation}
where $\Theta$ is the parameter space. In simple words, the ultimate aim of DQN is to solve the minimization problem defined through~\eqref{eq:DQN-as-minimization-problem} with a fixed $\theta^*$ and it updates the $\theta^*$ by the minimizer parameter of the neural network $\theta$.
\begin{algorithm}
\caption{FQI algorithm}\label{alg:fqi-pseudocode}
\begin{algorithmic}
	\State MDP as tuple $(\mathcal{S}, \mathcal{A}, p, \mathcal{R}, \gamma)$, define $\mathcal{F}$, sampling distribution $\sigma$, total iterations $J$, sample number $n$, initial estimator $\tilde{Q}_0$.\Comment{Input}
	\For {$j=0,\ldots,J-1$ }
	
	\textit{Sample} \textit{i.i.d.} $\{(s_j\in \mathcal{S}, a_j\in\mathcal{A},R_j\in\mathcal{R},s_{j+1})\}_{i\in[n]}$, with $(s_i, a_i)$ sampled from $\sigma$.
	
	\textit{Compute} $E_\textrm{targ}^j = R_j+\gamma\times\underset{a\in\mathcal{A}}{\text{max}}\tilde{Q}_j(s_{i+1},a)$
	
	\textit{Update} the action-value function
	\begin{equation}
		\tilde{Q}_{k+1}\leftarrow \underset{f\in\mathcal{F}}{\textrm{argmin}}\frac{1}{n}\sum_{i=1}^n\left[ E^i_\textrm{targ}-f(s_i,a_i) \right]^2\nonumber
	\end{equation}
	\EndFor\\
	\textit{Define} the policy $\pi_J$ as a greedy policy in respect with $\tilde{Q}_J$
	\State An estimator $\tilde{Q}_J$ of $Q^*$ and policy $\pi_J$.\Comment{Output}
\end{algorithmic}
\end{algorithm}

An implementable version of the above-discussed DQN is known as the neural Fitted Q-Iteration (FQI) algorithm. This generates a sequence of value functions. Let us consider $\mathcal{F}$ is a class of functions on the state-action space. For the $j$-th iteration of the algorithm, we consider $\tilde{Q}_j$ is the present estimate of the $Q^*$. Hence, the update rule of $\tilde{Q}_j$ is defined as
\begin{equation}
\tilde{Q}_{k+1} = \underset{f\in\mathcal{F}}{\textrm{argmin}}\frac{1}{n}\sum_{i=1}^n\left[ E^i_\textrm{targ}-f(s_i,a_i) \right]^2.
\end{equation}
We can replace the $\mathcal{F}$ by the neural networks, and then the algorithm is known as neural FQI~\cite{riedmiller2005neural}. Therefore, we can consider neural FQI as a variation of DQN, where we substitute experience replay with sampling from a stable distribution to understand the statistical characteristics.

\paragraph{Double Deep Q-Network}

{Deep RL methods employ neural networks to adapt the agent's policy for optimizing the return
$$
G_t=\sum_{k=0}^{\infty} \gamma^{k} r_{t + k + 1},
$$
with the discount factor $\gamma \in[0,1)$.
Each state and action pair $(s, a)$ can then be assigned an action-value that quantifies the expected return from state $s$ in step $t$ taking action $a$ under policy $\pi$
$$
q_\pi(s, a)=\mathbb{E}_\pi\left[G_t \mid s_t=s, a_t=a\right].
$$
The aim is to find the optimal policy that maximizes the expected return. 
Such a policy can be derived from the optimal action-value function $q_*$, defined by the Bellman optimality equation:
$$
q_*(s, a)=\mathbb{E}\left[r_{t+1}+\max_{a^{\prime}} q_*\left(s_{t+1}, a^{\prime}\right) \mid s_t=s, a_t=a\right].
$$
Instead of directly solving the Bellman optimality equation in value-based RL, the aim is to learn the optimal action-value function from data samples. 
One such prominent value-based $\mathrm{RL}$ algorithms is $Q$-learning, where each state-action pair $(s, a)$ is assigned a so-called $Q$-value $Q(s, a)$ which is updated to approximate $q_*$.
Starting from randomly initialized values, the Q-values are updated according to the following rule:
$$
Q(s_t, a_t) \leftarrow Q(s_t, a_t) + \alpha \left(r_{t+1} +\gamma \max_{a^{\prime}} Q\left(s_{t+1}, a^{\prime}\right) - Q(s_t, a_t)\right),
$$
where $\alpha$ is the learning rate, $r_{t+1}$ is the reward at time $t+1$, and $s_{t+1}$ is the next encountered state after taking action $a_t$ in state $s_t$. 
\begin{algorithm}[H]
	\caption{Double Q-Learning}\label{alg:DDQN-pseucode}
	\begin{algorithmic}
		\State Initial network $\leftarrow Q^\theta$, target network $\leftarrow Q^{\theta^\prime}$, replay buffer $\leftarrow\mathcal{D}$, $\tau<<1$ \Comment{Input}
		\For {each number of iterations}
		\For {each step}
		\State Observe state $s_t$ and choose $a_t\sim\pi(a_t,s_t)$
		\State Apply $a_t$, observe $s_{t+1}$ and $r_t = R(a_t,s_t)$
		\State Store $(s_t,a_t,r_t,s_{t+1})$ in $\mathcal{D}$
		\EndFor
		\For{each update step}
		\State Sample $e_t=(s_t,a_t,r_t,s_{t+1})\sim\mathcal{D}$
		\State Compute Q-value:
		\begin{equation}
			Q^*(s_t,a_t)\approx r_t +\gamma Q^\theta\left( s_{t+1},\mathrm{argmin}_aQ^{\theta^\prime}(s_{t+1},a) \right)
		\end{equation}
		\State Perform gradient descent on $\left[Q^*(s_t,a_t)-Q^\theta(s_t,a_t)\right]^2$
		\State Update the Q-network parameter:
		\begin{equation}
			\theta^\prime \leftarrow \tau\times\theta + (1-\tau)\times\theta^\prime
		\end{equation}
		\EndFor
		\EndFor
	\end{algorithmic}
\end{algorithm}

In the limit of visiting all $(s, a)$ pairs infinitely often, this update rule converges to the optimal Q-values in the tabular case~\cite{melo2001convergence}.
In practice, a so-called $\epsilon$-greedy policy is used to ensure sufficient exploration in a Q-learning setting. Formally, stated as,
$$
\pi(a \mid s)=\left\{\begin{array}{l}
	1-\epsilon_t \text { for } a=\max _{a^{\prime}} Q\left(s, a^{\prime}\right) \\
	\epsilon_t \text { otherwise }
\end{array}\right.
$$
The $\epsilon$-greedy policy is only used to introduce randomness to the actions selected by the agent during training, but once training is finished, a deterministic policy follows.}
{We employ neural networks (NN) as function approximators to extend Q-learning to large state and action spaces.
NN training typically requires independently and identically distributed data, which isn't naturally available in the sequential RL data. 
This problem is circumvented by experience replay.
This method divides past experiences into single-episode updates, creating batches randomly sampled from memory.
To stabilize training, two NNs are employed: a policy network that is continuously updated and a target network that is an earlier copy of the policy network. 
The policy network estimates the current value, while the target network provides a more stable target value represented by 
$Y$:
$$
Y_{\mathrm{DQN}}=r_{t+1}+\gamma \max_{a^{\prime}} Q_{\mathrm{target}}\left(s_{t+1}, a^{\prime}\right)
$$}
{In the Double deep Q-network (DDQN) algorithm, the action for the target value is sampled from the policy network to reduce the overestimation bias inherent in standard DQN. The corresponding target is defined as:
$$
Y_{\mathrm{DDQN}}=r_{t+1}+\gamma Q_{\mathrm{target}}\left(s_{t+1}, \underset{a^{\prime}}{\arg\max } Q_{\text {policy }}\left(s_{t+1}, a^{\prime}\right)\right) .
$$
{This target value is approximated using a selected loss function, in this case, a smooth L$1$-norm loss.}

\section{RL-based quantum architecture search algorithm}
Through the discussions in the previous sections, we notice that learning an RL-agent primarily depends on the following ingredients. The environment specifications, a proper description of environment (RL-state), the formulation of the reward function and the action space encoding. Hence, while constructing an RL-based quantum architecture algorithm to solve variational quantum algorithms, it is crucial to define the environment, the state of the RL, the action, and the reward function in an agent-friendly way. 
\begin{figure}[h]
	\centering
	\includegraphics[width=0.9\columnwidth]{ 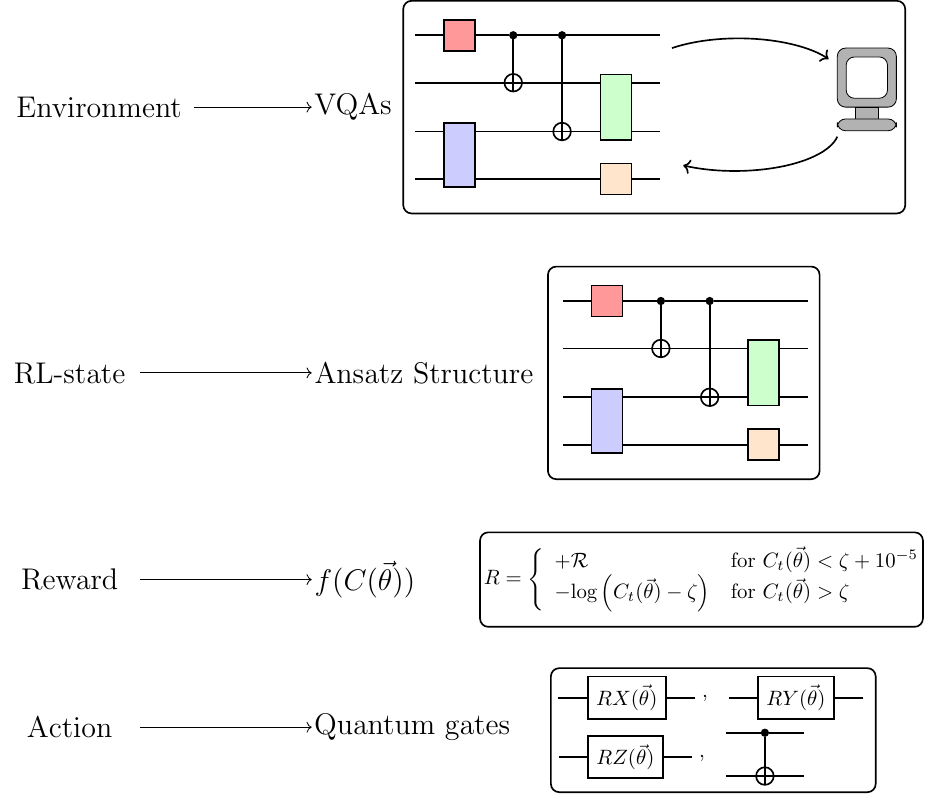}
	\caption{The crucial ingredients to cook a successful RL-based QAS algorithm. Here, the environment is defined through the hybrid quantum-classical algorithm. The RL-state is represented through the variational quantum circuit, the reward function is a function of the cost function that encodes the problem, and the actions are defined by one- and two-qubit quantum gates.}
	\label{fig:rlqas_ingredients}
\end{figure}


In the \figref{fig:rlqas_ingredients}, we illustrate an effective yet simple way to define the ingredients in order to construct an RL-based quantum architecture search (QAS) algorithm. The environment is specified through the quantum-classical algorithm, where the variational quantum circuit, i.e. the ansatz, is considered RL-state. It should be noted that there are several ways to encode the ansatz in an RL-agent readable way~\cite{fosel2021quantum, ostaszewski2021reinforcement}, but we (which is elaborated in the upcoming chapters) utilize a novel 3D tensor-based binary encoding scheme. Meanwhile, the reward function can be defined in a sparse or dense manner. In the illustration, the reward is dense as the cost function is calculated, and then, based on its value, the reward is defined in each step $t$ of an episode. If the cost function approaches the predefined threshold value $\zeta$, the agent receives a highly positive response (+$\mathcal{R}$). In other scenarios, instead of punishing the agent, we define the cost so that the agent reaches closer to the goal (here, the predefined threshold $\zeta$), and the reward becomes more positive. This makes the learning curve of the agent smoother and not stricter than a sparse (punishing) reward function. This will be more elaborately discussed in the upcoming section when we present the RL-enhanced variational quantum state diagonalization (RL-VQSD) algorithm in chapter~\ref{ch:vqsd_using_rl}. It should be noted that we also make use of a sparse reward function in the case of the curriculum RL-based variational quantum eigensolver (CRLVQE) algorithm presented in chapter~\ref{ch:crlvqe} in order to compare its performance with the state-of-the-art RL methods.
Finally, after each step $t$, the agent decides on an action to take from the action space. The action space is represented by a continuous set of parameterized one-qubit rotation gates, and for two-qubit gates throughout the article, we utilize the \texttt{CNOT} gate. The action space is defined using one-hot encoding, which will be elaborated in the upcoming chapter.

Having all the ingredients for constructing an RL-based QAS algorithm, the upcoming chapters focus on specific applications of this design. 

{%
	\section{The network architecture in Double Deep-Q Network}
	The neural network in the DDQN implementation is a \textit{fully connected feedforward network~\cite{abiodun2018state} also known as a multilayer perceptron} that serves as a function approximator for the Q-function \( Q(s, a) \). The architecture is defined as follows:
	\paragraph{Input Layer}
	The input layer has a size of RL-state, which corresponds to the number of features representing the state of the environment. In the case of quantum architecutre search the RL-state corresponds to the quantum circuit encoding represented in later section~\ref{sec:binary-encoding-scheme}. The state size can be adjusted based on the inclusion of additional information such as type of gates, which qubit it is added and the depth of the circuit.}
\begin{figure}[t!]
	\includegraphics[width=\linewidth]{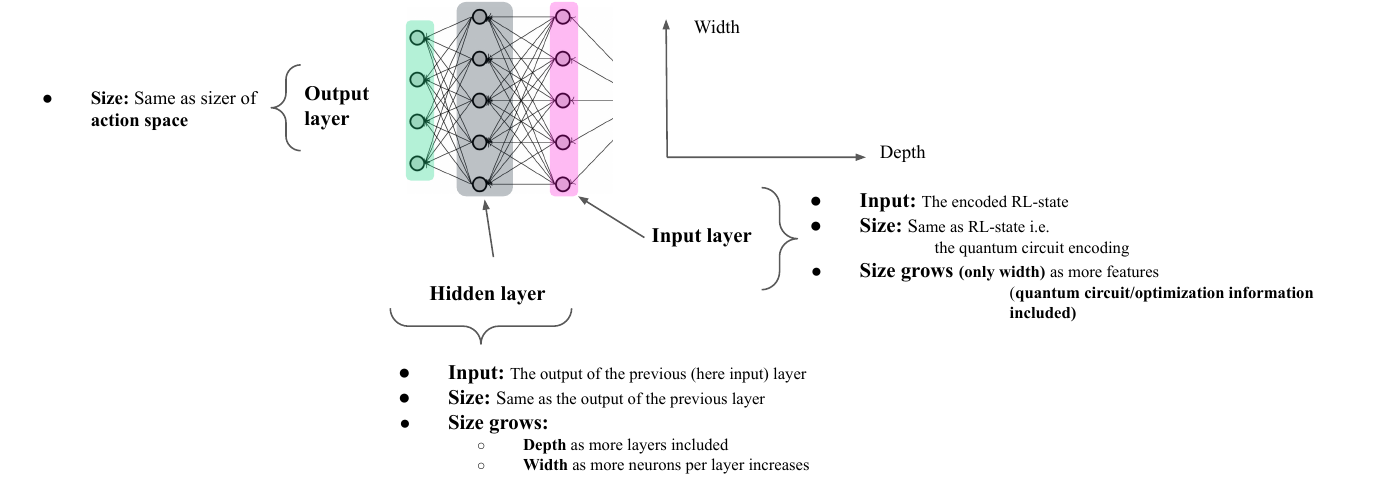}
	\caption{A schemetic of multi-layer perceptron architecture. The size of the input state is the size of the encoded ansatz and the output state size depends on the size of the action space. The action space contains all possible combinations of quantum gates on different qubits.}
\end{figure}

{\paragraph{Hidden Layers}
	The hidden layers are fully connected, with each layer followed by a LeakyReLU activation function and dropout for regularization:
	\[
	\text{Hidden Layers:} \quad [\text{Linear}(n_i, n_{i+1}) + \text{LeakyReLU} + \text{Dropout}(p)] 
	\]
	where \( n_i \) represents the number of neurons in layer \( i \) and \( p \) is the dropout probability. It should be noted that $n_1$ = size of the RL-state.
	\paragraph{Output Layer}
	The output layer is fully connected with a size equal to the action space. In our case the possible quantum gates and its connectivity among qubits. This layer directly produces the Q-values for all possible actions, without any activation function.
	\subsection{Role of the network and its training}
	The network (which we call the policy network) corresponding to the state of the environment, selects an action based on the $\epsilon$-greedy policy~\ref{eq:epsilon-greedy}. A copy of the network (which we term as target network) is updated periodically to stabilize learning process. The training of the target network involves the following steps:
	\begin{itemize}
		\item \textbf{Forward pass:} To compute Q-values for each possible action in a given state using the policy network.
		\item \textbf{Loss Calculation:} Computes the loss using the smooth L1 loss function. The smooth L1 loss is defined by:
		\begin{equation}
			\text{L1}(x) = 
			\begin{cases} 
				0.5 \cdot x^2 & \text{if } |x| < \beta, \\
				|x| - 0.5\beta & \text{otherwise},
			\end{cases}
		\end{equation} 
		where $x$ is the difference between the target and predicted value and $\beta$ is a hyperparameter that determines the threshold at which the loss transitions from quadratic (L2) to linear (L1). When $|x|<\beta$, the loss behaves like L2 loss.
		\item \textbf{Backpropagation:} Calculates the gradients with respect to the loss obtained in the previous step. We encourage the reader to see ref~\cite{rojas1996backpropagation} for the details of the backpropagation algorithm.
		\item \textbf{Optimization:} Update the policy network parameters using the Adam optimizer~\cite{kingma2014adam}.
	\end{itemize}
	{This network architecture, with its configurable depth and regularization techniques, allows the DDQN agent to effectively learn policies. In a recent study~\cite{kundu2024kanqas}, the authors replace the traditional \textit{multilayer perceptron} structure with a \textit{Kolmogorov-Arnold network}~\cite{liu2024kan} architecture for DDQN in QAS. This approach demonstrates significant potential in optimizing 2-qubit gates compared to the state-of-the-art multilayer perceptron models.}

\chapter{Reinforcement Learning assisted Variational Quantum State Diagonalization: RL-VQSD }\label{ch:vqsd_using_rl}

In this chapter, we provide an overview of the results concerning the utilization of Reinforcement Learning in the task of Variational Quantum State Diagonalization which we call the RL-VQSD method. We start with a review of the state-of-the-art approach for quantum state diagonalization, followed by an investigation of the efficiency of various classical optimization techniques and ansatz structure. 
Next, we introduce a novel binary encoding scheme for quantum circuits that improves the 
\emph{efficient search for a diagonalizing unitary and provides an enhancement in the accuracy for the variational quantum state diagonalization (VQSD) scheme}. Moreover, a carefully constructed dense reward function makes the RL-VQSD more efficient in terms of the number of gates and depth of the diagonalizing unitary.
In the last part, we demonstrate the example where the proposed techniques lead to a significant improvement of the VQSD algorithm.

\section{Introduction}
One of the most prominent variational quantum algorithms is called the variational quantum state diagonalization (VQSD)~\cite{larose2019variational}. This procedure utilizes a quantum-classical hybrid procedure to identify the unitary rotation under which the given quantum state becomes diagonal in the computational basis, i.e. it
diagonalizes a quantum state. This has  several applications, including quantum state fidelity
estimation~\cite{cerezo2020variational}, device certification~\cite{kundu2022variational}, Hamiltonian diagonalization~\cite{zeng2021variational}, and extracting the entanglement properties of a system~\cite{larose2019variational}. VQSD generalizes the well-studied problem of quantum state preparation, which can be understood as quantum state tomography for pure states. Considering it has applications that range from quantum information to condensed matter physics, an efficient way to deal with quantum state diagonalization may lead to interesting insights in these fields.

It is worth mentioning that classical methods of diagonalization scale polynomially with the size of the matrix~\cite{demmel2008performance} and in {Tab.~\ref{tab:classical_diag_algo}} we list a few well-known classical diagonalization algorithms and their complexity. Meanwhile, quantum principal component analysis~\cite{lloyd2014quantum} (qPCA) proposes an exact algorithmic method for quantum state diagonalization. 
\begin{table}[H]
	\centering
	\begin{tabular}{c|c}
		Methods           & Asymptotic Complexity \\
		\hline
		QR iteration~\cite{watkins1982understanding}
		& O($n^2$)                 \\
		BI iterations~\cite{strobach1997bi}
		& O($nm$)                 \\
		Divide \& conquer~\cite{gu1995divide}           
		& O($n^3$)                 \\
		MR method~\cite{petschow2012algorithm}
		& O($n^2$)                
	\end{tabular}
	\caption{Complexity comparison of the classical methods for diagonalization.}
	\label{tab:classical_diag_algo}
\end{table}

Nevertheless, qPCA often results in complex circuits, which variational approaches could potentially surpass. However, a significant challenge in VQSD lies in devising an \textit{efficient} ansatz capable of diagonalizing a specific quantum state. 


Many factors can be used as a pointer for an efficient ansatz ($\mathcal{A}_\text{eff}$) but for the sake of the thesis we primarily focus on the following definition to define an efficient ansatz~\cite{kundu2024enhancing}:



\begin{definition}\label{def:efficient_ansatz}
	An ansatz is efficient if the depth and the total number of gates are smaller compared to the state-of-the-art ansatz structures, and which returns a lower error in solving the problem.
\end{definition}


Here, we can formulate the task of finding an efficient ansatz that can be expressed (ideally) in the following way
\begin{equation}
	\mathcal{A}_\text{eff} = \underset{\mathcal{D},\mathcal{N}_g,\Delta}{\text{min}} f(\mathcal{D},\mathcal{N}_g,\Delta),
	\label{eq:efficient_ansatz}
\end{equation}
where $\mathcal{D}$ is the total depth, $\mathcal{N}_g$ is the total number of gates in the ansatz. Meanwhile, $\Delta$ is the error we receive using the ansatz. Based on the fact that the parameters $\mathcal{D}$, $\mathcal{N}_g$, and $\Delta$ increases with the increase in number of qubits the overall $\mathcal{A}_\text{eff}$ is a monotonically increasing function of these parameters. And the task of finding the efficient ansatz boils down to decreasing the rate of the growth of this function with the number of qubits. 
It should be noted that based on the problem, the definition of $\Delta$ varies, as we will show in the upcoming sections.


In the standard VQSD methods, a layered hardware efficient ansatz (LHEA) is
utilized as shown in~\figref{fig:layered-ansatz}. 
\begin{figure}[ht!]
	\begin{subfigure}{\textwidth}
		\centering
		\begin{tikzpicture}
			\node[scale=0.8] {
				\begin{quantikz}
					& \gate[wires=2][1cm]{G(\vec{\theta})} & \qw \\
					& \qw &\qw
				\end{quantikz}
				=\begin{quantikz}
					& \gate{RY(\theta_1)} & \ctrl{1} & \gate{RY(\theta_1)} \\
					& \gate{RY(\theta_1)} & \ctrl{-1} & \gate{RY(\theta_1)}
				\end{quantikz}
			};
		\end{tikzpicture}
		\label{fig:layered-anzatz-one}
	\end{subfigure}
	\begin{subfigure}{\textwidth}
		\centering
		\begin{tikzpicture}
			\node[scale=0.8] {
				\begin{quantikz}
					& \gate[wires=2][1cm]{G(\vec{\theta})} & \qw \\
					& \qw &\qw
				\end{quantikz}
				=\begin{quantikz}
					& \gate{RZ(\theta_1)} & \gate{RY(\theta_2)} & \gate{RZ(\theta_3)} & \ctrl{1} & \gate{RZ(\theta_1)} & \gate{RY(\theta_2)} & \gate{RZ(\theta_3)}\\
					& \gate{RZ(\theta_1)} & \gate{RY(\theta_2)} & \gate{RZ(\theta_3)} & \targ{} & \gate{RZ(\theta_1)} & \gate{RY(\theta_2)} & \gate{RZ(\theta_3)}
				\end{quantikz}
			};
		\end{tikzpicture}
		\label{fig:layered-anzatz-three}
	\end{subfigure}
	\caption{Two possible decompositions of the two-qubit rotations in each layer-wise unitary $U_i(\vec{\theta}_i)$ of the layer-ansatz in~\figref{fig:layered-ansatz}.}
	\label{fig:two_qubit_rot_decomp_layered_ansatz}
\end{figure}
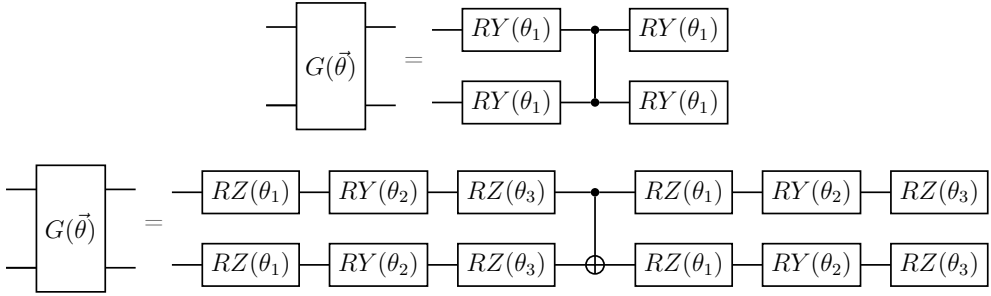

A single layer of the ansatz contains two-qubit gates acting on neighboring qubits. These gates are decomposed in two ways illustrated in Fig.~\ref{fig:two_qubit_rot_decomp_layered_ansatz}. Although in the LHEA, the parameter count increases linearly with the number of layers and qubits, it has trainability issues and often encounters local minima. To tackle the trainability issue, instead of using a fixed structure of LHEA, the authors allow additional updates (i.e. changes in the ansatz structure) during the classical optimization process. In this process, every optimization step involves the minimization of the cost function with a small random change to the ansatz structure. The new structure is approved or rejected based on a simulated annealing scheme~\cite{cincio2018learning}. Although the varying structure LHEA outperforms fixed structure LHEA, the number of gates in the quantum circuit increases rapidly as we scale the size of the quantum state. Hence, the problem of finding a method to construct an ansatz that satisfies all efficiency criteria is still an open problem.

To tackle the problem, we incorporate a Reinforcement Learning (RL) agent to automate the search for an efficient ansatz for VQAs. In the past several novel approaches have been introduced in the light of machine learning techniques to address the challenge of finding a new architecture of ansatz~\cite{ostaszewski2021reinforcement,kuo2021quantum,ye2021quantum,he2023gnn,zhang2022differentiable,du2022quantum,patel2022reinforcement,moro2021quantum, cao2022quantum}. In a nutshell, the RL-based algorithm, i.e. RL-VQSD we introduce in~\cite{kundu2024enhancing} utilizes a state-of-the-art encoding method for RL-state space, representing the variational circuit, along with a dense reward function to show that our approach can be successfully used for
diagonalizing arbitrary mixed quantum states. Moreover, we also demonstrate that compared to LHEA, the ansatz proposed by the RL-agent is more efficient (Def.~\ref{eq:efficient_ansatz} for the definition of efficient ansatz). The methods we discuss are algorithm-independent, so they can be easily adopted to tackle any VQAs.


In the following, we first briefly discuss the VQSD algorithm and the ways to implement it. Using LHEA we analyze its performance for random quantum states. Next, we give a brief discussion on how we can incorporate reinforcement learning (RL) to automate the VQSD process, which throughout the article we call the RL-VQSD method.

\section{Previous work}
In this section, we briefly discuss the variational quantum state diagonalization algorithm and analyze its performance. We will use the described procedure as the starting point for the utilization of reinforcement learning to improve efficiency.

\subsection{The VQSD algorithm} \label{sec:vqsd-algorithm-brief}
The variational quantum state diagonalization algorithms, introduced in~\cite{larose2019variational}, aim to identify the unitary rotation under which the given quantum state becomes diagonal on the computational basis. Hence, for a given state $\rho$, VQSD composed of the three following subroutines (see~\figref{fig:ill-vqsd})

\begin{itemize}
	
	\item \textsc{Training} In 
	this subroutine, for a given state 
	$\rho$, one optimizes the 
	parameters $\vec{\theta}$ of a 
	quantum gate sequence 
	$U(\vec{\theta})$, which 
	(ideally) after optimization 
	satisfies 
	\begin{equation}
		\tilde{\rho} = U(\vec{\theta}_\text{opt})\rho U(\vec{\theta}_\text{opt})^\dagger = \rho_\text{diag},
	\end{equation}
	where $\rho_\text{diag}$ is the diagonalized $\rho$ in its eigenbasis and $\vec{\theta}_\text{opt}$ are the optimal angles. One can utilize classical gradient-based methods such as \texttt{SPSA}~\cite{spall1998implementation} and Gradient-Descent~\cite{ruder2016overview}, or gradient-free optimization methods such as \texttt{COBYLA} and \texttt{POWELL}~\cite{powell1994direct} in the training process.
	
	\item \textsc{Eigenvalue Readout} In this subroutine, using the optimized unitary $U(\vec{\theta}_\text{opt})$ and one copy of state $\rho$, one can extract -- for low-rank states -- all the eigenvalues or --  for full-rank state -- the largest eigenvalues. This is achieved by measuring the $\tilde{\rho}$ in the computational basis, $\mathbf{b}=b_1b_2\ldots b_n$, as follows
	\begin{equation}
		\tilde{\lambda} = \bra{\mathbf{b}}\tilde{\rho}\ket{\mathbf{b}},
	\end{equation}
	where $\tilde{\lambda}$ are inferred eigenvalues.
	
	\item \textsc{Eigenvector Preparation} In the final {step} one can prepare the eigenvectors associated with the largest eigenvalues. If $\tilde{\mathbf{b}}$ is a bit string associated with $\tilde{\lambda}$ then one can get the inferred eigenvectors $|\tilde{v}_{\tilde{\mathbf{b}}}\rangle$ as follows
	\begin{equation}
		|\tilde{v}_{\tilde{\mathbf{b}}}\rangle = U(\theta_\textrm{opt})^\dagger\ket{\tilde{\mathbf{b}}} = U(\theta_\textrm{opt})^\dagger\left(X^{b_1}\otimes\ldots\otimes X^{b_n} \right)\ket{\mathbf{0}}.
	\end{equation}
\end{itemize}

\begin{figure}[ht]
	\centering
	\includegraphics[width = 0.75\linewidth]{ 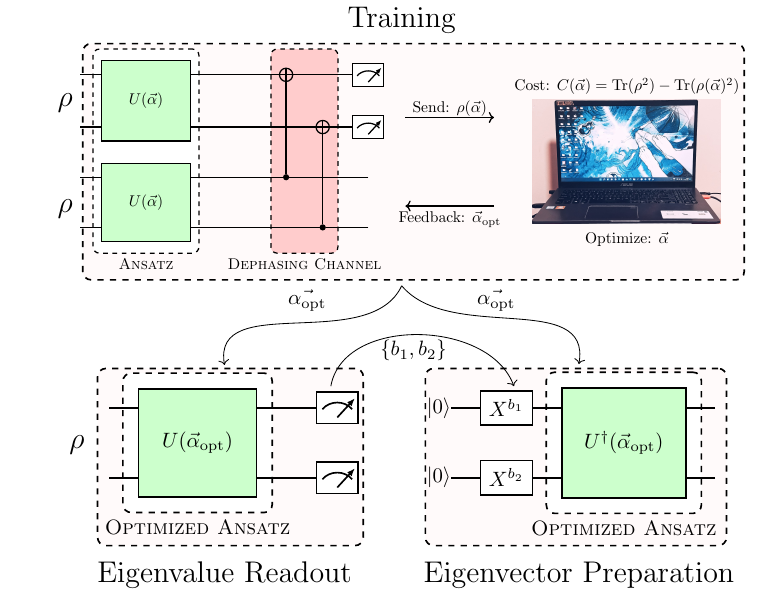}
	\caption{Elements of variational quantum state diagonalization (VQSD) algorithm. In the presented example, we consider the diagonalization for a two-qubit input state.}
	\label{fig:ill-vqsd}
\end{figure}

\subsection{The cost function}

The VQSD algorithm focuses on finding a unitary that diagonalizes a quantum state; hence, the cost function should encode the information about how far a state $\rho$ is from being diagonal. Meanwhile, we need to ensure that the cost function choice follows the three golden rules of cost function construction that are briefly discussed in {Sec.~\ref{sec:cost_function_introduced}}. One such cost-efficiently computable cost function is
\begin{equation}
	C(\vec{\theta}) = \text{Tr}(\rho^2) - \text{Tr}(\mathcal{D}(\tilde{\rho})^2),\label{eq:cost_func_vqsd}
\end{equation}
where $\mathcal{D}$ denotes the dephasing channel to minimize the off-diagonal terms in $\tilde{\rho}$. The cost function~\eqref{eq:cost_func_vqsd} is upper bounded by the error in the eigenvalue of $\rho$. Where we define the eigenvalue error by
\begin{equation}
	\Delta_i = \sum_{i=1}^m\left(\lambda_i - \tilde{\lambda}_i\right)^2,
\end{equation}
where $m$ two-qubit the number of the largest eigenvalues, $\lambda_i$ is the true eigenvalue and $\tilde{\lambda}_i$ is the inferred eigenvalue obtained from the \textsc{eigenvalue readout} subroutine. In the ideal case, where the state is completely diagonalized, $m = 2^n$ indicates all the eigenvalues have been considered. 

It should be noted that in the cost function~\eqref{eq:cost_func_vqsd}, the 1st term on the RHS can be computed outside the optimization loop, and we left out with the main task of evaluating the 2nd term. Also, the cost function landscape for the cost function turns insensitive to the changes in the ansatz when we scale up the size of the quantum state, but as we constrain ourselves to $n\leq6$ qubits, we can efficiently utilize the cost~\eqref{eq:cost_func_vqsd}.

In the following, we first analyze the performance of different ansatz and optimizers to diagonalize a two-qubit mixed quantum state.
\subsection{Benchmarking the performance of LHEA}


In the following, we rigorously investigate the state-of-the-art VQSD method for two- and three-qubit uniformly distributed random quantum states that are generated by using the \texttt{random\_density\_matrix} generator of \texttt{qiskit.quantum\_info} module\footnote{\href{https://qiskit.org/documentation/stubs/qiskit.quantum_info.random_density_matrix.html}{https://qiskit.org/documentation/stubs/qiskit.quantum\_info.random\_density\_matrix.html}} of \texttt{qiskit}~\cite{qiskit2019}. We do the investigation for two-qubit random mixed quantum states and focus on the performance of LHEA. Further, we benchmark classical optimization methods in the task of diagonalizing a two-qubit mixed quantum state.
\begin{figure}[ht]
	\centering
	\includegraphics[scale=0.9]{ 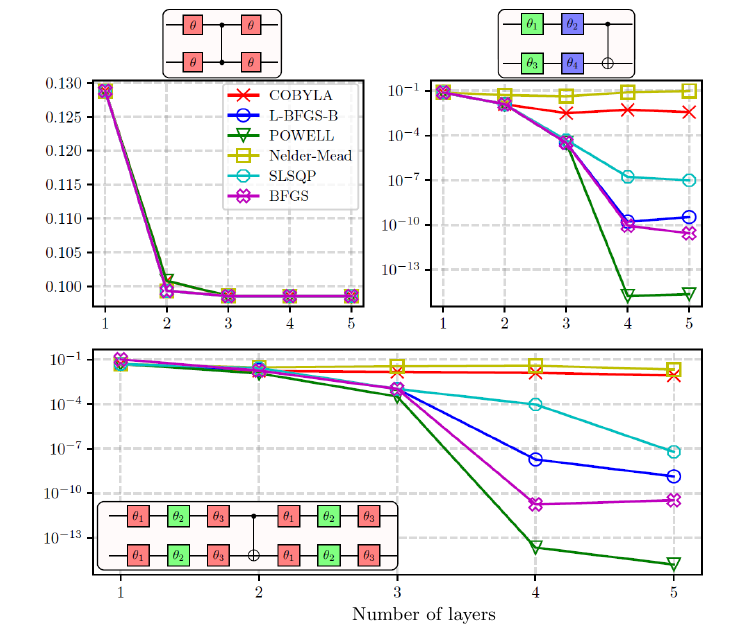}
	\caption{Three kinds of LHEA in diagonalizing a two-qubit quantum state using VQSD. We can see from the results that the \texttt{RYCZ} structure performs worse than the other two structures in terms of finding the optimal cost function for the task.
		In the pictures \fcolorbox{black}{red!60}{\rule{0pt}{6pt}\rule{6pt}{0pt}}=\texttt{RY}, \fcolorbox{black}{green!60}{\rule{0pt}{6pt}\rule{6pt}{0pt}}=\texttt{RZ} and \fcolorbox{black}{blue!60}{\rule{0pt}{6pt}\rule{6pt}{0pt}}=\texttt{RX} gate.}
	\label{fig:diff_ansatz_compare}
\end{figure}

\paragraph{Performance of LHEA}\label{sec:peformance_of_lhea} 

In~\figref{fig:diff_ansatz_compare}, we explore the possible constructions of HEA and briefly discuss their performance. For name convention, we call the ansatz in LHS of the first row as \texttt{RYCZ} and the RHS one as \texttt{RZRX CNOT} ansatz where we call the ansatz in the 2nd row as \texttt{RYRZRY CNOT} ansatz. It can be seen that for the \texttt{RYCZ} ansatz, the convergence with the number of parameters saturates at $10^{-1}$ for all kinds of optimizers, whereas the other two ansatz shows good performance in minimizing the cost. 
\begin{figure}[ht]
	\centering
	\includegraphics[scale=0.9]{ 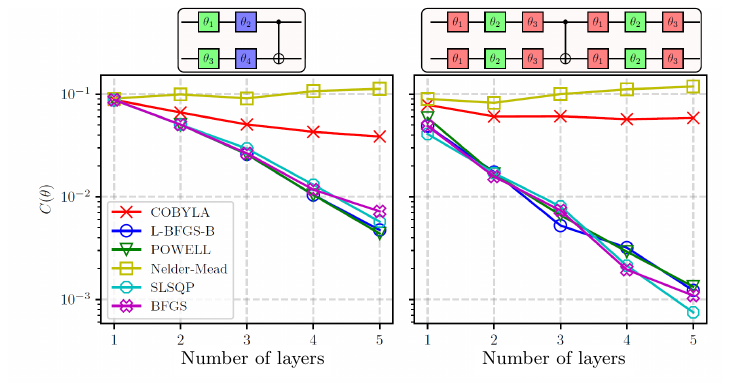}
	\caption{By considering two best performing structures of ansatz for two-qubit, we investigate the performance of these structures in cost function minimization for three-qubit state diagonalization. The figure shows that \texttt{RYRZRY CNOT} structure outperforms in terms of accuracy compared to \texttt{RZRX CNOT}.}
	\label{fig:3q_ansatz_compare}
\end{figure}



The convergence of the \texttt{RYRZRY CNOT} and \texttt{RZRX CNOT} ansatz shows comparable performance. Both ansatzes have their advantages and disadvantages which are discussed in the following points.


\begin{itemize}
	\item In the case of the \texttt{RZRX CNOT} ansatz, the number of parametrized gates is very small and relatively susceptible to noise. On the other hand, due to a higher number of different parameters, i.e. $\theta_1,\theta_2,\theta_3$ and $\theta_4$ for larger quantum states, it fails to optimize the cost function.
	\item In the case of the \texttt{RYRZRY CNOT} ansatz, the number of different parameters in gates is very small, making it more suitable for optimization. On the other hand, due to the increase in number of gates, it is not susceptible to noise.
\end{itemize}
In~\figref{fig:3q_ansatz_compare} we further investigate the performance of \texttt{RYRZRY CNOT} and \texttt{RZRX CNOT} ansatz for a three quantum mixed quantum state, and we can clearly see that \texttt{RYRZRY CNOT} outperforms the other so we use the \texttt{RYRZRY CNOT} construction of LHEA for VQSD throughout the paper, if not stated otherwise.
\paragraph{Benchmarking optimizers}

In~\figref{fig:optimizer_benchmark_vqsd} \texttt{RYRZRY CNOT} construction of LHEA to diagonalize 100 Haar-random mixed quantum states is presented.

\begin{figure}[ht]
	\centering
	\includegraphics[scale=0.9]{ 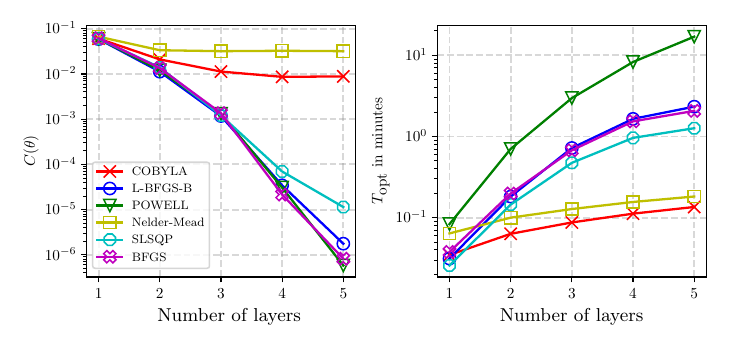}
	\caption{Different optimizers' performance in minimizing the cost function in diagonalizing 100 random two-qubit mixed states. It can be seen that among the gradient-free optimizers, {POWELL} performs optimally, whereas the time ($T_\textrm{opt}$ in seconds) it takes to provide the optimized result is higher than other optimizers. In the case of the gradient-based optimizer, the \texttt{L-BFGS-B} and the \texttt{BFGS} outperform other similar optimizers.}
	\label{fig:optimizer_benchmark_vqsd}
\end{figure}



From the left-hand side of~\figref{fig:optimizer_benchmark_vqsd} we can see that we need at least 3 layers to get a good approximation to the diagonal state of an arbitrary two-qubit mixed state. This indicates we need at least 4 \texttt{CNOT} and $12\times4 =48$ parameterized gates with a total depth of $ 7\times4 = 28$ to get a good estimate on the diagonalizing unitary. Among all the optimizers, the POWELL optimizer gives us the best outcome; meanwhile, on the right-hand side of~\figref{fig:optimizer_benchmark_vqsd} we notice that the $T_\text{opt}$, i.e. the time (in minutes) it takes to finish a complete round of optimization\footnote{Each round of optimization contains 10 runs of the optimizer with uniform random initialization and we choose the best parameters among the 10 runs.} grows rapidly as we increase the number of layers. On the other, COBYLA finishes the complete optimization process in just 0.1 minutes. Still, the cost function does not go below $10^{-3}$ and hence fails to give us a good approximation to the diagonal state.

In the following, we primarily use COBYLA in the classical optimization subroutine for reinforcement learning-assisted VQSD, i.e., the RL-VQSD method. The main motivation behind this decision is the following. While utilizing a learning-based method like RL \textit{we start from an empty quantum circuit, and using a state-of-the-art encoding method, we encode the circuit into an RL-state}. After each application of a gate to the quantum circuit, the RL-state updates, and the circuit passes through a classical optimization method. Using the outcome of the optimization, we evaluate a reward function. \textit{Based on the reward, the RL-agent decides on an action that encodes the information of a particular quantum gate}. This process repeats until a predefined threshold in cost estimation is reached, and the total number of episodes is exhausted. After the application of each action, we harness the power of classical optimization to make sure that the optimization process does not consume most of the resources in the RL-VQSD.

\section{Components of RL-VQSD algorithm}

Before jumping into RL-VQSD, in this section, we provide a detailed description of the building blocks used in the proposed algorithm. As in any RL method, the proposed methods included the specification of the state, the actions, and the reward function.

\subsection{RL-State}\label{sec:binary-encoding-scheme}

Learning-based quantum architecture search algorithms require a concise circuit representation that is commonly referred to as encoding. Through the encoding scheme, we can modify, compare, and explore the quantum circuit. This helps us to navigate the search space of all possible alternatives efficiently and discover effective and innovative quantum circuits. In the following, we will elaborate on a tensor-based binary encoding approach for the quantum circuit, which can efficiently be utilized as an RL-state.
The encoding was first introduced in~\cite{rl-vqe-paper}. In this scheme, the gate structure of the ansatz is expressed as a tensor of dimension $\left[T \times ((N + 3) \times N )\right]$, where $N$ two-qubit the size of the problem and $T$ is the considered maximum depth of the ansatz. For VQSD, $N$ {two-qubit the number of qubits in the quantum state} that need to be diagonalized. The proposed encoding  can be explained through the following two points:
\begin{enumerate}
	\item \textbf{Freedom in connectivity} The encoding enables \textit{all-to-all} qubit connectivity, but it can be restricted by considering \textit{unidirectional nearest neighbour} connections only. In this scenario, the matrix dimension $\left((N+3)\times N \right)$ is reduced to $\left(4\times N \right)$. One should note that in the case of a two-qubit gate, one is not required to keep track of the control and target simultaneously. Hence, defining one argument of the two-qubit gate implicitly provides information about the other argument due to its nearest neighbor and unidirectional nature. A similar encoding scheme is described in~\cite{fosel2021quantum}.
	
	\begin{figure}[ht]
		\centering
		\includegraphics[width = 0.5\linewidth]{ 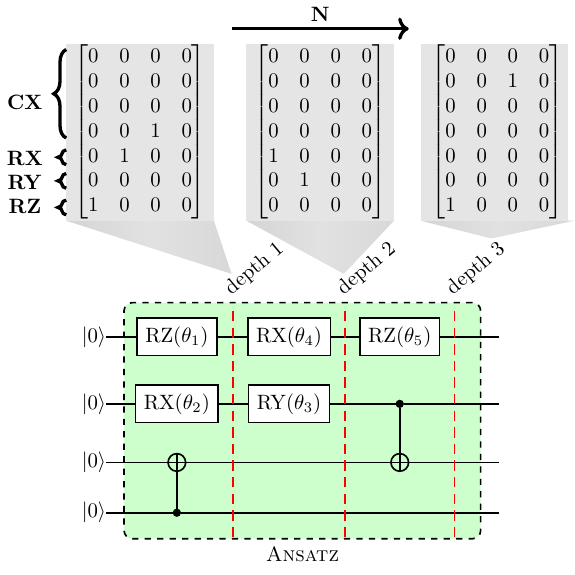}
		\caption{Example of the proposed encoding for four-qubit ansatz. The first $\left(N\times N\right)$ square matrix is reserved for the \texttt{CNOT} connectivity. The columns of the square matrix encode the \textit{target qubit}, and the rows represent \textit{control qubits}. The remaining $\left((N+j)\times N\right)$ elements encode arbitrary rotation towards $j$ direction where $j = 1$, $2$, and $3$, for $\texttt{X}$, $\texttt{Y}$ and $\texttt{Z}$ rotations, respectively.}
		\label{fig:encoding}
	\end{figure}
	
	\item \textbf{Depth-based encoding} In previous work~\cite{ostaszewski2021reinforcement} each $\left((N+3)\times N \right)$ matrix carries information corresponding to each action taken by the agent, where each action two-qubit either a single or a two-qubit gate. Additionally, the information was integer-based, in the range $0$ to $N$.
	
	On the contrary, In our work, the encoding is binary and depth-based. For example, if $T=3$, then the encoding initiates by filling up the~$\left[i \times ((N + 3) \times N )\right]$ for $i =1$ until a depth of RL-ansatz is encoded. Which is described as follows:
	\begin{equation}
		\left[ \underbrace{ ((N + 3) \times N )}_{\textrm{depth = 1}}, \underbrace{((N + 3) \times N )}_{\textrm{all zeros}},\underbrace{((N + 3) \times N )}_{\textrm{all zeros}} \right].
	\end{equation}
	Then, as $i=1$ is filled up, we move to $i=2$ to encode depth = 2 of the RL-ansatz, which yields 
	\begin{equation}
		\left[ \underbrace{ ((N + 3) \times N )}_{\textrm{depth = 1}}, \underbrace{((N + 3) \times N )}_{\textrm{depth = 2}},\underbrace{((N + 3) \times N )}_{\textrm{all zeros}} \right].
	\end{equation}
	Finally, the depth = 3 is encoded in $i=3$ resulting in
	\begin{equation}
		\left[ \underbrace{ ((N + 3) \times N )}_{\textrm{depth = 1}}, \underbrace{((N + 3) \times N )}_{\textrm{depth = 2}},\underbrace{((N + 3) \times N )}_{\textrm{depth = 3}} \right].
	\end{equation}
\end{enumerate}
Each depth encoding follows the scheme shown in \ref{fig:encoding}. In the following, we give a simple example of $T=3$ encoding. In the {Appendix~\ref{app:rlstate_code}}, we provide a code that is used for the RL-state encoding.

\subsection{RL-action}\label{sec:rl-action}

For constructing the quantum circuits, we use the scheme developed in~\cite{ostaszewski2021reinforcement} with \texttt{CNOT} and one-qubit rotation gates, which are feasible on currently available quantum devices. The encoding of the action space can be defined as follows. The \texttt{CNOT} gates are represented by a pair of values that indicate the positions of the control and target qubits, with enumeration starting from 0. As for the rotation gates, they are encoded using two integers, also starting from 0. The first integer identifies the qubit register, while the second integer specifies the rotation axis. For an $N$-size quantum state,  the agent can choose from $3\times N$ single-qubit gates and from $2\times\binom{N}{2}$ two-qubit gates. Additionally, we utilize a one-hot encoding for the action. We provide a simple example to make it more clear.

Each action is represented by a list of 4 numbers where

\[
\mathcal{A}_t= 
\begin{cases}
	[N, N, e_3, e_4]_t, & \text{if gate = \texttt{ROT}}\\
	[e_1, e_2, N, N]_t, & \text{if gate = \texttt{CNOT}}
\end{cases}
\]
where $t\in T$ denotes the time step and $e_i\in\{ 0,1,\ldots,N-1\}$. In the case of one qubit parameterized gates defined through \texttt{ROT} encoding, say $[N, N, e_3, e_4]$ the \textit{qubit position} is defined by $e_3$ and the \textit{rotation axis} is encoded through by $e_4$. Meanwhile for two-qubit \texttt{CNOT} gates the encoding $[e_1, e_2, N, N]$ two-qubit that the \textit{control} is on $e_1$ qubit whereas the target is on $(e_1+e_2)$~$\pmod{N}$ qubit.

\paragraph{An example}\label{sec:example-binary-encoding-and-action}
Here, we give a simple example of the encoding and the action scheme for $T=3$. In the~\figref{fig:encoding-action-elaborate}, we take a three-qubit system. The empty circuit is denoted by a tensor of size $3\times(6\times3)$ tensor. 
\begin{figure}[ht]
	\centering
	\includegraphics[scale=0.8]{ 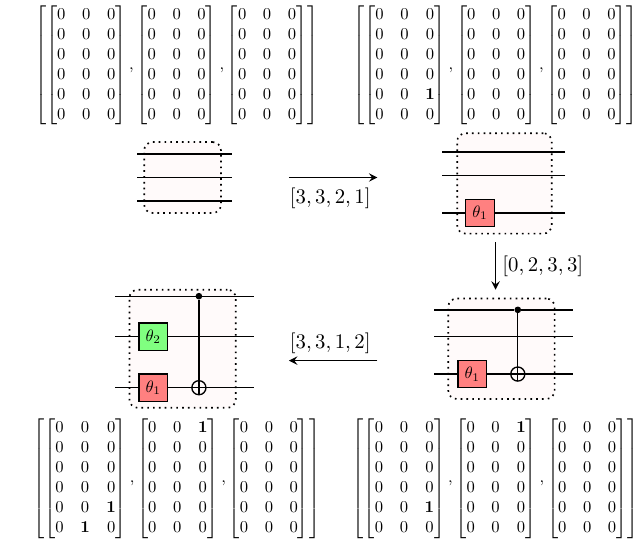}
	\caption{A toy example of action and the encoding scheme. The actions are represented with a list of four numbers. The first two elements in the digit carry the information about the controls and the target of a two-qubit gate that the RL-agent suggests adding in the next step to the ansatz. Meanwhile, the last two digits tell us about a one-qubit gate, on which qubit the gate should be added, and in which direction the rotation is.}
	\label{fig:encoding-action-elaborate}
\end{figure}
As we can see from the example, the first $(6\times3)$ matrix is filled when the \texttt{RY} gate is applied on the 3rd qubit. But the information is encoded in the next empty tensor for the next gate, which is \texttt{CNOT} with a control on the first and a target on 3rd qubit. This is because each $(6\times3)$ matrix encodes the complete information corresponding to a depth. That is why when the next \texttt{RZ} gate is applied on the second qubit, the information corresponding to it was encoded in the previous $(6\times3)$ matrix because it does not increase the depth of the circuit.

The encoding presented in the paper is an RL-agent-friendly representation of the action space. Many such encoding schemes can be adopted as a representation of the action space. Meanwhile, as we constrain the action space to only one-qubit rotations and \texttt{CNOT} gate, it would also be interesting to investigate the inclusion of other gate-sets in the action space.

\subsection{RL-reward}\label{sec:reward-rl-vqsd}
To guide the agent quickly towards the goal, we introduce a reward that is dense in time step $t$. The reward used in this work is given as
\begin{equation}
	R = \left\{\begin{array}{ll}
		+\mathcal{R} & \text{for } C_t(\vec{\theta})<\zeta+10^{-5},\\
		-\textrm{log}\left(C_t(\vec{\theta})-\zeta\right) & \text{for } C_t(\vec{\theta})>\zeta.
	\end{array}\right.\label{eq:log_reward}
\end{equation}
where the goal of the agent is to reach the minimum error for a predefined threshold $\zeta$, \ie the tolerance for cost function minimization. The $\zeta$ is a hyperparameter of the model. The cost function at each step $t$ is calculated for the ansatz which outputs a state $\rho_t(\vec{\theta})$ as
\begin{equation}
	C_t(\vec{\theta}) = \text{Tr}(\rho^2) - \text{Tr}(\rho_t(\vec{\theta})^2).
\end{equation}

\begin{figure}[ht]
	\centering
	\includegraphics[width=0.6\columnwidth]{ 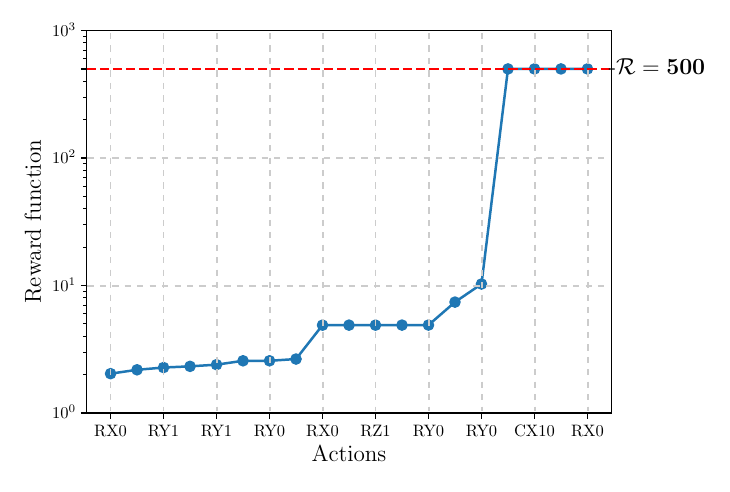}
	\caption{The variation of reward function with actions taken by the RL-agent. In this illustration, we use RL-VQSD (which is discussed in the upcoming section) to diagonalize a two-qubit random quantum state. The agent receives a high reward of value $\mathcal{R}=500$ if the RL-ansatz passes a predefined threshold (which is $\zeta=10^{-5}$) in any step of an episode. The total number of steps considered for this task is $20$. The labels in the x-axis represent each action in terms of a gate where for one qubit gate, the label is denoted as \texttt{Gj} where \texttt{G} is the rotation gate on \texttt{j}-th qubit. For the two-qubit gate, we use the notation \texttt{ CNOTij} where \texttt{CNOT} is the controlled-NOT gate with control on \texttt{i}-th and the target on the \texttt{j}-th qubit.}
	\label{fig:reward_w_action}
\end{figure}

{The motivation for choosing a dense reward function~\eqref{eq:log_reward} instead of a sparse reward (see \eqref{eq:old_reward}), which is utilized in finding the ground state of chemical Hamiltonian) stems from the fact that diagonalizing a quantum state is computationally more demanding than finding its minimum eigenvalue. Recent research has demonstrated that quantum gradient descent can determine the lowest eigenvalue with a number of elementary gates that scales polylogarithmically with the number of qubits \cite{liang2022quantum}. In contrast, state-of-the-art algorithms for diagonalizing an $n \times n$ matrix generally require $\mathcal{O}(n^3)$ operations~\cite{larose2019variational}. Given the complexity involved in state diagonalization, a dense reward function is favored in reinforcement learning tasks due to its capacity to provide continuous feedback and improve learning efficiency.}

Through the illustration in \figref{fig:reward_w_action}, we show how, for a successful episode, the reward function defined in \eqref{eq:log_reward} changes after the application of each action. The illustration shows the learning curve of the RL-ansatz where after taking $15$ actions, the cumulative reward reaches the maximum (i.e. $\mathcal{R}=500$), which means the RL-agent at this point is able to pass the predefined threshold (which $\zeta=10^{-5}$).

\subsection{RL-VQSD}
In this section, we utilize the previously described approaches, such as ansatz encoding, action formulation, and various approaches to reinforcement learning (see {Sec.~\ref{sec:model_free_rl}} for details
\begin{figure}[H]
	\centering
	\includegraphics[scale=0.6]{ 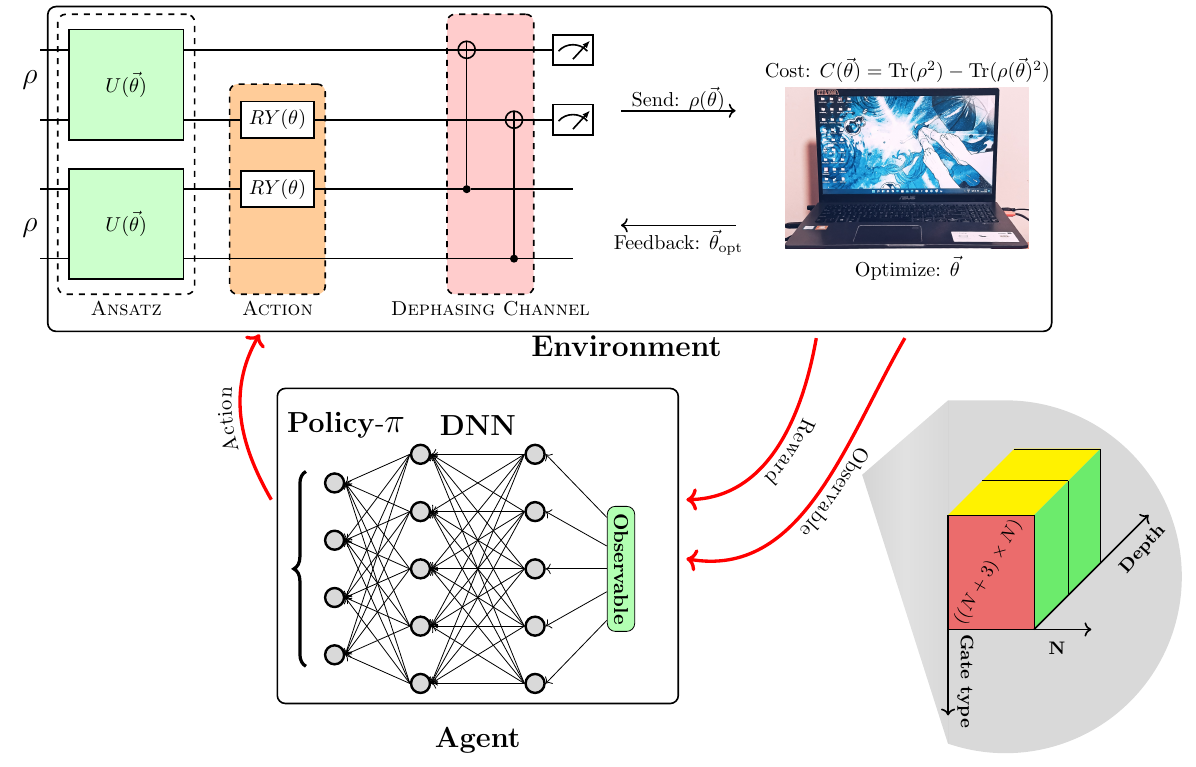}
	\caption{Demonstration of the RL-VQSD procedure: In this approach, the VQA task is linked to an environment where ansatz serves as the RL-state in reinforcement learning. The RL-agent receives a reward in the form of an optimized cost function from the environment, along with the current RL-state. Employing an $\epsilon$-greedy policy, the agent selects an action (a quantum gate), which updates the RL-state for the subsequent step. Using this updated RL-state, the VQA optimizes the cost function, generating a new reward for the agent. This iterative process continues until all episode steps are completed or the cost function reaches a predefined threshold. Throughout our study, we initiate RL-VQSD with an empty circuit, and at each step, the agent's action constructs the RL-ansatz, symbolized by $U(\vec{\alpha})=\mathbb{I}$.}
	\label{fig:rl_vqsd_diagram}
\end{figure}
on the various approaches) to enhance the performance of the VQSD algorithm. We primarily focus on the number of gates, the depth, and the accuracy of diagonalization of an ansatz. In~\figref{fig:rl_vqsd_diagram}, we illustrate the RL-VQSD algorithm.

To achieve the results, we use a Double Deep-Q network~\cite{mnih2013playing} (DDQN) for better stability with an $\epsilon$-greedy policy and the ADAM optimizer~\cite{kingma2014adam}. We start with the parameter specifications given in~\cite{ostaszewski2021reinforcement}, which uses the $n$-step DDQN algorithm with a discount factor of $\gamma = 0.88$ and an $\epsilon$-greedy policy for selecting random actions. The value of $\epsilon$ is gradually decreased from $1$ to a minimum value of $0.05$ by a factor of $0.99995$ at each step. The size of the memory replay buffer is set to $2 \times 10^4$, and the target network in the DDQN training is updated with every $500$ action. Following each training episode, we conduct a testing phase where the probability of selecting a random action is set to $0$, and the experience replay procedure is turned off. Experiences obtained during the testing phase are not added to the memory replay buffer.

To obtain a reward $\mathcal{R}$ for the circuit (\ie for each environmental state), an optimization subroutine needs to be applied to determine the values of the rotation gate angles. We use well-developed methods for continuous optimization, such as COBYLA~\cite{powell1994direct}, which has been shown to be among the best performing when there is no noise in the system. In this chapter, we set $\mathcal{R} = 500$.



\section{Diagonalizing quantum state with RL-VQSD}\label{sec:RL-VQSD-results}


In this section, we briefly investigate the performance of RL-VQSD on diagonalizing two-, three- and four-qubit quantum states. We start with diagonalizing two-qubit randomly sampled quantum states and show that the RL-VQSD outperforms the state-of-the-art VQSD method using the \texttt{RYRZRY CNOT} construction of LHEA. Due to the fast optimization time throughout the chapter, we utilize \texttt{COBYLA} optimizer with $400$, $500$, and $1000$ iterations for two-, three- and four-qubit states. In the case of diagonalizing three- and four-qubit states, we consider the reduced ground state of the six- and right-qubit Heisenberg model and show that the RL-agent proposed ansatz (so-called the RL-ansatz) provides us with a smaller quantum circuit with lesser gates, depth and better accuracy compared to state-of-the-art ansatz structures.

\subsection{Two-qubit states}\label{sec:2-qubit-diagonalization}


In this section, we aim to utilize a quantum computer to diagonalize (1) a fixed mixed quantum state and (2) 50 random quantum states to get the average eigenvalue approximation error and count the gates in RL-ansatz. We utilized the \texttt{random\_density\_matrix} of the module \texttt{quantum\_info} of \texttt{qiskit}~\cite{qiskit2019} to obtain the random density matrices. The states are sampled from the Haar measure. By (1), we argue that RL-VQSD can exactly diagonalize a quantum state. The results of (2) demonstrate that the average performance of RL-VQSD is better than state-of-the-art ansatz.

It can be seen from \ref{fig:2-qubit-eigenvalue-convergence} we show that the agent is able to propose an ansatz that provides us with the exact eigenvalues for a two-qubit random quantum state with 12 quantum logic gates containing 10 rotations and 2 \texttt{CNOT} gates. In \ref{fig:2_qubit_fixed_depth_ansatz}, we illustrate the quantum circuit proposed by the agent.

\begin{figure}[ht]
	\centering
	\begin{subfigure}{.5\textwidth}
		\centering
		\includegraphics[width=\linewidth]{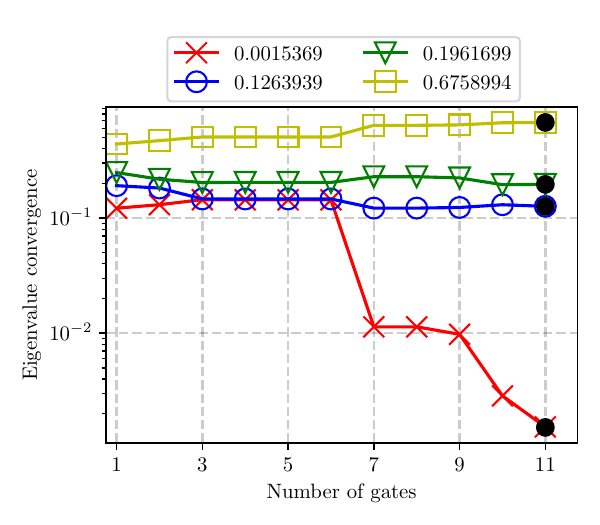}
		\caption{}
		\label{fig:2-qubit-eigenvalue-convergence}
	\end{subfigure}%
	\begin{subfigure}{.5\textwidth}
		\centering
		\includegraphics[width=\linewidth]{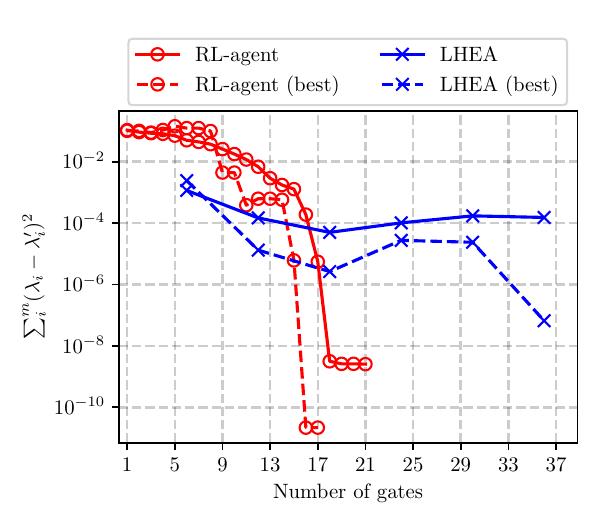}
		\caption{}
		\label{fig:2_qubit_average_error_50_states}
	\end{subfigure}
	\caption{Diagonalization of a two-qubit random density matrix of full rank. In (\subref{fig:2-qubit-eigenvalue-convergence}), we illustrate the convergence of eigenvalues of a mixed quantum state. Meanwhile, in  (\subref{fig:2_qubit_average_error_50_states}), we compare the performance of the RL-agent proposed ansatz with the LHEA. It can be seen that the RL-agent ansatz gives us a better approximation of the eigenvalues. Additionally, the RL-based methods can achieve the accuracy of the LHEA using the circuit with significantly reduced depth of the resulting circuit.}
	\label{fig:2_qubit_eig_conv_HEA_comparison}
\end{figure}

In \figref{fig:2_qubit_average_error_50_states}, we benchmark the performance of RL-ansatz against LHEA. In the illustration, we show that the agent not only gives us a small ansatz to diagonalize with a specific predefined threshold $\zeta$, but it also helps us achieve a very low error in eigenvalue estimation compared to LHEA.
\begin{figure}[ht]
	\centering	\includegraphics{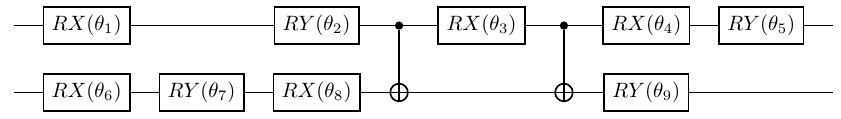}
	\caption{The ansatz proposed by RL-agent the state with eigenvalues convergence illustrated in \figref{fig:2-qubit-eigenvalue-convergence}.}
	\label{fig:2_qubit_fixed_depth_ansatz}
\end{figure}

Furthermore, we explore the possibility of utilizing the ansatz proposed by the RL agent, learning on a fixed quantum state, for the diagonalization of other states. We can confirm that this is indeed possible in the case of the two-qubit state. The corresponding results are presented
\begin{figure}[ht]
	\centering
	\includegraphics[width=0.85\textwidth]{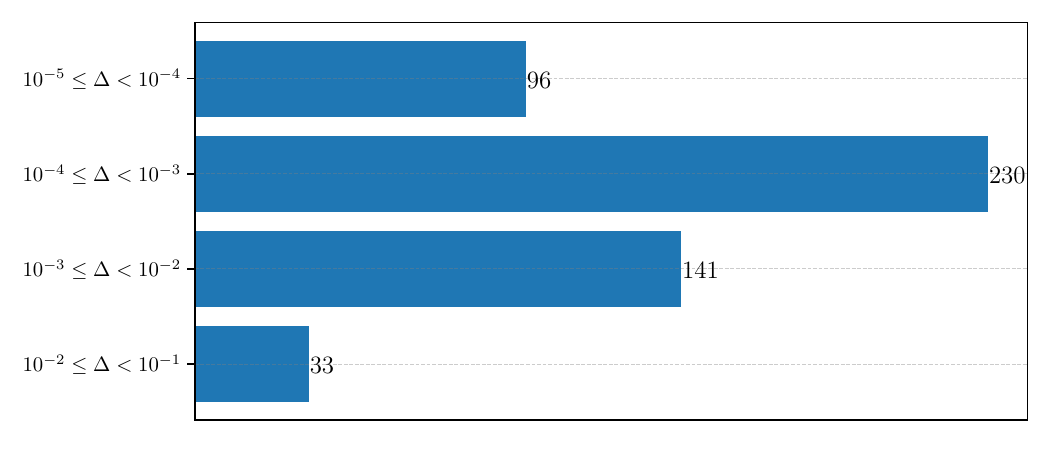}
	\caption{Statistics of error in eigenvalue estimation for 500 arbitrary quantum states. As an ansatz to diagonalize all the quantum states, we consider the fixed RL-ansatz in \figref{fig:2_qubit_fixed_depth_ansatz}.}
	\label{fig:2_qubit_error_with_fixed_ansatz}
\end{figure}
in~\figref{fig:2_qubit_error_with_fixed_ansatz}. One can argue that, in this case, the diagonalization task is relatively easy. However, one should note that the ansatz proposed by the RL agent gives an average error in eigenvalue estimation as in the case of using a standard approach based on LHEA (cf. \figref{fig:2_qubit_average_error_50_states}) while enabling the utilization of a shorter quantum circuit, and reducing the potential influence of errors.

\subsection{Three-qubit reduced Heisenberg model}\label{sec:3-qubit-diagonalization}
One of the important applications of VQSD is to study the entanglement in condensed matter systems~\cite{li2008entanglement}. Hence, in this experiment, to get a better understanding of the efficacy of our method in this regard, we consider a three-qubit reduced state of the ground state ($|\psi_{S_1, S_2}\rangle$) of the one-dimensional Heisenberg model defined on six-qubits which have the following form
\begin{equation}
	H=\sum_{j=1}^{2n} \vec{S}^{(j)}\cdot\vec{S}^{(j+1)},
\end{equation}
where $\vec{S}^{(j)} = \frac{1}{\sqrt{3}}\left( X^{(j)} \hat{x}+Y^{(j)} \hat{y}+Z^{(j)} \hat{z} \right)$ with periodic boundary condition $\vec{S}^{(2n+1)} = \vec{S}^{(1)}$, where $X$, $Y$, and $Z$ are the Pauli operators. To perform entanglement spectroscopy on the ground state of the 6-spin Heisenberg model (\ie $2n=6$), we diagonalize the reduced state $\rho_\textrm = \textrm{Tr}_{S_2}\left[ |\psi_{S_1,S_2}\rangle\langle\psi_{S_1,S_2}| \right]$. We consider choosing the threshold $\zeta = 10^{-4}$ for 500 iterations of the global COBYLA method.

\begin{figure}[ht]
	\centering \includegraphics[scale=0.7]{ 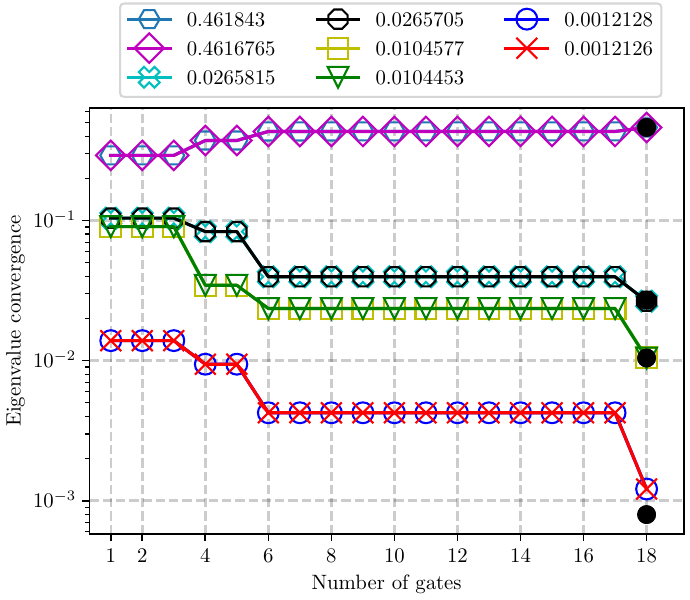}
	\caption{Convergence of the eigenvalues of three-qubit, reduced Heisenberg model. The labels on the top of the figure correspond to the eigenvalues. The black dots represent the true eigenvalues. Due to degeneracy in energy level, some dots are overlapped.
	}
	\label{fig:3_qubit_reduced_heisenberg_model}
\end{figure}

The results presented in \figref{fig:3_qubit_reduced_heisenberg_model} confirm that the RL-agent can learn to construct an ansatz that can find all the eigenvalues with good accuracy with a very small number of gates and depth. The \tikz\draw[draw, fill=black] (0,0) circle (.5ex); two-qubit the true eigenvalues. We can see that the ansatz takes 18 gates to give us 6 out of 8 exact eigenvalues of a three-qubit Heisenberg model. Additionally, the RL-ansatz finds the remaining two smallest eigenvalues with $\Delta_7 = \Delta_8 =  1.73\times 10^{-7}$ accuracy. In \figref{fig:3_qubit_reduced_heisenberg_model_ansatz} we present the RL-ansatz that contains 10 rotations and 8  CNOT gates, proposed by our methods.

\begin{figure}[ht]
	\centering	\includegraphics[width = \textwidth]{ 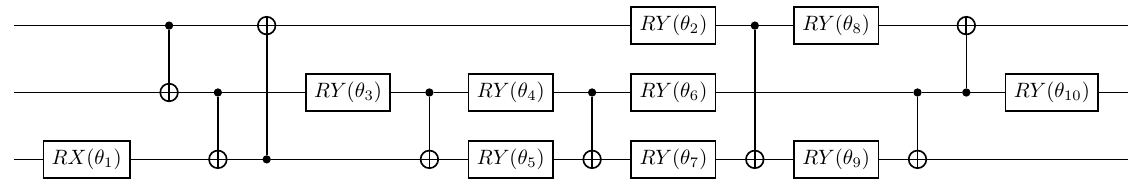}
	\caption{The ansatz proposed by the RL-agent for diagonalizing a state in the three-qubit reduced Heisenberg model. The circuit contains 10 rotations and 8  CNOT gates.}
	\label{fig:3_qubit_reduced_heisenberg_model_ansatz}
\end{figure}

It should be noted from circuits in \figref{fig:2_qubit_fixed_depth_ansatz} and in \figref{fig:3_qubit_reduced_heisenberg_model_ansatz} that the rotation in the Z direction, \ie \texttt{RZ} quantum logic gate, does not play a crucial part in the diagonalizing unitary. Thus, one might attempt to diagonalize a random quantum state of two- and three-qubits, excluding \texttt{RZ} rotation from the list of quantum gates. This gives us a hint concerning the action space that could be significantly reduced in these examples.

\paragraph{Performance of LHEA}
Here we compare the performance of the LHEA as shown in \figref{fig:LHEA-3-qubit-heisen-model} with the RL-ansatz~\ref{fig:3_qubit_reduced_heisenberg_model_ansatz} for diagonalizing a three-qubit reduced Heisenberg model. In the case of LHEA, each layer is composed of all-to-all connected two-qubit rotations containing a sandwiched \texttt{CNOT} gate in between \texttt{RZRYRZ} rotations of the form presented in \ref{fig:two_qubit_rot_decomp_layered_ansatz}(with 3 rotation gates). Hence, each layer contains 12 parameters. To achieve a cost function in the order of $10^{-4}$, it takes more than 60 parameters and at least 5 two-qubit gates. At the same time, the RL-ansatz achieves an accuracy lower than $10^{-4}$ in just $10$ rotations and $8$ two-qubit gates. We get a $6$ times improvement compared to LHEA in terms of parameters.

\begin{figure}[t]
	\centering
	\includegraphics[scale=0.8]{ 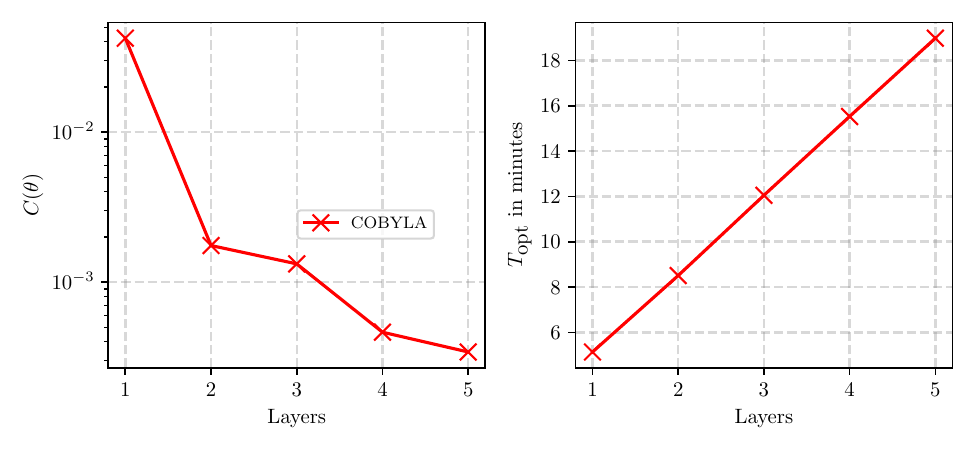}
	\caption{Performance of LHEA with the number of layers on (see LHS) and on the RHS, we show the time it takes to complete layers.}
	\label{fig:LHEA-3-qubit-heisen-model}
\end{figure}

\paragraph{RL-ansatz scaling with accuracy} Here, we briefly investigate the scaling of the minimum number of gates and the depth of the ansatz to diagonalize the three-qubit ground state of the reduced three-qubit Heisenberg model using the RL-VQSD.
\begin{figure}[ht]
	\centering
	\includegraphics[scale=0.7]{ 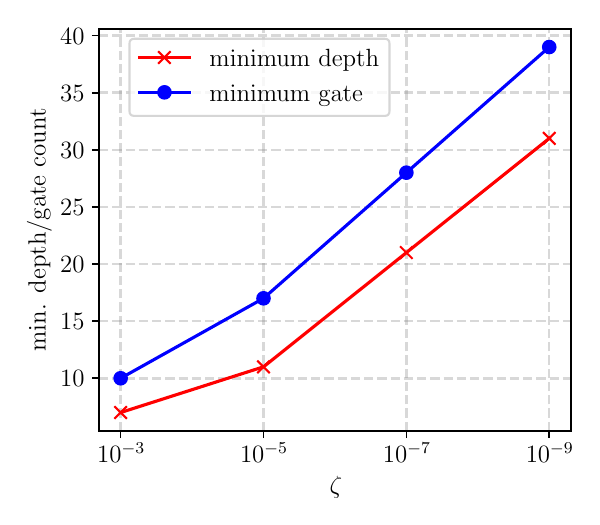}
	\caption{In this illustration, we show the minimum number of gates and depth required to diagonalize the ground state of the reduced 6-spin Heisenberg model using an RL-VQSD.}    \label{fig:3_qubit_heisen_diff_tolerance}
\end{figure}
The results are illustrated in the~\figref{fig:3_qubit_heisen_diff_tolerance} where we can see that the number of gates and the depth both scale linearly with the increase in $\zeta$, the pre-defined tolerance provided to the RL-agent.
\begin{table}[H]
	\centering
	\begin{tabular}{|c|c|c|c|c|c|c}
		\hline
		$\zeta$  & min. ROT & min.  CNOT & ave. depth & ave. num gate \\
		\hline
		$10^{-3}$ & 5        & 4             & 16.63      & 23.084         \\
		$10^{-5}$ & 8        & 5            & 19.67      & 27.863         \\
		$10^{-7}$ & 14       & 10           & 31.63      & 41.13          \\
		$10^{-9}$ & 18       & 15          & 36.8       & 46.15\\ 
		\hline
	\end{tabular}
	\caption{In this table, we summarize the scaling of various important RL-ansatz components averaged over 3000 episodes of RL-VQSD.}
	\label{tab:3-q-heisen-model-benchmark}
\end{table}

\subsection{Four-qubit reduced Heisenberg model}\label{sec:4-qubit-diagonalization}
We extend the results of the previous section for the ground state of the 8-spin Heisenberg model (\ie\; $2n = 8$). We diagonalize the four-qubit reduced state of the ground state of the 8-spin Heisenberg model. 
\begin{figure}[ht]
	\centering
	\includegraphics[width=\textwidth]{ 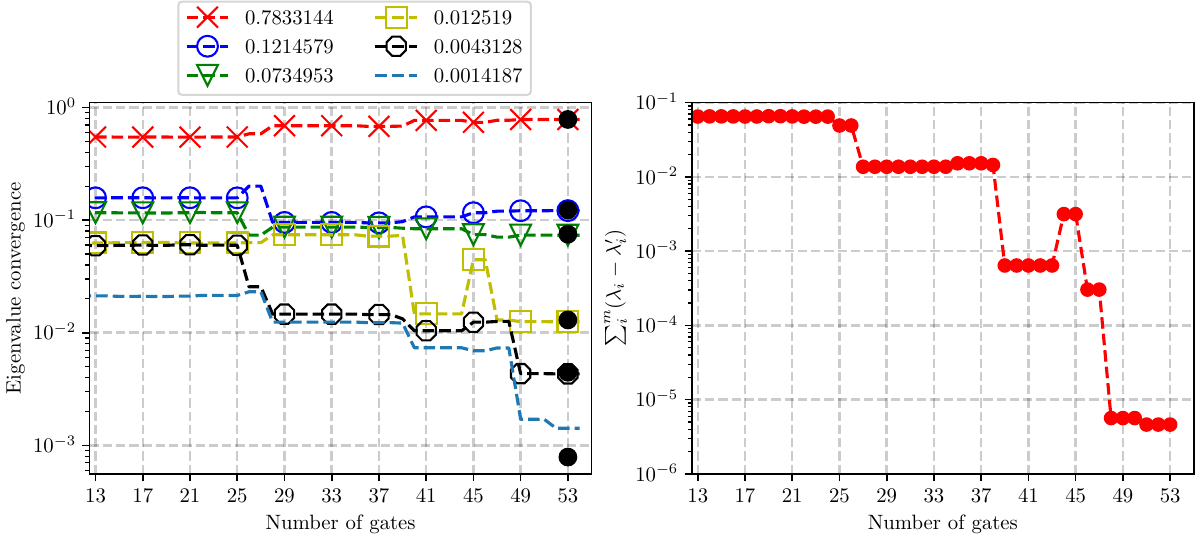}
	\caption{The convergence of individual (left panel) and the overall error (right panel) in eigenvalues for four-qubit reduced Heisenberg model. This provides a significant improvement in terms of gate count and depth compared to the result reported in~\cite{larose2019variational}.}
	\label{fig:my_label}
\end{figure}

It takes 53 gates to find the first 6 largest eigenvalues, with an error below $10^{-5}$. Out of $53$ gates, $16$ are \texttt{CNOT} gates, and the remaining are one-qubit rotations. We consider choosing the threshold $\zeta = 10^{-3}$ for 1000 iterations of the global COBYLA method.
\begin{table}[H]
	\centering
	\begin{tabular}{|c|c|c|c|}
		\hline
		Qubits & Minimum depth & Minimum \# of rotations & Minimum \# of  CNOT \\
		\hline
		2 & 8  & 9 & 2              \\
		3 & 12 & 10 & 8              \\
		4 & 33 & 28 & 16\\
		\hline
	\end{tabular}
	\caption{A summary of the minimum number of one- and two-qubit gates required in RL-ansatz to diagonalize two-, three- and four-qubit systems.}
	\label{tab:resource-count}
\end{table}

The summary of our results is provided in Table~\ref{tab:resource-count}. One can notice that there is a relation between the number of \texttt{CNOT}s and the dimension of the state that we want to diagonalize. The number of \texttt{CNOT}s grows exponentially with the number of qubits. As for the two-qubit case, we find all the eigenvalues with $10^{-10}$ error with just two \texttt{CNOT}s. Whereas for three-qubits, we are able to find the first 6 eigenvalues with an error below $10^{-8}$ but the smallest two eigenvalues we find with $1.73\times10^{-7}$ error with 8 \texttt{CNOT}s. Finally, for four-qubit, we find the first 6 eigenvalues with an error below $10^{-8}$ and the remaining eigenvalues with an error in the range $10^{-4}\leq\Delta\leq 10^{-6} $ with 16 \texttt{CNOT}s. This observation suggests that for a full-rank quantum state of $N\geq 3$, we require at least as many \texttt{CNOT}s as the rank of the quantum state to get a good approximation of the largest eigenvalues. It should be noted that to find the first 5 largest eigenvalues with error $10^{-5}$, the ansatz proposed by the RL-agent is of depth 18 and a total of 30 gates, among which 12 are \texttt{CNOT} gates and the remaining are rotations. This significantly improves the depth, and the gate count in the diagonalizing ansatz compared to the results in~\cite{cerezo2022variational} and \cite{larose2019variational}.

\subsection{Performance of random search}\label{sec:random-serch-performance}
\begin{figure}[H]
	\begin{subfigure}[b]{0.65\textwidth}
		\centering
		\includegraphics{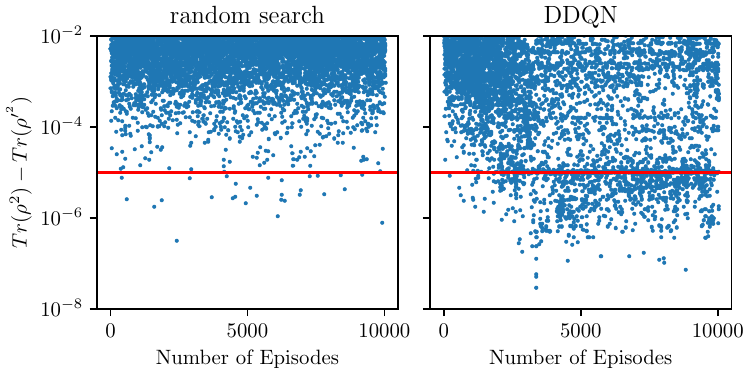}
		\caption{}
		\label{fig:random_search_vs_agent_2_qubit}
	\end{subfigure}
	\begin{subfigure}[b]{0.65\textwidth}
		\centering	\includegraphics{ 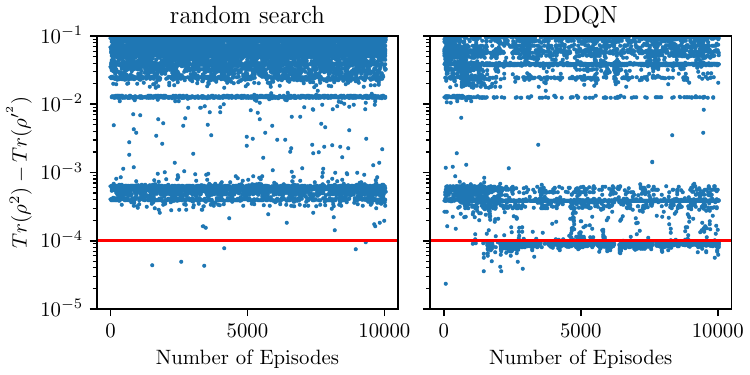}
		\caption{}
		\label{fig:random_search_vs_agent_3_qubit}
	\end{subfigure}
	\caption{The RL-agent can give us more frequent solutions, whereas the random search can hardly solve the problem. Comparison of accuracy obtained using random search and RL-based method. Illustration of $10^4$ episodes to solve the full rank random quantum state of (\subref{fig:random_search_vs_agent_2_qubit}) two-qubits and (\subref{fig:random_search_vs_agent_3_qubit}) three-qubits. The red line denotes the pre-defined tolerance for the approximation of the cost function.}
	\label{fig:random_search_vs_agent}
\end{figure}

To demonstrate the hardness of the variational diagonalization task, we utilize random search to find an efficient ansatz in this section. Unlike the previous examples where an RL-agent selects an action based on a policy, here, the action at each step is chosen randomly from a uniform distribution. 

In~\figref{fig:random_search_vs_agent} (in the first column), we show the results for random search to diagonalize a two- and three-qubit quantum state. It can be seen that the number of successful episodes (the episodes that pass the predefined tolerance of cost function) drastically reduces as we scale the number of qubits in the state. At the same time, the RL-agent (in the second column) provides us with a more consistent outcome.

For GPU/CPU specifications and the computational time of each episode of RL we refer the readers to the details in Table~\ref{tab:training-time-record}.

\section{Takeaways}
The goal of this chapter was to investigate the VQSD algorithm, to identify its weak points, and to provide a novel method that can significantly enhance this technique. To summarize, our result indicated that the following aspects are crucial for the efficiency of the VQSD technique.
\begin{itemize}
	\item \textbf{Choice of best fixed-structure ansatz} In the {Sec.~\ref{sec:peformance_of_lhea}} we briefly investigate the state-of-the-art VQSD algorithm. For the sake of comparison, we choose 6 different optimization methods and 3 different generic structures of hardware efficient ansatz (HEA). Our results first show that The \texttt{RYRZRY CNOT} of~\figref{fig:two_qubit_rot_decomp_layered_ansatz} overall outperforms the other structures of HEA in terms of accuracy and performs well with the scaling in the number of qubits in the diagonalizing state for all the optimizers.
	
	\item \textbf{Time vs. accuracy trade-off in classical optimizer} Once again in the {Sec.~\ref{sec:peformance_of_lhea}}, using the LHEA, and \texttt{RYRZRY CNOT} as ansatz we consider \texttt{COBYLA}, \texttt{L-BFGS-B}, \texttt{POWELL}, \texttt{Nelder-Mead}, \texttt{SLSQP} and \texttt{BFGS} as classical optimizer. And through~\figref{fig:optimizer_benchmark_vqsd} we see that \texttt{POWELL} gives us the best accuracy towards the diagonalizing unitary but it has a significant time overhead. Whereas \texttt{COBYLA} takes 100 times less time compared to \texttt{POWELL}, and it fails to return a good approximation to the diagonalizing unitary. Hence, there is a time vs. accuracy overhead in the choice of classical optimizers.
	Due to the impressive time efficiency, we utilize  \texttt{COBYLA} optimizer in RL-VQSD; because if we have $\mathcal{A}_N$ number of actions per episode and $N$ is the number of episodes, then we need $\mathcal{A}_N\times N$ query of the classical optimizer. Hence, to minimize the time spent on each query of the classical optimizer, it is ideal to use an optimizer that takes less time, pointing towards \texttt{COBYLA}. 
	
	\item \textbf{Tensor-based binary encoding scheme for quantum circuit} In \textbf{Section~\ref{sec:binary-encoding-scheme}} we present a qubit efficient encoding scheme for quantum circuits that scales polynomially with the number of qubits. Each depth of the circuit is encoded in a block of dimension $\left((N+3)\times N\right)$. We show that the encoding scheme provides more degree of freedom in qubit connectivity compared to previously proposed schemes, and the depth-based nature of the encoding ensures that it is more efficient. Each depth is filled up by the action scheme briefly described in \textbf{Section~\ref{sec:rl-action}}. For further clearance, please see the example in the same section.
	
	\item \textbf{RL-VQSD outperforms existing algorithms} Through the numerical simulations in \textbf{Section~\ref{sec:2-qubit-diagonalization}}, \textbf{\ref{sec:3-qubit-diagonalization}} and \textbf{\ref{sec:4-qubit-diagonalization}} we show that RL-VQSD outperforms the state-of-the-art VQSD methods and provides us with a diagonalizing unitary with a minimal number of parameters, depth and with very high accuracy. This claim becomes more prominent when we show that in~\figref{fig:LHEA-3-qubit-heisen-model}, the state-of-the-art VQSD method fails to reach the accuracy reached by RL-VQSD with the same (or even more) number of gates for the same classical optimizer. Later in~\figref{fig:3_qubit_heisen_diff_tolerance}, we show that using RL-VQSD, the number of gates and the depth of the diagonalizing unitary scales linearly with increasing accuracy.
	
	\item 
	\textbf{Random search is inefficient for diagonalizing task} In \textbf{Section~\ref{sec:random-serch-performance}} we replace the RL-agent by random search where the actions per step are chosen from a uniformly random distribution. Our investigation shows that (1) The random search fails to achieve a better accuracy compared to the RL-agent and (2) the number of solutions after the same number of episodes decreases drastically for the random search. This not only gives us an idea about the hardness of the diagonalizing problem but also provides insight into the efficiency of the novel RL-VQSD algorithm.
\end{itemize}


\chapter{Ansatze synthesis using curriculum reinforcement learning for variational quantum eigensolver}\label{ch:crlvqe}

This chapter presents a novel curriculum-based reinforcement learning (CRL) based quantum architecture search algorithm tailored to address the challenges inherent in deploying variational quantum algorithms in realistic noisy scenarios. This approach integrates three key elements: (i) the tensor-based ansatz encoding scheme presented in chapter~\ref{ch:vqsd_using_rl}, an \textit{illegal action} scheme to constrain the search in the action space, enabling efficient exploration of potential circuits, an episode-halting scheme guiding the agent towards discovering shorter circuits, and a variant of the simultaneous perturbation stochastic approximation algorithm, fostering faster convergence for optimization. Through a series of numerical experiments focusing on quantum chemistry problems, we showcase our methods' performance compared to existing quantum architecture search (QAS) algorithms in noiseless and noisy environments. Through the investigation, we show that our algorithm provides a more efficient ansatz than state-of-the-art algorithms. The influence of noise on the architecture search of ansatz is poorly understood. This chapter addresses this issue by showing that our CRL-based QAS algorithm can efficiently solve quantum chemistry problems under realistic quantum noises and constrained connectivity.

\section{Introduction}
The quantum phase estimation (QPE)~\cite{kitaev1995quantum,abrams1999quantum,abrams1997simulation} is a quantum algorithm introduced to extract the eigenvalues of a unitary operator utilizing the inverse quantum Fourier transform (IQFT) and phase kickback. It is shown that QPE can achieve exponential speedup in obtaining the eigeninformation of unitaries as long as the trial state is appropriately prepared. The promise of achieving quantum advantage is evident with QPE if a big enough fault-tolerant quantum computer is available. But its subroutine with IQFT requires a lot of resources, in terms of gates and qubits, for even relatively small quantum systems. Although recent developments in QPE are explored to minimize depth~\cite{ni2023low} and computational resources~\cite{kang2022optimized}, exploring the true potential of QPE is beyond the capabilities of present NISQ devices.

Keeping the hardware constraints in mind a hybrid quantum-classical algorithm, variational quantum eigensolver (VQE) is introduced~\cite{peruzzo2014variational,mcclean2016theory,romero2018strategies}. At its core, VQE utilizes both quantum devices and classical optimization techniques to find the ground state energy of a molecular Hamiltonian. Finding the ground state is a crucial problem in quantum chemistry and is essential for predicting chemical properties and reactions by understanding the electronic structure of molecules. This has applications from material science~\cite{lordi2021advances} to engineering~\cite{cao2019quantum}. In VQE a trial wave function or an ansatz using a parametrized quantum circuit (PQC), $U(\Vec{\theta})$, is prepared as follows
\begin{equation}
	\ket{\psi(\Vec{\theta})} = U(\Vec{\theta})\ket{\psi_0}.
\end{equation}
$\ket{\psi_0}$ is the initial state usually chosen as $\ket{0\ldots0}$. The $U(\Vec{\theta})$ can be decomposed into a series of parametrized and non-parametrized using~\eqref{eq:ansatz-main-equation}. If the electronic Hamiltonian of a molecule is defined as $H_\text{mol}$, then VQE aims to minimize the cost
\begin{equation}
	C(\Vec{\theta}) = \langle\psi(\Vec{\theta})|H_{\text{mol}}\ket{\psi(\Vec{\theta})},\label{eq:rl-vqe-cost}
\end{equation}
using a classical optimizer. The variational principle guarantees that $E_\text{ground} \leq C(\Vec{\theta})$ where $E_\text{ground}$ is the ground energy of $H_{\text{mol}}$. Meanwhile, $C(\Vec{\theta})$ is the energy expectation value of $H_{\text{mol}}$. The electronic Hamiltonian $H_{\text{mol}}$ is constructed (for an overview of the construction of molecular Hamiltonian see~\cite{meyer2002molecular, howard1970molecular}) with fermions and in order to evaluate the energy one can use indistinguishable fermions to distinguishable qubit mapping using the Jordan-Wigner~\cite{jordan1993paulische}, Parity~\cite{bravyi2017tapering} and Bravyi-Kiteav~\cite{bravyi2002fermionic,seeley2012bravyi} encodings. In recent years a Bravyi-Kitaev superfast~\cite{setia2018bravyi} and a qubit-efficient encoding~\cite{shee2022qubit} for Hamiltonian is introduced. For a comparison QPE requires $\mathcal{O}(1)$ repetitions with circuit depth scaling in precision $\mathcal{O}(\frac{1}{\epsilon})$ whereas VQE requires $\mathcal{O}(\frac{1}{\epsilon^2})$ shots with circuit depth scaling in precision $\mathcal{O}(1)$~\cite{wang2019accelerated}.

The performance of VQE can be influenced by the structure of the ansatz~\cite{grimsley2019adaptive,tang2021qubit,liu2022layer} due to the fact that the accuracy of the energy depends on the state manifold accessible by the PQC. So, finding new methods to construct PQC can lead to breakthroughs in VQAs for chemistry problems. In the {Sec.~\ref{sec:ansatz_introduced}}, we briefly discuss the various ansatz constructions depending on how much knowledge of the Hamiltonian we have at our disposal. The number of gates and depth increases exponentially in a problem-inspired PQC as we scale up the size of the molecule. Meanwhile, the hardware-efficient and problem-agnostic ansatz reflect trainability issues. To address these challenges, new methods have been introduced that draw on the insight and techniques of machine learning~\cite{benedetti2019parameterized,he2023gnn,romero2021variational}. In recent times, numerous explored reinforcement learning methods to find an efficient PQC for VQE~\cite{ostaszewski2021reinforcement,du2022quantum,fosel2021quantum} problem. 

In this chapter, we introduce a novel approach towards efficient ansatz construction using RL, which exploits a Double Deep-Q-Network~\cite{van2016deep} (see pseudocode~\textbf{\ref{alg:DDQN-pseucode}} for the pseudocode). Notably, our approach outperforms the existing state-of-the-art VQE algorithm in the noiseless scenario. We also benchmark the performance of our RL-based approach under quantum noise and show that even under hardware connectivity constraints and decoherence noise, the introduced approach shows impressive performance.

\subsection{Previous works}
As for the groundwork, we are going to mainly focus on the research~\cite{ostaszewski2021reinforcement} where the authors introduce an RL method for quantum circuit construction, utilizing the Double Deep-Q network (DDQN) with an $\epsilon$-greedy policy and an ADAM optimizer. They primarily focus on finding the ground state of the four- and six-qubit \texttt{LiH} molecule with bond distances 1.2, 2.2 and 3.4, respectively. In order to achieve a solution, they make use of a reward function of the following form.
\begin{figure}[th!]
	\centering
	\includegraphics{ 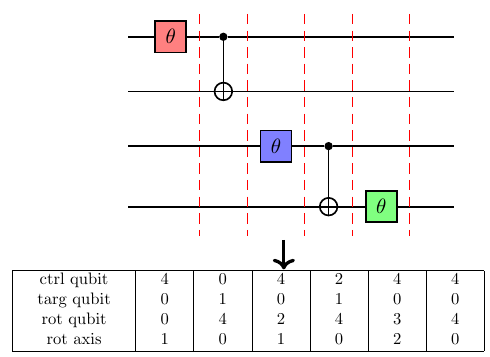}
	\caption{Example of state representation. In this example, the maximum length of the circuit L is set to six for four-qubit. Since we count qubits from 0, the lack of a particular gate at each layer is represented by the maximum number of qubits, i.e., in this case, it is 4, which can be seen in the last column where the \texttt{ctrl} and the \texttt{rot} qubit are both set to 4. Where \fcolorbox{black}{red!60}{\rule{0pt}{6pt}\rule{6pt}{0pt}}=\texttt{RY}, \fcolorbox{black}{green!60}{\rule{0pt}{6pt}\rule{6pt}{0pt}}=\texttt{RZ} and \fcolorbox{black}{blue!60}{\rule{0pt}{6pt}\rule{6pt}{0pt}}=\texttt{RX} gate.}
	\label{fig:old_encoding}
\end{figure}
\begin{equation}
	R
	= \begin{cases}
		5 &\text{if $C_t< \zeta$,}\\
		-5 &\text{if $t\geq L$ and $C_t\geq \zeta$,}\\
		\textrm{max}\left( \frac{C_{t-1}-C_t}{C_{t-1}-E_\text{min}},-1 \right) &\text{otherwise}.
	\end{cases}
	\label{eq:old_reward}
\end{equation}

The $C_t$ is calculated at each time step $t$ according to the~\eqref{eq:rl-vqe-cost}. Here, the main goal of the RL-agent is to achieve the $E_\text{min}$ within a predefined threshold $\epsilon$, where $L$ is the maximum number of circuit layers. To encode the quantum circuit the authors introduce an ordered list of layers that are composed of single quantum gates. As \texttt{CNOT}, \texttt{RX}, \texttt{RY} and \texttt{RZ} are considered as the building blocks for quantum circuits, the environment state is represented by a list that fully describes the circuit in terms of the \texttt{CNOT} and one qubit gates. The parameterized rotations are encoded using two integers: the first number indicates the registered qubit, and the second is the axis of rotation. Meanwhile, \texttt{CNOT} gate is also represented using two integers to indicate the position of control and target qubit.
To get rid of the continuous parameter of rotation angles, the estimated energy by the circuit is appended to the state representation. In~\figref{fig:old_encoding}, we illustrate the circuit encoding scheme the authors utilize to prepare the state. In both four- and six-qubit \texttt{LiH} problem using a {global strategy}\footnote{A global strategy corresponds to optimizing all the angles of the PQC after application of each action in terms of quantum gates to the quantum circuit.} the authors outperform the hardware efficient and the UCCSD ansatz in terms of depth and minimum gate count.

Meanwhile, in article~\cite{du2022quantum}, the authors introduce a quantum architecture search method to improve the learning performance of VQAs by enhancing trainability. To do so, the work considers a pool of all possible ansatz, say $\mathcal{P_\mathcal{A}}$ to build the ansatz for the VQA, where
\begin{equation}
	|\mathcal{P_\mathcal{A}}| = f(\mathcal{G}^{N\times L})
\end{equation}
where $\mathcal{G}$ is the set of different types of quantum gates, $N$ is the number of qubits and $L$ is the maximum depth of the circuit. To incorporate the realistic noisy scenario, we consider a quantum noise channel $\mathcal{C}_{a_i}$ for the $a_i$-th ansatz. Now, if the problem is defined through a Hamiltonian $H$, then the VQA objective is redefined as
\begin{equation}                (\vec{\theta}^*,a_i^*)=\text{arg}\underset{\vec{\theta},a_i\in\mathcal{P_\mathcal{A}}}{\text{min}}\mathcal{L}(\vec{\theta},a_i,H,\mathcal{C}_{a_i}),\label{eq:qas_loss_func}
\end{equation}
The authors make use of \textit{supernet} and \textit{weight sharing strategy} to get a good estimation of the~\eqref{eq:qas_loss_func} in a runtime comparable with the runtime of VQAs and minimal memory usage. On the one hand, the weight-sharing strategy helps to correlate the parameters among different analyses, helping reduce the parameter space during optimization. Meanwhile, the supernet is utilized as an indicator for the ansatz in the pool $\mathcal{P_\mathcal{A}}$ and parametrizes each ansatz in the pool.

Utilizing this approach, the authors tackle the ground state finding problem of the four-qubit $\texttt{H}_2$ molecule, showing that the energy converges to the true energy in a few iterations. However, the performance of this method is comparable to the conventional VQE in the noiseless scenario. In the case of a noisy scenario, the authors consider running in real superconducting quantum hardware, i.e., $\text{Ibmq}\_\text{ourense}$, and they show that their method gives a better approximation to the ground state compared to a hardware-efficient ansatz-driven VQE. The estimated ground energy of the method with $W=1$ and $W=5$ achieves $-0.93$ and $-1.05$Ha\footnote{1 Ha (Hartree) is 27.211 electron volt.}, respectively. Where $W$ is the number of supernets, the energies are better than the conventional VQE algorithm with a Hardware Efficient Ansatz (HEA), which achieves an energy of $-0.4$Ha.

A very recent development based on differentiable quantum architecture search (DQAS) to automate the design of PQCs is introduced in ref.~\cite{wu2023quantumdarts}. Before this work in ref.~\cite{zhang2022differentiable}, the authors used the differentiable search for QAS using Monte Carlo sampling to estimate the gradient by sampling multiple circuits at each epoch. This helps to get the approximation of the continuous distribution of quantum circuit architecture weights. But to achieve higher efficiency in the work ref.~\cite{wu2023quantumdarts} the authors use the Gumbel-Softmax~\cite{gumbel1948statistical,bengio2013estimating, jang2016categorical} technique to sample quantum circuits instead of Monte Carlo. Right after the sampling, the circuit architecture weights are updated by the gradient descent method. In the paper, the authors propose \textit{micro} and \textit{macro} search methods where the \textit{micro} search focuses on searching for the sub-circuits of an architecture, and these sub-structures are later stacked to form the whole circuit. Meanwhile, macro search directly searches for the whole circuit not focusing on the sub-structures. Using this differentiable search QAS method and $\{\texttt{RX},\texttt{RY},\texttt{RZ},\texttt{P}, \texttt{CNOT}\}$ as the candidate gate set, where $\texttt{P}$ is a phase-shift gate, they find the ground state of the $\texttt{H}_2$, $\texttt{LiH}$ and $\texttt{H}_2\texttt{O}$ problem of configuration given in {Tab.~\ref{tab:list_of_mol}}.
\begin{table}
	\small
	\begin{tabular}{c|c|c|c}
		Molecule & Fermion to qubit mapping & Configuration & Number of qubits \\
		\hline
		$H_2$     & Jordan-Wigner            & \makecell{$H$ ($0,0,-0.35$);\\ $H$ ($0,0,0.35$)}                               & 4                \\
		\hline
		$LiH$      & Parity                   & \makecell{$Li$ ($0,0,0$);\\ $H$ ($0,0,2.2$)}                                     & 4                \\
		\cline{2-4} 
		& Jordan-Wigner            & \makecell{$Li$ ($0,0,0$);\\ $H$ ($0,0,3.4$)}                                     & 6                \\
		\hline
		$H_2O$    & Jordan-Wigner            & \makecell{$H$ ($-0.021,-0.002,0$);\\ $O$ ($0.835,0.452,0$);\\ $H$ ($1.477,-0.273,0$)} & 8\\
		\hline
	\end{tabular}
	\caption{List of molecules considered for noisy and noiseless simulation.}
	\label{tab:list_of_mol}
\end{table}

\section{Groundwork}
We divide this section into subsections where we (1) briefly compare the tensor-based quantum circuit encoding that we use as the RL-state with previously proposed encoding schemes for the VQE task. (2) Next, we introduce a simple mechanism, namely \textit{illegal actions}, which helps narrow down the search space significantly, and finally (3) to facilitate the agent’s ability to discover more compact ansatz in
early successful episodes, we introduce a technique called \textit{random halting}. Before briefing the subroutines, we outline the RL-agent and environment specifications in the upcoming section.

\section{Agent and environment specification}
We utilize the double Deep-Q Network algorithm in the noiseless experiments with $\texttt{H}_2$, $\texttt{LiH}$ four-qubit. 
Meanwhile, the noisy simulations and the noiseless simulation of harder molecules such as six-qubit $\texttt{LiH}$ and eight-qubit $\texttt{H}_2\texttt{O}$, we make use of the Double Deep-Q Network step algorithm, where we utilize differing step sizes in the $n$-step trajectory roll-out update~\cite{sutton2018reinforcement}.
In these settings, we set the discount a factor set to $\gamma = 0.88$, and the probability of a random action being selected is set by an $\epsilon$-greedy policy, with $\epsilon$ decayed in each step by a factor of $0.99995$ from its initial value $\epsilon = 1$, down to a minimal value $\epsilon = 0.05$.
The memory replay buffer size is set to $20000$, and the target network in the DQN training procedure is updated after every $500$ action. After each training episode, we included a testing phase where the probability of random action is set to $\epsilon = 0$, and the experiences obtained during the testing phase are not included in the memory replay buffer. 
In the curriculum learning approach, the threshold is changed greedily after $500$ episodes for two-, three-, and four-qubit problems, whereas the threshold is changed after every $2000$ episode for six- and eight-qubit problems with an amortization radius of $0.0001$. After $50$ successfully solved episodes, the amortization radius is decreased by $0.00001$. The initial threshold value is set to $\varepsilon = 0.005$. Simulations of quantum circuits were performed using the Qulacs library~\cite{suzuki2021qulacs}. 
The hyperparameters were selected through coarse grain search. The employed network is a fully connected network with $5$ hidden layers with $1000$ neurons each for the four-qubit case, $2000$ neurons each for the six-qubit case, and $5000$ neurons each for eight-qubit. The maximum number of gates is set to $40$ for four-qubit, $70$ for six-qubit and $250$ for eight-qubit.

In the case of noiseless scenario we used the global strategy with COBYLA optimizer whereas for the noiseless case we use the same strategy with the introduced multistage Adam-SPSA optimizer with Pauli transfer matrix (PTM)~\cite{caro2022learning} formalism.  

\subsection{The tensor-based vs integer encoding}\label{sec:encoding_comparison}

Recalling the tensor-based encoding (TBE) presented in the {Sec.~\ref{sec:binary-encoding-scheme}} where we encode the quantum circuit depth-wise. Each depth is encoded into a 2-D grid of size $[T\times \left( (N+3)\times N\right)]$ where $T$ is predefined as the maximum depth of the circuit and $N$ is the number of qubits. We refer to the {Sec.~\ref{sec:binary-encoding-scheme}} for an elaboration of the encoding scheme.  

\begin{figure}[H]
	\centering
	\includegraphics{ 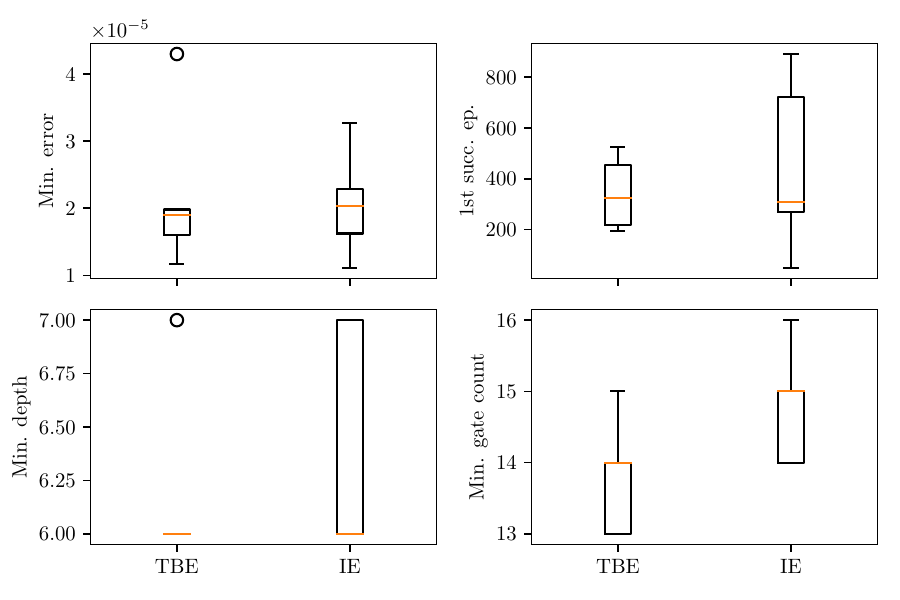}
	\caption{{The tensor-based encoding (TBE) outperforms the integer encoding (IE) and returns a lower depth and smaller number of gates circuit with lower error}. To conduct this experiment, we consider \texttt{LiH} molecules with a bond distance of 3.4 with parity encoding. The results are averaged over 5 seeds, and for each seed, we initialize the double deep-Q network with different values. For comparison, we consider the minimum error (Min. error), the first successful episode (1st succ. ep.), the minimum depth (Min. depth) and the minimum number of gates (Min. gate count) for both the TBE and IE. Additionally, the TBE makes the agent more stable compared to IE by preventing the spread of the deviation from the median.}
	\label{fig:encoding_comparison_LiH_4q}
\end{figure}

In one of the very first works in~\cite{ostaszewski2021reinforcement}, the authors introduce an integer-based encoding scheme where each block of the RL-state carries information about each gate applied. In order to show improvement over this state-of-the-art encoding scheme, here we compare the performance of the integer encoding (IE) with the encoding presented in this thesis, namely TBE. In a nutshell, the investigation shows that TBE outperforms IE in all aspects. To do this experiment, we consider the four-qubit LiH molecule with parity (fermion to qubit) encoding. For a detailed geometry of the configuration of LiH, see {Tab.~\ref{tab:list_of_mol}}.

In the~\figref{fig:encoding_comparison_LiH_4q} we illustrate the results and as a measure of performance, we consider the minimum error in energy (\textit{min error}), the first successful episode (\textit{1st succ. ep.}), the minimum depth of the ansatz (\textit{min. depth}), and the minimum number of gates (\textit{min. gate count}). The first thing to notice in~\figref{fig:encoding_comparison_LiH_4q} is that the TBE is more stable in IE, which can be seen through the spread of the quartiles. 

Another remarkable aspect of the TBE is that it gives the minimum error in the ground state energy lower than the IE with a smaller number of gates. This is ideal for the NISQ era to mitigate the negative impact of gate errors and decoherence effectively. It is important to ensure the circuits are both gate-efficient and have minimal depth.
\subsection{Illegal actions: The reduction of search space}\label{sec:ill-actions}
The QAS algorithms present a challenging combinatorial problem, which is characterized by an extensive search space. In order to narrow down this search space proves advantageous for the RL-agent. This not only helps the agent to discover a quantum circuit with diverse structures but also improves the learning by the agent.
\begin{figure}[H]
	\centering
	\includegraphics{ 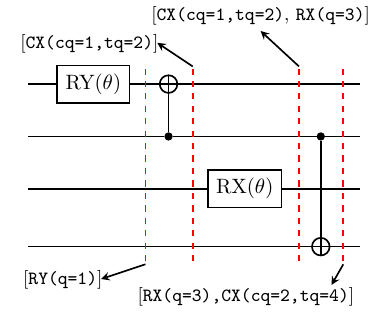}
	\caption{{The illegal action scheme on four-qubits}. Here, \texttt{cq} and \texttt{tq} represent the control and the target qubit. This approach encourages the RL-agent to avoid selecting the same action in two consecutive steps, effectively reducing the action space.}
	\label{fig:illegal_action_figure}
\end{figure}

Therefore, we present a straightforward approach referred to as \textit{illegal actions} to greatly reduce the search space. At the core of the technique, it leverages the inherent property of quantum gates being unitary, which results in the cancellation of two similar gates when applied to the same qubit. We employ this mechanism to reduce the search space for the RL-agent. In~\figref{fig:illegal_action_figure} illustrate the \texttt{illegal action} mechanism on a four-qubit system, and an elaborated discussion of the technique is provided in {Appendix~\ref{apndix:illegal_action_elaboration}}. Meanwhile, in the {Appendix~\ref{app:ill_action_code}}, we provide a code that is used to implement the {illegal actions} technique.

Although {illegal action} mechanism is straightforward to implement for any ansatz construction, one can consider more complex pruning of quantum circuits as provided in~\cite{fosel2021quantum}.

\subsection{Investigation of reward function}\label{sec:reward_comparison}
A well-defined reward function can accelerate the rate of convergence of an RL-agent towards the target. There are many ways to formulate a reward based on the task under consideration. In previous work~\cite{ostaszewski2021reinforcement}, for solving chemistry problems using Deep neural networks, a reward function of type~\ref{eq:old_reward} is considered, which is quite sparse in nature. 

In this section, we keenly investigate and compare the reward presented through~\eqref{eq:old_reward} and the reward that is used in {Chapter~\ref{ch:vqsd_using_rl}} of the form
\begin{equation}
	R(\mathcal{R}) = \left\{\begin{array}{ll}
		+\mathcal{R} & \text{for } C_t(\vec{\theta})<\zeta+10^{-5}\\
		-\textrm{log}\left(C_t(\vec{\theta})-\zeta\right) & \text{for } C_t(\vec{\theta})>\zeta
	\end{array}\right.
\end{equation}
where the $\mathcal{R}$ is a positive large number. For the sake of investigation, we compare the performance of the reward function~\ref{eq:old_reward} with $R(\mathcal{R}=0)$, $R(\mathcal{R}=50)$, $R(\mathcal{R}=100)$, $R(\mathcal{R}=500)$ and $R(\mathcal{R}=1000)$.
\begin{figure}[H]
	\centering
	\includegraphics[scale=0.5]{ 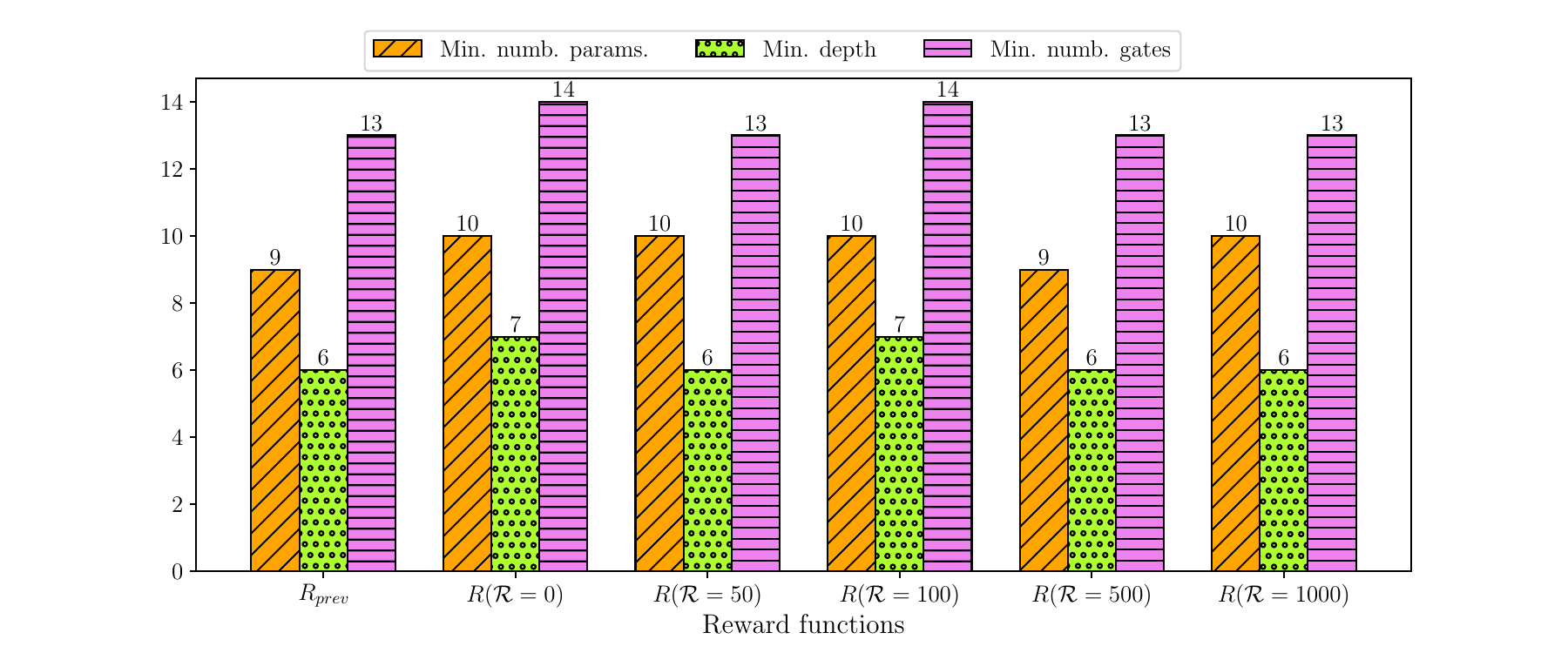}
	\caption{{Comparison of different reward functions where using $R_\textbf{prev.}$ and $R(\mathcal{R}=500)$ we get the minimum number of parameters (Min. numb. params.), minimum depth (Min. depth) and number of gates (Min. numb. gates) for four-qubits \texttt{LiH} molecule}. We evaluate the models based on 5 different seeds where in each seed the neural network is initialized with different input. 
		It should be noted that both the rewards i.e. $R_\textbf{prev.}$ and $R(\mathcal{R}=500)$ give us the same averaged error ($9.99\times10^{-4}$ for $R_\textbf{prev.}$  and $1.06\times10^{-3}$ for $R(\mathcal{R}=500)$) over the $5$ seeds.} 
\label{fig:reward-comparison}
\end{figure}

It can be seen from~\figref{fig:reward-comparison} that the performance of the RL-agent with $R_\textrm{prev.}$ is similar to the performance with $R(\mathcal{R}=500)$ and the reward~\eqref{eq:log_reward} with $\mathcal{R}=500$ outperforms the other variants. An in-depth investigation of the $R_\textrm{prev.}$ and $R(\mathcal{R}=500)$ unveils that the average number of successful episodes over $5$ different seeds is $16567$ for $R_\textrm{prev.}$ and for $R(\mathcal{R}=500)$ it is just $2301$ but the later reward function helps us to achieve the first successful episode faster. For $R_\textrm{prev.}$, the first successful episode on average appears at episode $343$ whereas, for $R(\mathcal{R}=500)$, it is at $181$-th episode.

In this chapter, we primarily focus on finding a very compact ansatz. Hence it is significant to have a higher number of successful episodes because it will yield a greater array of ansatz options for our investigation and to pick the best one among them. That's why we decided on utilizing the $R_\textrm{prev.}$ instead of $R(\mathcal{R}=500)$.

In the upcoming section, we introduce a straightforward technique aimed at enabling the RL-agent to accommodate shorter-length episodes in the first few successful episodes.

\subsection{Random halting: quickly discovering compact ansatz}\label{sec:random_halting_performance}
In the case of the previous works, such as in~\cite{ostaszewski2021reinforcement}, a full-length episode can be decomposed into a constant number of time steps $T_s$. Each time when noise is applied to a quantum circuit, a completely positive trace preserving (CPTP) channel is applied to the circuit which not only reduces the performance of the circuit to achieve a task but increases the computation time by many times compared to the noiseless scenario.
\begin{figure}[H]
\centering
\includegraphics[scale=0.7]{ 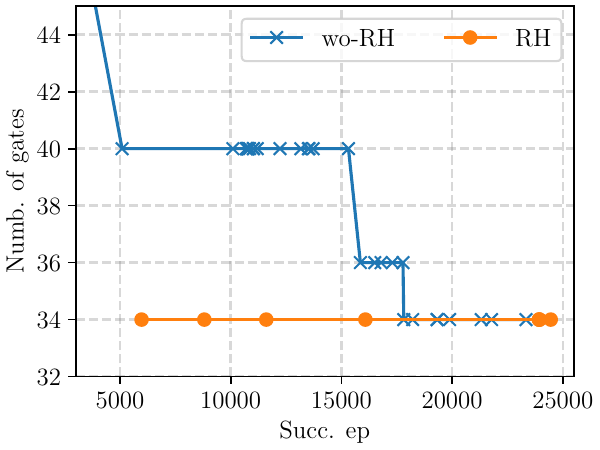}
\caption{{The \textit{random halting} (RH) gives a more compact ansatz compared to without RH settings in very early successful episodes (Succ. ep) with six-qubit \texttt{LiH} molecule}. It can be seen with RH in around $5,000$ episodes. The minimum number of gates (Numb. of gates) we require to solve the six-qubit \texttt{LiH} problem reaches $34$. Meanwhile, reaching the same number of gates without RH settings requires around $17,000$ episodes. Hence, RH helps the RL-agent learn $3$ times faster, the optimal quantum circuit compared to without RH settings.}
\label{fig:random_halting_effect}
\end{figure}

Hence, in a realistic scenario, in the midst of quantum noise, the RL-agent begins by proposing a lengthy ansatz in terms of the number of gates and depth that achieves chemical accuracy in the initial few successful episodes. However, over thousands of subsequent episodes, it gradually shifts towards a shorter ansatz. This situation is suboptimal due to the significantly extended duration of noisy simulation.

To address this challenge, we propose a method referred to as \textit{random halting} (RH), where the value of $T_s$ is no longer a constant parameter. Instead, it varies from one episode to another according to a particular probability distribution that is dependent on the number of qubits. To be more precise, we sample the episode-specific number of step $s$, denoted as $T_s$, from the following negative-binomial distribution:
\begin{equation}
T_s\sim\binom{n_f+n_s-1}{n_f}p^{n_f}(1-p)^{n_s},\label{eq:negative_binom_dist}
\end{equation}
In this context, $n_s$ is the count of successes. Meanwhile, $n_f$ denotes the count of failures. The sum of successes and failures determines the total number of trials, represented as $n_f + n_s$, and $p$ signifies the probability associated with each success.

The primary motivation for incorporating RH into the algorithm is to empower the RL agent to accommodate shorter episode lengths. This, in turn, enhances the agent's capability to uncover more concise ansatz in the early stages of successful episodes, even if it occasionally delays achieving the first successful episode. We observe in~\figref{fig:random_halting_effect} that within approximately $5,000$ episodes, the minimum number of gates needed to solve the six-qubit LiH problem (see {Tab.~\ref{tab:list_of_mol}} for details of the molecule geometry) decreases to 34. In contrast, reaching the same gate count without RH settings necessitates approximately $\texttt{3}\times$ more episodes.

\subsection{Multistage ADAM-SPSA algorithm}
In the case of VQE, the budget for measurement samples is restricted. To exhibit robustness against finite sampling noise, several versions of simultaneous perturbation stochastic approximation (SPSA) are utilized~\cite{cade2020strategies,bonet2023performance}. Among these variants, multi-stage SPSA adjusts the decaying parameters while tuning the permitted measurement sample budget between stages. Moreover, incorporating a moment adaptation subroutine from classical machine learning, like Adam~\cite{kingma2014adam}, alongside standard gradient descent helps us enhance the robustness and accelerates the convergence of the algorithm.
\begin{figure}[H]
\centering
\includegraphics[scale=0.1]{ 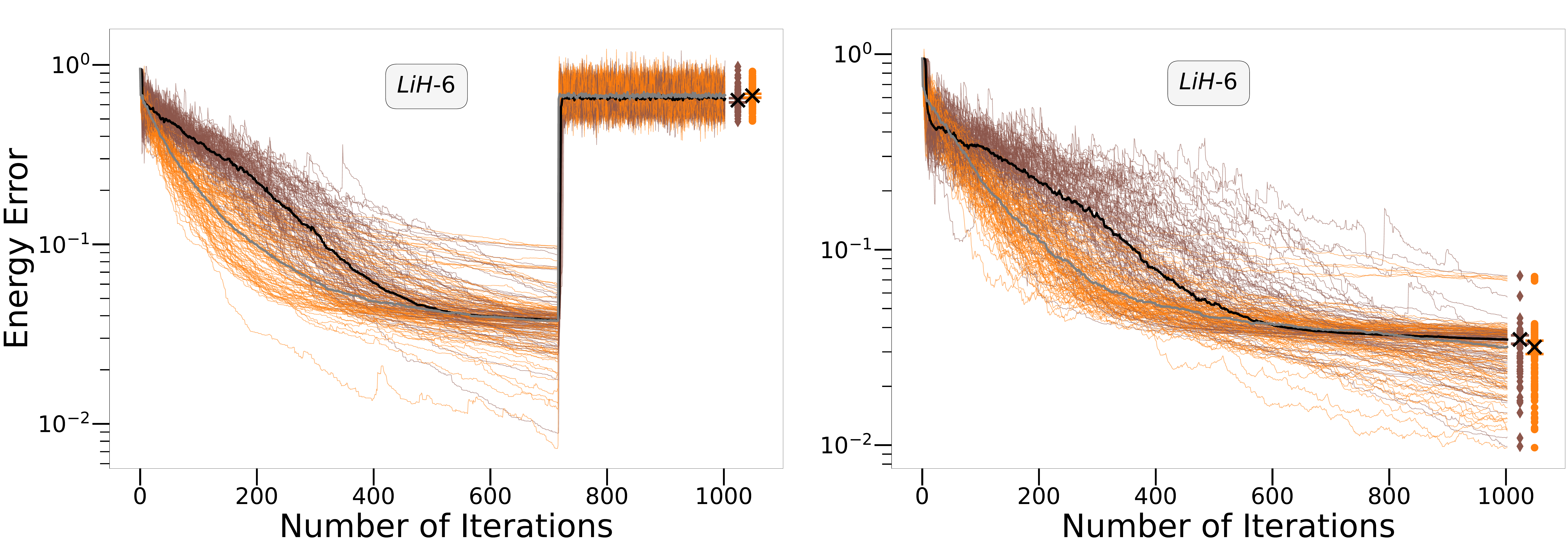}
\caption{Optimization traces of the 3-stage sampling strategy of SPSA (brown and black) and Adam-SPSA (orange and grey) on the six-qubit \texttt{LiH} (\texttt{LiH}-6) molecules using the hyperparameters outlined in Appendix~\ref{algo:adamspsa}.
	The individual traces are represented by thin lines, while the thick line on top indicates the median of $100$ independent runs. The left and right panels showcase the resetting and continuous evolution of SPSA (Adam-SPSA) hyperparameters, respectively.}
\label{fig:3stage_performance}
\end{figure}
As we discussed in the previous section, noise increases the computational time per episode, and it is very crucial to enhance the robustness and faster convergence rates of VQAs under realistic scenarios. For this reason, we leverage a 3-stage Adam-SPSA (whose pseudocode is provided in the Appendix~\ref{algo:adamspsa}). For a rigorous illustration of the 3-stage Adam-SPSA we investigate six-qubit \texttt{LiH} (\texttt{LiH}-6) molecule where the hyperparameters of the algorithm are set according to the Tab.~\ref{tab:spsa_hyperparams} in the Appendix~\ref{algo:adamspsa}. We see that in $3$-stage Adam-SPSA, unlike the vanilla SPSA without Adam momentum, the convergence towards the minima is qualitatively much faster which is qualitatively shown in~\figref{fig:3stage_performance}.

Utilizing these insights from our analysis of various SPSA variants, we employ $1$- and $3$-stage Adam-SPSA in our noisy experiments. 
This helps cut down the total number of function evaluations by half, thereby doubling the speed of our RL training. 
This improvement at the algorithm level helped us simulate noisy systems that suffer from computational complexity and large run times.

\subsection{Pauli-transfer matrix formalism on GPU}
Restating the fact that QAS demands a significant number of noisy function assessments unless a training-free approach is adopted. However, executing the steps poses enormous challenges within state-of-the-art simulation framework. The noisy simulation process not only encounter difficulty due to the curse of dimensionality related to dense matrix operations but also due to the exponential increase in the number of noise channels and their corresponding Kraus operators.

To address this challenge, a Pauli-transfer matrix (PTM) formalism is utilized, enabling the precomputation of noise channel fusion with respective gates offline. This eliminates the need for recalculations at each step. Alongside PTM formalism we integrate GPU computation along with just-in-time (JIT) compiled functions in JAX~\cite{jax2018github}, resulting in up to a $6\times$ enhancement in RL-agent training efficiency while simulating noisy quantum circuits.

\section{Curriculum reinforcement learning}
The moving threshold technique (see Fig.~\ref{fig:amortization}) is a feedback-driven curriculum learning method introduced in \cite{ostaszewski2021reinforcement}. 
{During the learning process, the agent pursues a parameter $\xi_2$ that marks the lowest energy known by the agent so far and updates a threshold parameter with respect to this parameter based on some rules. In the beginning, the $\xi_2$ parameter is set to a hyperparameter $\xi_1$. If the agent finds an energy value lower than the current one, it updates $\xi_2$ to this new energy value. Another hyperparameter ``fake minimum energy" $\mu$, a proxy to the lower bound of attainable ground state energy, is set as a target for the agent.\footnote{One can set the target of the agent to such a value for VQE because, from Rayleigh's variational principle, the agent theoretically can never attain energy below the true ground state energy.}
We compute this proxy by taking the summation of absolute values of Pauli string coefficients stemming from the Hamiltonian.

In the absence of amortization, the algorithm shifts the threshold to $|\mu - \xi_2|$ for the new $\xi_2$. In the presence of amortization, however, it adds a parameter to that threshold as $|\mu - \xi_2| + \delta$, where $\delta$ is the amortization hyperparameter. In the meantime, the agent continues its exploration with subsequent actions and episodes and records the number of successful actions. Here, there are two rules at play. The first rule greedily shifts the threshold to $|\mu - \xi_2|$ after $G$ episodes. Here $G$ is a hyperparameter as well. The second rule slowly decreases the threshold parameter each time there is a successful episode by subtracting a factor of $\delta / \kappa$. Here $\kappa$ is the radius of shifts, also a hyperparameter. Upon setting the threshold to $|\mu - \xi_2|$,  if the agent fails to improve the energy value in consecutive episodes, the threshold is increased back to $|\mu - \xi_2| + \delta$, as demonstrated in~\figref{fig:amortization}. This way, the agent is given an opportunity to trace its steps back if it was stuck in a local minimum.}
\begin{figure}[H]
\centering
\includegraphics[width=.4\textwidth]{ 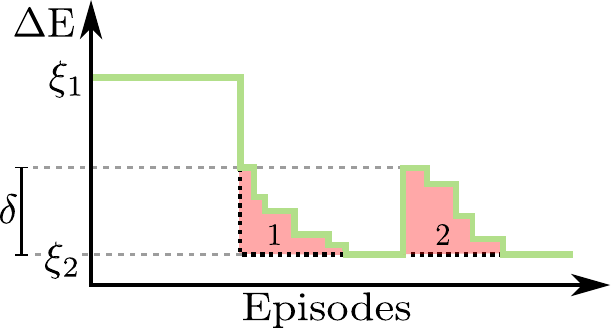}
\caption{\small{Demonstration of the feedback-driven (green) process, depicting two amortization occurrences (pink), $\delta$. The initial occurrence adjusts the threshold from $\xi_1$ to $\xi_2$, signifying the improvement. The subsequent event occurs when the agent fails to surpass $\xi_2$ or the gain is marginal. This prompts a sudden threshold increase due to amortization reset. Note that the final threshold, after the second amortization, may be less than $\xi_2$.
} } 
\label{fig:amortization}
\end{figure}
Notably, this method does not require any prior knowledge regarding the true value of the ground state energy and does not impose any specific constraints on the initial threshold value, unlike existing QAS methods in the literature. 

\section{Results}

In this section, we in detail present the results of finding the ground state of $\texttt{H}_2$, $\texttt{LiH}$, and $\texttt{H}_2\texttt{O}$ molecules of two-, three-, four- and eight-qubit. The structure of the molecules is provided in Tab.~\ref{tab:list_of_mol}. We initiate the section by simulating the molecules in a noiseless scenario and comparing its results with the state-of-the-art QAS algorithms. Through the rigorous comparison, we show that our CRLVQE algorithm outperforms the state-of-the-art and the existing learning-based QAS algorithm in tackling the same optimization task. Later on, we run our algorithm to find the ground in a realistic noisy scenario obtained from IBM devices such as \texttt{ibmq\_mumbai} and \texttt{ibmq\_ourense}.

\subsection{Noiseless case}
This section primarily focuses on the noiseless simulation of the molecules listed in the {Tab.~\ref{tab:list_of_mol}}. We compare the performance of our algorithm to the existing algorithms such as \textit{RL-VQE} introduced in~\cite{ostaszewski2021reinforcement} and \textit{quantumDARTS} described in~\cite{wu2023quantumdarts} while showing that our algorithm outperforms them.

To obtain the results, we consider the six-qubit $\texttt{LiH}$ and right-qubit $\texttt{H}_2\texttt{O}$ molecule. The detailed configurations are in the {Tab.~\ref{tab:list_of_mol}}. The results are summarized in {Tab.~\ref{tab:noiseless_rl_algorithm_compatison}}. It can be seen from the table that our algorithm outperforms the UCCSD, RL-VQE and the quantumDARTS ansatz and provides a quantum circuit with a smaller number of gates ($N_G$) and parameters ($N_P$) for six-qubit $\texttt{LiH}$ but the UCCSD ansatz provides smaller error in the ground state with the trade-off having $18$ times more gates, which is really costly.

\begin{figure}[h!]
	\centering
	\includegraphics[width=0.7\linewidth]{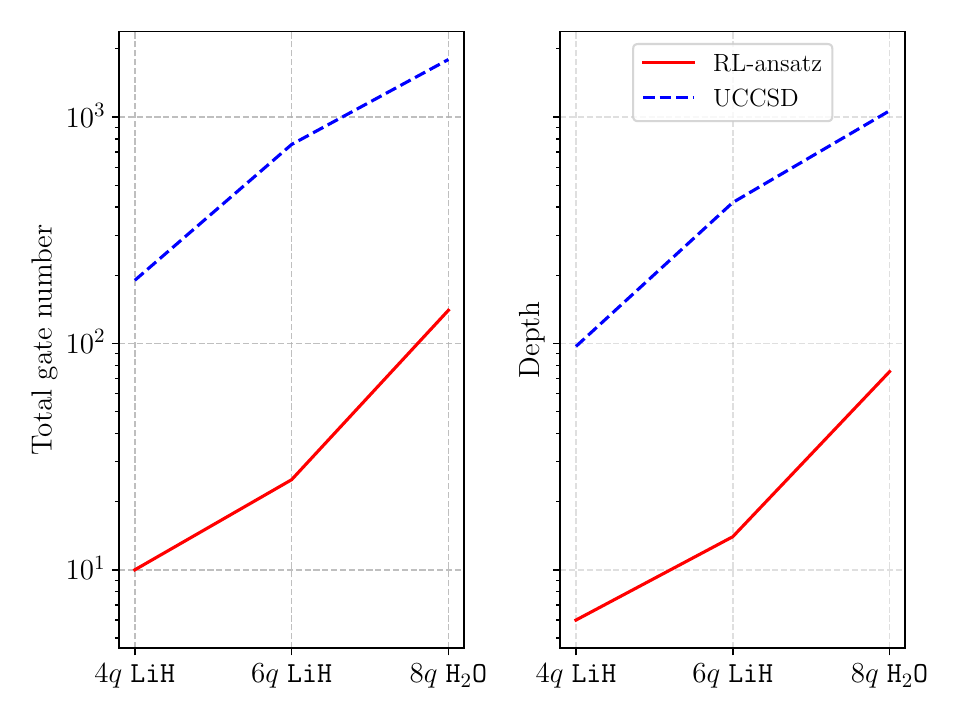}
	\caption{UCCSD ansatz requires approximately $20$ times more number of gates and circuit depth to find the ground state of chemical Hamiltonians compared to the CRLVQE proposed ansatz. Here the $x$-ticks define the molecules and respective number of qubits.}
	\label{fig:complexity_analysis}
\end{figure}
On the other hand, for eight-qubit $\texttt{H}_2\texttt{O}$, we compare our algorithm's performance with the quantumDARTS algorithm and show that we achieve $1.68$ times lower error in ground energy with $79$ fewer quantum gates and with $4.31$ times less parameters.

{%
\paragraph{Complexity of anstaz}
To provide a perspective from the complexity theory, in this section we compare the scaling of a decomposed UCCSD and the RL ansatz with the increase in number of qubits. The decomposition is conducted using \texttt{qiskit}'s \texttt{decompose} module. We observe that the gate count and the total number of parameters in the UCCSD ansatz scale as $1.45\times N^{3.45}\approx\mathcal{O}(N^{3.45})$ under the Jordan-Wigner transformation, where $N$ represents the number of qubits (i.e., spin orbitals). In contrast, the gate count in CRLVQE scales as $0.26\times N^3\approx\mathcal{O}(N^3)$\footnote{The curve fitting data was estimated using the \texttt{polyfit} module in \texttt{python}.}. Although the improvement over the number of gates and depth of the parameterized quantum circuits is not significant, CRLVQE demonstrates the potential to address quantum chemistry problems on NISQ devices with circuit containing number of gates $20.68$ times and depth up to $20$ times smaller than UCCSD as shown in \figref{fig:complexity_analysis}.
}

\subsection{Noisy case}\label{sec:noisy_simulation_molecule}
Let us now discuss the effect of quantum noise, such as shot noise and real device noise with connectivity constraints, in the proposed algorithm. To see the performance of the CRLVQE method, we find the ground state of $\texttt{H}_2$ molecule with two-, three-qubit and $\texttt{LiH}$ with four-qubit under different amplitude of shot noise. Furthermore, we take the maximum noise from the \texttt{ibmq\_mumbai} and \texttt{ibmq\_ourense} device of \texttt{IBM quantum} and uniformly apply it to all the qubits for three- and four-qubit $\texttt{H}_2$ molecule. The noise includes the single, two-qubit depolarizing noise, readout error, thermal relaxation, and single and two-qubit gate time without connectivity constraints for \texttt{ibmq\_mumbai} and with connectivity constraints for \texttt{ibmq\_ourense}.
\begin{table}[h!]
\caption{\small{Tabular representation of the maximum noise of $\texttt{ibmq\_mumbai}$ device}. Additionally, the qubit frequency and the anharmonicity are the same for maximum noise settings and are set to $4.896$ GHz and $-0.33$ GHz, respectively.}
\centering
\tiny
\begin{tabular}{@{}ccccccc@{}}
	\toprule
	Model/Noise    & \makecell{1q \\ dep.} & \makecell{2q \\ dep.} & \makecell{Read. \\ error}       & \makecell{Therm. rel. \\ noise ($\mu$s)} & \makecell{1q gate\\ time (s)} & \makecell{2q gate\\ time (s)} \\ \midrule
	Max            & $1.45\times10^{-3}$     & $2.30\times10^{-2}$ & $8.7\times10^{-2}$  & \makecell{$T_1 = 122.286$\\ $T_2 = 167.2$} & $35\times10^{-9}$ & $739.5\times10^{-8}$             \\ \bottomrule
\end{tabular}

\label{tab:ibm_mumbai_noise_table}
\end{table}

\newpage
%
%
%
\begin{table}[H]
\setlength\extrarowheight{15pt}
\centering
\small
\begin{tabular}{l|cccc|cccc}
	\hline
	\diagbox[]{Methods}{Molecule}& \multicolumn{4}{c}{six-qubit \texttt{LiH}}& \multicolumn{4}{c}{eight-qubit $\texttt{H}_2\texttt{O}$} \\ 
	\hline
	&$\varepsilon$&$N_P$&$N_D$&$N_G$  &$\varepsilon$&$N_P$&$N_D$&$N_G$ \\
	\cline{2-9}
	Ours (RH)   &$1.22\times10^{-3}$&$\cellcolor{blue!30}{15}$&$14$&$\cellcolor{blue!30}{25}$   &&&&    \\ 
	Ours (wo-RH)   &$1.98\times10^{-4}$&$25$&$22$&$42$   &$\cellcolor{blue!30}{1.84\times10^{-4}}$&$\cellcolor{blue!30}{35}$&$75$&$\cellcolor{blue!30}{140}$    \\ 
	%
	%
	UCCSD   &$\cellcolor{blue!30}{4.0\times10^{-5}}$&$224$&$347$&$464$   &&&&    \\ 
	%
	%
	RL-VQE   &CA&17&$\cellcolor{blue!30}{11}$&$27$   &&&&    \\ 
	%
	%
	quantumDARTS   &$2.9\times10^{-4}$&$80$&$54$&$132$   &$3.1\times10^{-4}$&$151$&$\cellcolor{blue!30}{64}$&$219$  \\ \hline
\end{tabular}
\caption{\small{Our algorithm outperforms the existing ansatz such as UCCSD, the ansatz proposed in RL-VQE and the quantumDARTS algorithm in terms of the number of parameters ($N_P$) and number of gates ($N_G$) for six-qubit \texttt{LiH} molecule and in terms of the error in energy ($\varepsilon$), $N_P$ and $N_G$ for eight-qubit $\texttt{H}_2\texttt{O}$ molecule}. In the table, the empty cells correspond to the results that are not relevant or unavailable for that particular algorithm. The \textit{CA} corresponds to chemical accuracy. For six-qubit \texttt{LiH}, we run our algorithm with the random halting technique and without it and present both results. On the other hand, for $\texttt{H}_2\texttt{O}$, we only conduct the simulation without random halting settings.}
\label{tab:noiseless_rl_algorithm_compatison}
\end{table}
%
%
%
%
%
In~\cite{du2022quantum}, the authors consider the four-qubit $\texttt{H}_2$ molecule under \texttt{ibmq\_ourense} device noise and constrained connectivity. Hence, we compare the performance of our algorithm with the one introduced in~\cite{du2022quantum} in {Tab.~\ref{tab:noisy_rl_algorithm_compatison}}, highlighted with blue colour) and show that our algorithm outperforms in terms of error in energy estimation ($\varepsilon$), number of parameters ($N_P$), depth ($N_D$), and number of gates ($N_G$). Using the QAS algorithm~\cite{du2022quantum}, the minimum error in energy recorded is $1.88\times10^{-2}$ with the number of parameters $10$, depth $9$ and the number of gates $16$, but using our algorithm, we achieve an error in energy $2.98\times10^{-4}$ with ansatz containing $6$ parameters, $6$ depth and with just $10$ gates.

Additionally, from the results, we conclude that the algorithm we present is susceptible to shot noise, and we solve almost every instance of the neural network for two-, three-, four-qubit $\texttt{H}_2$ and $\texttt{LiH}$ molecules. For two- and three-qubits, the error in the ground energy estimation goes well below $10^{4}$, and for the four-qubit, we get an error below $10^{-3}$ (below chemical accuracy). The parameter \texttt{SN} corresponds to the amplitude of shot noise applied for the molecule. It can be seen that even with $10^3$ (two-qubit $\texttt{H}_2$) and $10^4$ (three-qubit $\texttt{H}_2$) shots our algorithm finds an error below $10^{-4}$ with just $8$ and $5$ gates respectively.

Finally, to show the diversity of our algorithm in {Tab.~\ref{tab:noisy_rl_algorithm_compatison}} we solve the ground state of the three-qubit $\texttt{H}_2$ molecule under maximum noise of $\texttt{ibmq\_mumbai}$ as presented in {Tab.~\ref{tab:ibmq_lima_noise_table}}. We show that in all the seeds, we are able to solve the molecule with an error in energy in the order of $10^{-4}$ with a minimum number of parameters $2$, depth $7$ and a total number of gates $8$.

For GPU/CPU specifications and the computational time of each episode of RL we encourage the readers to see the details in Table~\ref{tab:training-time-record}.
\newpage

\begin{table}[H]
\caption{ \small{Our algorithm solves the two-, three-, four-qubits $\texttt{H}_2$ in all initialization of the neural network and four-qubit $\texttt{LiH}$ problem in 2 out of 3 seeds in the presence of different amplitude of shot noise. Our algorithm outperforms the QAS algorithm presented in~\cite{du2022quantum} under \texttt{ibmq\_ourense} noise and connectivity in terms of the number of parameters ($N_P$) and number of gates ($N_G$) for six-qubit \texttt{LiH} molecule and in terms of the error in energy ($\varepsilon$)}. Unlike in QAS~\cite{du2022quantum} where the algorithm could not achieve the chemical accuracy, we show that using our algorithm, we can go $10\times$ below chemical accuracy using an ansatz with $6$ parameters and $4$ rotations for four-qubit \texttt{LiH} problem. The \texttt{SN} corresponds to the number of shots that are considered for the molecule.}
\setlength\extrarowheight{2.5pt}
\centering
\small
\begin{tabular}{l|ccccc}
	\hline
	\diagbox[width=7cm]{Molecules}{Methods}& \multicolumn{5}{c}{Our algorithm} \\ 
	\hline
	&seed&$\varepsilon$&$N_P$&$N_D$&$N_G$ \\
	\cline{2-6}
	& $100$   &$3.63\times10^{-6}$&$4$&$4$&$5$     \\
	$\texttt{H}_2\; (2 \texttt{qubit}, \texttt{SN}=10^3,\texttt{RH})$ & $101$   &$1.16\times10^{-4}$&$14$&$15$&$16$     \\
	& $102$   &$9.25\times10^{-6}$&$38$&$24$&$40$      \\\cline{2-6}
	& $100$   &$2.81\times10^{-5}$&$6$&$7$&$9$     \\
	$\texttt{H}_2\; (3 \texttt{qubit}, \texttt{SN}=10^4,\texttt{RH})$ & $101$   &$4.31\times10^{-5}$&$7$&$4$&$9$    \\
	& $102$   &$7.94\times10^{-5}$&$5$&$5$&$8$      \\\cline{2-6}
	& $100$   &$3.30\times10^{-4}$&$7$&$8$&$13$     \\
	$\texttt{H}_2\; (4 \texttt{qubit}, \texttt{QAS},\texttt{RH})$ & $101$   &$\cellcolor{blue!30}{2.98\times10^{-4}}$&$\cellcolor{blue!30}{6}$&$\cellcolor{blue!30}{6}$&$\cellcolor{blue!30}{10}$    \\
	& $102$   &$3.30\times10^{-4}$&$7$&$8$&$13$    \\ \cline{2-6}
	& $100$   &$4.38\times10^{-4}$&$2$&$8$&$8$    \\
	$\texttt{H}_2\; (3 \texttt{qubit}, \texttt{ibmq\_mumbai max}, \texttt{RH})$ & $101$   &$3.38\times10^{-4}$&$3$&$7$&$8$    \\
	& $102$   &$3.94\times10^{-4}$&$2$&$7$&$8$   \\ \cline{2-6}
	& $100$   &$1.32\times10^{-3}$&$23$&$16$&$31$    \\
	$\texttt{LiH}\; (4 \texttt{qubit}, \texttt{SN}=10^6, \texttt{RH})$ & $101$   &$1.19\times10^{-3}$&$25$&$15$&$35$    \\
	& $102$   &NS&NS&NS&NS   \\
	\hline
\end{tabular}
\label{tab:noisy_rl_algorithm_compatison}
\end{table}
\section{Takeaways}
This chapter introduced a vanilla and curriculum reinforcement learning-based quantum architecture search algorithm for variational quantum algorithms. The algorithm is benchmarked for under noiseless and realistic noisy scenarios based on \texttt{IBM} hardware. The crucial takeaways from the chapter are as follows 
\begin{itemize}
\item \textbf{Tensor-based ansatz encoding provides efficient data representation for RL} In {Sec.~\ref{sec:encoding_comparison}} we introduce a depth-based binary encoding, namely TBE (tensor-based encoding) for quantum circuits that we utilize as an RL-state. In the very heart of the encoding lies a 3D grid structure where each dimension carries information about the depth, the type of the gate, and the position of the gate (i.e. on which qubit the gate is to be placed), respectively. The grid is of size $[T\times \left( (N+3)\times N\right)]$ where $T$ is a predefined number corresponding to maximum depth and $N$ is the number of qubits.

We benchmark the TBE with previously proposed integer encoding, namely IE, in the task of finding the ground state of molecules. Through~\figref{fig:encoding_comparison_LiH_4q}, we simulate a four-qubit \texttt{LiH} molecule and show that the TBE is more stable than IE, and it gives the minimum error in the ground state energy lower than the IE with a smaller number of gates. This is beneficial for the NISQ era, as it helps effectively mitigate the negative impact of gate errors and decoherence.

\item \textbf{Enhanced insight on the reward function}
In {Sec.~\ref{sec:reward_comparison}}, we extend our understanding of a dense and sparse reward based on two formulations of the reward function in finding the ground state of the four-qubit \texttt{LiH} molecule. The first kind of reward we consider is introduced in the RL-VQSD chapter by~\eqref{eq:log_reward}, namely log reward, which depends on a large positive integer $\mathcal{R}$. For the sake of understanding the performance of the reward function, we choose $\mathcal{R}=0,50,100,500,1000$ and compare it with the reward function proposed in~\cite{ostaszewski2021reinforcement} (see~\eqref{eq:old_reward}), namely $R_\text{prev}$.

The results are illustrated in~\figref{fig:reward-comparison} where we can clearly see that the log reward improves as the $\mathcal{R}$ increases up to $\mathcal{R}=500$, and after that the improvement diminishes. Interestingly, the performance of the log reward at $\mathcal{R}=500$ mimics the performance of the $R_\text{prev}$ in terms of the minimum number of gates, depth, number of parameters in ansatz and even in the accuracy of estimating the ground energy. But as the number of successful episodes with the $R_\text{prev}$ is larger than the log reward at $\mathcal{R}=500$, we utilize the $R_\text{prev}$ as the reward function for larger molecules. The main motivation behind this is that the higher the number of successful episodes, the greater the array of ansatz. This helps us investigate a wide arrangement of gates in an ansatz and pick the best one among them.

\item \textbf{Random halting helps discover compact ansatz quickly}
Throughout the {Sec.~\ref{sec:random_halting_performance}}, we elaborate on a simple technique called the \textit{random halting} (RH), which is introduced by keeping the realistic noisy scenario in mind. In the midst of quantum noise, predictably, the RL-agent starts by proposing a lengthy ansatz in early successful episodes. However, over thousands of subsequent episodes, it gradually shifts towards a shorter ansatz. This is inefficient in terms of the computational time of our algorithm when we compare it with the noiseless case. Hence, to address this challenge, we introduce this method where the number of time steps per episode is a variable and changes from one episode to another based on the probability distribution provided in~\eqref{eq:negative_binom_dist}.

In the~\figref{fig:random_halting_effect}, we illustrate the number of gates it requires to solve a six-qubit \texttt{LiH} molecule in the presence and absence of RH. We clearly noted that in the presence of RH, the RL-agent learns $3\times$ faster than the optimal quantum circuit without the RH setting. However, it should be noted that the number of successful episodes drastically decreases with RH.

\item \textbf{Solving molecules under physical noise and connectivity constrained}
In the {Sec.~\ref{sec:noisy_simulation_molecule}}, we utilize our algorithm to find the ground state of $\texttt{H}_2$ and \texttt{LiH} molecules with two-, three- and four-qubit. For our simulation, we consider shot and physical device noise. The noise is imported from the \texttt{IBM quantum} hardware \texttt{ibmq\_ourense} (we consider the maximum noise among all the qubits from the noise model is provided in ref.~\cite{du2022quantum} and uniformly applied to all qubits taking the qubit connectivity into account) and \texttt{ibmq\_mumbai} (we consider the maximum noise among all the qubits from the noise model in {Tab.~\ref{tab:ibm_mumbai_noise_table}} and uniformly apply it all qubits and does not take the qubit connectivity into account). In the case of shot noise, we show that for two- and three-qubit $\texttt{H}_2$ and four-qubit \texttt{LiH} molecule, we can solve the problem with $10^3$, $10^4$ and $10^6$ shots with 5, 9 and 31 gates respectively. Under \texttt{ibmq\_ourense} noise, we solve four-qubit $\texttt{H}_2$ in 10 gates, and it takes 8 gates to solve three-qubits $\texttt{H}_2$ with \texttt{ibmq\_mumbai} noise.

This shows that our algorithm is susceptible to shot noise and can solve \texttt{LiH} problem with ease. Meanwhile, for device noise and constrained connectivity, we can solve the four-qubit $\texttt{H}_2$ with a very small ansatz.

\item\textbf{Introduced algorithm outperforms existing QAS algorithms}
Through the {Tab.~\ref{tab:noiseless_rl_algorithm_compatison}} and {Tab.~\ref{tab:noisy_rl_algorithm_compatison}} we compare the performance of our algorithm with the RL-VQE~\cite{ostaszewski2021reinforcement}, quantumDARTS~\cite{wu2023quantumdarts} and the net based QAS~\cite{du2022quantum} algorithms. We show that in the case of the noiseless scenario, our algorithm outperforms the RL-VQE and the quantumDARTS in terms of the number of gates, accuracy in energy estimation, parameter number and depth of the ansatz. Meanwhile, for the noisy scenario, the authors in using the QAS algorithm~\cite{du2022quantum} could not find the chemical accuracy but using our algorithm we not only provide a shorted ansatz and find the ground energy.
\end{itemize}


\chapter{Reinforcement learning assisted variational certification of quantum channels}\label{ch:vqcd_application_rl_vqsd}

The goal of this chapter is to describe an application of variational quantum state diagonalization (VQSD) techniques in the area of quantum technologies. To this end, we focus on a protocol for quantum channel certification based on the variational approach. We demonstrate the building blocks of the protocol, and we demonstrate the implementation of the proposed algorithm on a near-term quantum computer. The results in the chapter are based on~\cite{kundu2022variational}, and the accompanying source code can be found in~\cite{kundu2021qiskit}. After introducing the quantum channel certification algorithm, we elaborate on how the reinforcement learning (RL)-based quantum architecture search method can be used to enhance the performance of the certification algorithm.


\section{Introduction}
One of the primary applications of quantum state diagonalization is the certification of quantum devices. However, certification tasks pose a significant challenge in quantum computing applications. The certification of the characteristics of a quantum system is similar to trying to recreate the results we can get from a regular classical simulation. However, this task is computationally complex, which aligns with the essence of quantum supremacy~\cite{arute2019quantum,boixo2018characterizing,harrow2017quantum,terhal2018quantum,lund2017quantum}.

The challenge in the certification of a quantum device primarily arises from the inherent computational advantage of quantum computers. Hence, it is better to explore the potential use of quantum computers to certify quantum devices. Here we dive into the scenario where we present a certification approach that is based on the structure of the space of quantum operations. The inherent link between states and operations in quantum mechanics i.e. the Choi-Jamio\l{}kowski isomorphism~\cite{choi1975completely,jamiolkowski1972linear} leads to novel techniques of quantum information processing that has the potential to go beyond the possibilities of classical mechanics. If the quantum device is represented through a quantum channel $\Phi$ then through Choi-Jamio\l{}kowski isomorphism we get the corresponding state as
\begin{equation}
	\rho_\Phi = \mathcal{J}(\Phi) = \left( \mathbb{I} \otimes \Phi\right)\sum_{i=1}^n \ket{i}\otimes\ket{i},
	\label{eq:choi-jamil-iso}
\end{equation}
where $\sum_{i=1}^n \ket{i}\otimes\ket{i}$ is the maximally entangled states.

The problem of distinguishing between two or more quantum devices is equivalent to defining the distance in the space of density matrices. The fascination with the physical implementation of quantum information processing has led to the introduction of a class of distance/similarity measures, as evident from the substantial work carried out in this area~\cite{duan2009perfect,ji2006identification,piani2009all,wang2006unambiguous}. Notably, with a concentrated emphasis on assessment of the practical viability of the suggested methodologies~\cite{sedlak2009unambiguous}.

One of the well-known measures of similarity between two quantum states is quantum state fidelity which is defined as follows~\cite{uhlmann1976transition}
\begin{equation}
	F(\rho,\sigma) = || \sqrt{\rho}\sqrt{\sigma} || = \text{tr}\sqrt{\sqrt{\rho}\sigma\sqrt{\rho}},
	\label{eq:fidelity-formula}
\end{equation}
it gives us the quantum counterpart of the Bhattacharyya coefficient~\cite{bhattacharyya1946measure} which measures the similarity between two probability distributions and it reduces to the scalar product for rank-1 operators.

A significant research effort has been devoted to finding methods for approximating fidelity~\cite{liang2019quantum}. Hence, for ease of calculation, the bounds for the values of fidelity using the functional are introduced~\cite{miszczak2009sub}. The bounds are defined by the sub- and super-fidelity bounds (SSFB) 
\begin{align}
	&F_\text{sub}(\rho, \sigma) = \text{tr}(\rho\sigma) + \sqrt{2\left[ \text{tr}(\rho\sigma) - \text{tr}(\rho\sigma)^2 \right]},\label{eq:sub-fb}\\
	&F_\text{sup}(\rho, \sigma) = \text{tr}(\rho\sigma) + \sqrt{\left(1-\text{tr}\rho^2\right)\left(1-\text{tr}\sigma^2\right)]},
	\label{eq:super-fb}
\end{align}
which satisfies the property
\begin{equation}
	F_\text{sub}(\rho, \sigma)\leq F(\rho, \sigma)\leq F_\text{sup}(\rho, \sigma).
\end{equation}
\begin{figure}[H]
	\centering    \includegraphics[width=0.5\textwidth]{ 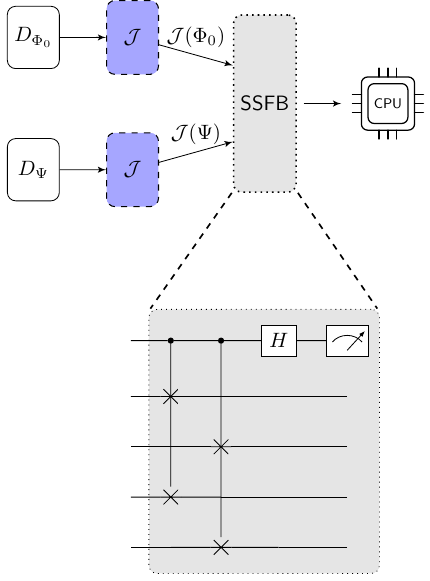}
	\caption{{The super and sub-fidelity bound-based quantum device certification procedure}. In this
		scheme, the goal is to certify a quantum device $\Psi$ against an ideal device $\Phi_0$ using the calculable bounds for
		fidelity. The input of the procedure is given in the form of physical devices, $D_{\Phi_0}$ and $D_\Psi$ for the ideal and the unknown quantum channel respectively. In the CPU part of the scheme, we do classical post-processing of the results obtained using the quantum subroutine.}
	\label{fig:ssfb-certification-process.}
\end{figure}

It should be noted that there is no exact quantum algorithm that can be used to calculate the fidelity. This is because the calculation of the fidelity requires the non-integer powers of the quantum states. In~\cite{cerezo2020variational} the authors introduce the variational quantum state fidelity (VQFE) algorithm which is a variational quantum-classical algorithm to find the bounds of fidelity. The VQFE computes upper
and lower bounds on fidelity. The bounds are based on the \textit{truncated fidelity}, which can be evaluated using~\eqref{eq:fidelity-formula} for state $\sigma$
and a state $\rho_m$. The state $\rho_m$ is called the truncated state of $\rho$ that can be obtained by projecting the state$\rho$ onto the subspace associated with its $m$-largest
eigenvalues. The bounds can be tightened monotonically with the increment in $m$, and finally, the bounds will converge to the true fidelity when $m = \text{rank}(\rho)$. The bounds on the fidelity between states $\rho$ and $\sigma$ are expressed by
\begin{equation}
	F(\rho_m, \sigma^\rho_m)\leq F(\rho, \sigma)\leq F_*(\rho_m, \sigma_m^\rho),
	\label{eq:truncated-fidelity-bound}
\end{equation}
where
\begin{equation}
	F_*(\rho_m, \sigma_m^\rho) = || \sqrt{\rho_m}\sqrt{\sigma_m^\rho} || + \sqrt{\left(1-\text{tr}\rho_m\right)\left(1-\text{tr}\sigma^\rho_m\right)]},
	\label{eq:generalized-truncated fidelity}
\end{equation}
where $\sigma_m^\rho = \Pi_m^\rho \sigma \Pi_m^\rho$ is the operator obtained as the projection of $\sigma$ onto the subspace spanned by $m$ largest eigenvectors of $\rho$. The $F_*(\rho_m, \sigma_\rho^m)$ is the \textit{truncated fidelity}. The $F_*$ is also utilized to compute quantum Fisher information~\cite{sone2021generalized}.

Following these developments, we present a novel algorithm that utilizes the super and sub-fidelity bounds and VQFE procedures as building blocks for quantum device certification. We achieve this by combining the procedures for the estimation of bounds on the fidelity with the resulting density matrix obtained by using Choi-Jamio\l{}kowski isomorphism given in~\eqref{eq:choi-jamil-iso}. In~\figref{fig:ssfb-certification-process.} we illustrate an algorithm based on the bounds given in~\eqref{eq:sub-fb} and in~\eqref{eq:super-fb}. 

The procedure takes two devices as input – the standard device ($\Psi$) with the operational capacity already confirmed, and the device for which its conformation with the standard device is to be confirmed.

One should note that in this scheme classical data processing is required only at the final step of the procedure. This step is required to compute the bounds for the fidelity based on the measurement results.

In the following, we first briefly discuss the problem statement and then describe the novel quantum device certification algorithm. Next, we briefly discuss the results of the algorithm which is followed by a brief investigation of the results. Finally, we give conclusive remarks.

\section{Groundwork}
\subsection{Problem statement}
Let's consider a scenario where a quantum start-up has successfully developed a quantum device that can address critical optimization problems or can generate valuable states essential for quantum communication protocols. In this context, it becomes crucial to provide some testing procedures that will reassure the buyers about the device's actual functionality. Hence, the main object of the buyer would be to verify the quantum device, if it performs as advertised by the seller.

In the general context of differentiating between quantum channels, it is customary to assume that we have for our disposal a set of $N$ quantum devices that are denoted by the quantum channels $\Psi_1, \Psi_1,\ldots, \Psi_N$. Each device operates as a black box, which directly indicates that we are not aware of the Kraus representation of the channels. In such a situation, it is impossible to determine whether the input devices can be perfectly distinguished~\cite{duan2009perfect}.

However, in our situation, the task is straightforward. All we intend to do is to convince the buyer that the device we would like to sell emulates the operation of an ideal device, denoted by $D_{\Phi_0}$ which in turn characterized by a quantum channel $\Phi_0$, operating on the space of $n$ qubits. Moreover, we have the second device, $D_\Psi$, which is claimed to be indistinguishable compared with $D_{\Phi_0}$.

Furthermore, as the start-up specializes in quantum technology, the board decided that the certification procedures should also benefit from the quantum advantage. Such a decision yields two benefits. Firstly, it supports the claims concerning the ubiquitous applications of quantum computing. Secondly, it provides an opportunity to develop a unique certification service that can be offered to other quantum start-ups~\cite{mohseni2017commercialize}.

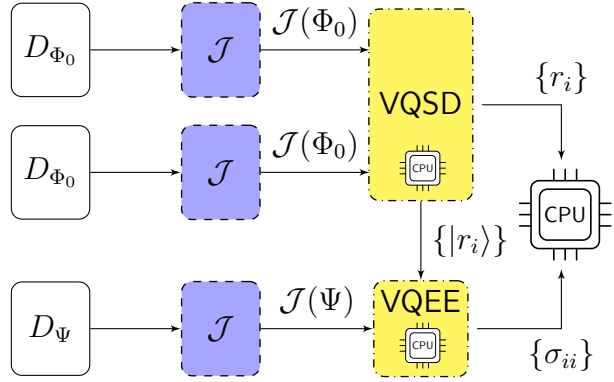
\begin{figure}
	\centering
	\begin{tikzpicture}[auto, node distance=2cm,>=latex',scale=1.12]
		
		\node at (0,  1.75) [qdev, name=qdev1] {$D_{\Phi_0}$};
		\node at (0,  0.30) [qdev, name=qdev2] {$D_{\Phi_0}$};
		\node at (0, -1.55) [qdev, name=qdev3] {$D_{\Psi}$};
		
		\node at (2,  1.75) [jam1, name=jamPhi01] {$\mathcal{J}$};
		\node at (2,  0.30) [jam1, name=jamPhi02] {$\mathcal{J}$};
		\node at (2, -1.55) [jam1, name=jamPsi] {$\mathcal{J}$};
		
		\node at (4.35, 1.10) [vqfe, name=vqsd] {\sf VQSD};
		\node at (4.35, 0.35) [cpu,scale=0.5] {};
		
		\node at (4.35, -1.55) [vqfe-short, name=vqsd-part] {};
		\node at (4.35, -1.25) [] {\sf VQEE};
		\node at (4.35, -1.75) [cpu,scale=0.5] {};
		
		\node at (6, -0.20)[cpu, name=cpu-pp] {};
		
		\draw [draw,->] (qdev1.east) -- (jamPhi01);
		\draw [draw,->] (qdev2.east) -- (jamPhi02);
		\draw [draw,->] (qdev3.east) -- (jamPsi);
		
		\draw [draw,->] (jamPhi01.east) --node[anchor=south,above]{$\Jam{\Phi_0}$} (3.75,  1.75);
		\draw [draw,->] (jamPhi02.east) --node[anchor=south,above]{$\Jam{\Phi_0}$} (3.75,  0.30);
		\draw [draw,->] (jamPsi.east) --node[anchor=south,above]{$\Jam{\Psi}$} (vqsd-part);

		\draw [draw,->] (vqsd.south)  --node[anchor=south,right]{$\{\ket{r_i}\}$} (vqsd-part.north);

		\draw [draw,->] (vqsd.east)+(0.1,0) -| node[anchor=south,above]{$\{r_i\}$}  (cpu-pp.north);

		\draw [draw,->] (vqsd-part.east)+(0.1,0) -| node[anchor=west,below]{$\{\sigma_{ii}\}$}  (cpu-pp.south);
		
	\end{tikzpicture}
	\caption{The variational quantum fidelity estimation (VQFE) certification algorithm.}
	\label{fig:vqfe-certification-scheme}
\end{figure}

\subsection{The algorithm}\label{sec:vqfe-based-vertification-salgorithm}

Apart from the process described through~\figref{fig:ssfb-certification-process.} we introduce an alternative certification approach based on variational quantum fidelity estimation~\cite{kundu2022variational}. The primary goal of the algorithm is to find the truncated fidelity of $\mathcal{J}(\Psi)$ on the basis of $m$ largest eigenvalues of $\mathcal{J}(\Phi_0)$. In~\figref{fig:vqfe-certification-scheme} we illustrate the VQFE-based certification scheme. In this scheme to certify a quantum device $\Psi$ against the ideal device $\Phi_0$, the VQFE procedure is used. The input of the procedure is given in the form of physical devices, $D_{\Phi_0}$ and $D_{\Psi}$, implementing channels $\Phi_0$ and $\Psi$ respectively. The first step is to apply Choi-Jamio{\l}kowski isomorphism to obtain dynamical matrices for the input devices. Next, two copies of the ideal device ($\Jam{\Phi_0}$) are used as an input for the variational quantum state diagonalization (VQSD) procedure, which is briefly discussed in the {Sec.~\ref{sec:vqsd-algorithm-brief}}. At the same time, $\Jam{\Psi}$ is processed by a quantum-classical algorithm to obtain its matrix elements in the eigenbasis of $\Jam{\Phi_0}$. Finally, classical processing of the obtained eigenvalues is used to calculate the approximation of the fidelity between quantum operations.
Note that in this procedure, the CPU part is utilized at several steps -- as a part of VQSD used for calculating eigenvalues $r_i$ of $\Jam{\Phi_0}$ and matrix elements $\sigma_{ij}$ of $\Jam{\Psi}$. The yellow blocks in~\figref{fig:vqfe-certification-scheme} indicate hybrid quantum-classical sub-procedures.

In a nutshell, the variational quantum fidelity-based certification procedure consists of the following steps.
\begin{itemize}
	\item First we prepare two copies of $\Jam{\Phi_0}$ to process in the VQSD algorithm and one copy of $\Jam{\Psi}$.
	
	\item The two copies of the ideal device are utilized to diagonalize $\Jam{\Phi_0}$. After this process, we get the $m$-largest eigenvalues $\{r_i\}$, which can be subsequently stored on a classical CPU, and the eigenvectors of $\Jam{\Phi_0}$, which will be useful in the upcoming step.
	
	\item We now make use of the VQFE procedure with the $\Jam{\Psi}$ and provide the eigenvectors of the ideal device obtained from the previous VQSD sub-procedure to obtain matrix elements $\sigma_{ii}=\bra{r_i}\Jam{\Psi}\ket{r_j}$ in the eigenbasis of $\Jam{\Phi_0}$, if the cost function in VQSD process is non-zero then we get $\sigma_{ii}^\prime=\bra{r_i^\prime}\Jam{\Psi}\ket{r_j^\prime}$, where $r_i^\prime$ is the inferred eigenbasis.
	
	\item The resulting matrix elements of $\Jam{\Psi}$ in the eigenbasis of $\Jam{\Phi_0}$, and eigenvalues of $\Jam{\Phi_0}$ are used to calculate truncated fidelity bounds according to Eq.~(\ref{eq:truncated-fidelity-bound}). To obtain these bounds one needs to first compute the RHS of~\eqref{eq:generalized-truncated fidelity}.
	\begin{equation}
		||\sqrt{\rho_m}\sqrt{\sigma_m^\rho}||=\Tr{\sqrt{\sum_{i,j}T_{i,j}\ket{r_i}\bra{r_j}}},
	\end{equation}
	where $T_{i,j}$ is a matrix whose dimension is dependent on the number of largest eigenvalues we can retrieve from the VQSD process. If we have $m$ largest eigenvalues then $T_{i,j}$ is if $m\times m$ which elements are computed by
	\begin{equation}
		T_{i,j} = \sqrt{r_ir_j}\bra{r_i}\Jam{\Psi}\ket{r_j}, \;\;\; \text{such that}\;\;T\geq0.
		\label{eq:T-matrix-tfb}
	\end{equation}
	Explicitly the truncated fidelity bounds are computed as follows
	\begin{align}
		& F_*(\rho_m,\sigma_m^\rho) = \sum_i\sqrt{\lambda_i}+\sqrt{\left(1-\sum_ir_i\right)\left(1-\sum_i\sigma_{ii}\right)},\nonumber\\
		& F(\rho_m,\sigma_m^\rho) = \sum_i\sqrt{\lambda_i},
		\label{eq:TFB-numericals}
	\end{align}
	with $\lambda_i$ are the eigenvalues of $T$ where $i=1,2,\ldots,m$.
\end{itemize}
\subsection{Noise models}\label{sec:noise-model-certification}
This section briefly describes a class of noise models that we utilize to investigate the variational device certification process. The noise models are constructed using the Kraus operator representation~\cite{nielsen2010quantum}. In the following we provide a brief introduction to quantum channels.

\paragraph{Depolarizing noise} For a single \textit{depolarized} qubit with probability $\gamma$, the term \textit{depolarized} means that the single qubit state is replaced by a completely mixed state i.e. $\id/2$ and with $(1-\gamma)$ probability the qubit is completely left untouched. Hence the state of the system after getting affected by the noise is
\begin{equation}
	\Delta_\gamma = (1-\gamma)\rho + \frac{\gamma}{2}\id.
\end{equation}
The circuit model of simulating depolarizing noise contains three-qubits where one of the qubits contains the input quantum state and the remaining two lines are an environment to simulate the channel.
\begin{figure}[H]
	\centering
	\begin{quantikz}[thin lines,scale=1.12]
		\lstick{$\rho$} & \swap{1} & \qw\\
		\lstick{$\frac{\id}{2}$} & \targX{} & \qw\\
		\lstick{$(1-\gamma)\rho + \frac{\gamma}{2}\id$} & \ctrl{-2} & \qw
	\end{quantikz}
	\caption{Illustration of a circuit that simulates depolarizing noise.}
	\label{fig:depol-noise-circ}
\end{figure}
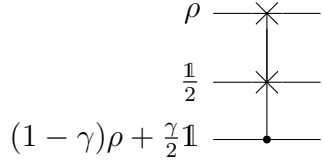
In the~\figref{fig:depol-noise-circ} the main idea behind the circuit is that the controlled qubit is the mixture of $\ket{0}$ with $(1-\gamma)$ probability and state $\ket{1}$ with $\gamma$ probability and this decides whether or not the state $\id/2$ is swapped into the first qubit.

\paragraph{Amplitude damping noise} The amplitude-damping channels lead to a decay of energy from an excited state to the ground state depending on the probability $\gamma$. Hence, the channel's action on a state is given as 
\begin{equation}
	A_\gamma = \mathcal{K}_0\rho\mathcal{K}_0^\dagger + \mathcal{K}_1\rho\mathcal{K}_1^\dagger,
\end{equation}
where 
\begin{equation}
	\mathcal{K}_0 = 
	\begin{bmatrix}
		1 & 0 \\
		0 & \sqrt{1-\gamma} 
	\end{bmatrix},\;\;
	\mathcal{K}_1 = 
	\begin{bmatrix}
		0 & \sqrt{\gamma}\\
		0 & 0
	\end{bmatrix}.
\end{equation}
In~\figref{fig:amp-damp-noise-circ} we illustrate the circuit to simulate depolarizing noise.
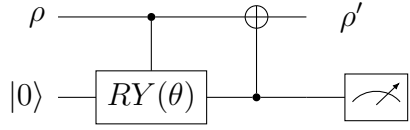
\begin{figure}[H]
	\centering
	\begin{quantikz}[thin lines,scale=1.12]
		\lstick{$\rho$}     & \ctrl{1} & \targ{}    & \qw       & \lstick{$\rho'$}\\
		\lstick{$\ket{0}$}  & \gate{RY(\theta)}     & \ctrl{-1}& \qw & \meter{}
	\end{quantikz}
	\caption{Illustration of a circuit that simulates amplitude damping noise.}
	\label{fig:amp-damp-noise-circ}
\end{figure}
The intuition behind the circuit is. Let us consider the $\rho$ state is $a\ket{0}+b\ket{1}$ if the probability corresponding to $\ket{1}$ is higher then it is more probable that the ground state of the environment will be $\ket{1}$, which activates the \texttt{CNOT} gate and makes $\rho' = a\ket{1}+b\ket{0}$. So depending on the time the channel is applied, it is possible that the information corresponding to the qubit can be nearly completely dissipated.

\paragraph{Random X noise} In the case of the random $X$ quantum noise model the effect of the noise can be described as follows
\begin{equation}
	R_\gamma = \gamma X + (1-\gamma)\id,
\end{equation}
where with probability $\gamma$ we apply $X$-gate and with $(1-\gamma)$ no gate is applied. 

It should be noted that unlike the previous two noise models, which directly affect the state of a quantum system, random X noise is a gate-based noise model where during the construction of the ansatz for variation algorithms after each successive gate with $\gamma$ probability, an $X$-gate might be applied depending on the probability distribution.

\paragraph{Real device noise}
By utilizing the \texttt{IBMQ.get\_privder}(\textit{device}), we can access the IBM Quantum real device backends, which simulate the exact noise model for that particular IBMQ device. During noisy simulation, we use the noise models of \texttt{IBMQ\_lima} and \texttt{IBMQ\_manila}. A tabular representation of various noise amplitudes is provided in {Tab.~\ref{tab:ibmq_lima_noise_table}} and {Tab.~\ref{tab:ibmq_manila_noise_table}} for \texttt{IBMQ\_lima} and \texttt{IBMQ\_manila} respectively.

\begin{table}[H]
	\centering
	\begin{tabular}{|l|l|l|l|l|l|l|}
		\hline
		Qubit & T1 (us) & T2 (us) & 1q gate error & CNOT error \\ \hline
		0 & 87.497 & 195.117 & 3.769e-4 & 0\_1:0.00657 \\ \hline
		1 & 113.244 & 115.320 & 4.574e-4 & 1\_0:0.00657 \\ \hline
		2 & 111.164 & 134.323 & 3.521e-4 & 2\_1:0.00657 \\ \hline
		3 & 98.681 & 81.855 & 2.283e-4 & 3\_4:0.0143 \\ \hline
		4 & 23.512 & 26.554 & 7.011e-4 & 4\_3:0.0143 \\ \hline
	\end{tabular}
	\caption{Parameters of various noises in \texttt{IBMQ\_lima} device.}
	\label{tab:ibmq_lima_noise_table}
\end{table}

\begin{table}[H]
	\centering
	\begin{tabular}{|l|l|l|l|l|l|l|}
		\hline
		Qubit & T1 (us) & T2 (us) & 1q gate error &  CNOT error \\ \hline
		0 & 148.164 &	56.529 &	3.113e-4  &	0\_1:0.00813\\ \hline
		1 & 308.163 &	80.769 &	2.217e-4  &	1\_2:0.00892\\ \hline
		2 & 94.381	&   21.952 &	2.149e-4  &	2\_3:0.00716\\ \hline
		3 & 141.151	&   73.101 &	2.464e-4  & 3\_4:0.00744\\ \hline
		4 & 91.733	&   41.979 &    4.804e-4  &	4\_3:0.00744
		\\ \hline
	\end{tabular}
	\caption{Representation of various noises in \texttt{IBMQ\_manila} device.}
	\label{tab:ibmq_manila_noise_table}
\end{table}

It should be noted that in the tables the \textit{1q gate error} decomposes in the error in $\id$, $SX$ (square root of $X$ gate), and the $X$ gates because the basis gates for \texttt{IBMQ\_manila} and \texttt{IBMQ\_manila} are \texttt{CNOT}, $\id$, \texttt{RZ}, \texttt{SX}, and \texttt{X} so it decomposes any gate into the following basis gates and apply noise. 

\subsection{Error quantification} We quantify the error in the VQFE-based certification process calculating the difference in the truncated fidelity and the true fidelity defined as
\begin{equation}
	\Delta F(\rho,\sigma^\rho) = F(\rho_m,\sigma_m^\rho) - F(\rho,\sigma).
	\label{eq:certification-error-quant}
\end{equation}
Throughout this chapter, we use the above quantifier mentioned in~\ref{eq:certification-error-quant} if not stated otherwise.


\section{Results}
In this section, we demonstrate the performance of the VQFE-based device certification algorithm where we consider (1) random 1-qubit quantum channels and then we scale up the system to consider (2) two-qubit quantum channels. 

\begin{figure}[H]
	\centering
	\begin{subfigure}{.5\textwidth}
		\centering
		\includegraphics[width=\linewidth]{ 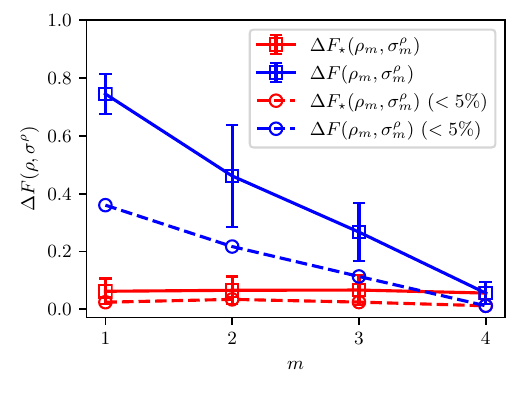}
		\caption{\texttt{CZ} as entangler.}
		\label{fig:TFB_cz}
	\end{subfigure}%
	\begin{subfigure}{.5\textwidth}
		\centering
		\includegraphics[width=\linewidth]{ 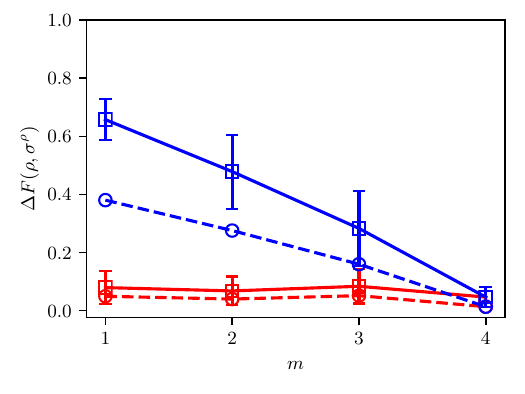}
		\caption{\texttt{CNOT} as entangler.}
		\label{fig:TFB_ CNOT}
	\end{subfigure}
	\caption{{Error in fidelity estimation with respect to $m$. The results were obtained by taking an average over $1000$ random 1-qubit quantum channels ($n = 1$) of rank~$4$, using IBM quantum computer simulator from 
			\texttt{Qiskit Aer} package}. Here $\Delta F(\rho, \sigma^\rho) = F(\rho_m, \sigma_m^\rho) - F(\rho, \sigma)$. The dashed lines depict the average over channels with less than $5\%$ error in the estimation of fidelity. The source code for the implementation can be obtained from~\cite{kundu2021qiskit}.}
	\label{fig:sim_device_certification_algo}
\end{figure}

In the case of a one-qubit channel, we do a rigorous investigation of the certification algorithm in the presence and absence of noise. While for two-qubit we investigate the noiseless scenario briefly.

\subsection{One-qubit quantum channel}

\paragraph{Noiseless scenario}
The results for the 1-qubit random quantum channels are illustrated in~\figref{fig:sim_device_certification_algo}. We sample the quantum channels from Haar distribution using the \texttt{random\_quantum\_channel} module of \texttt{qiskit.quantum\_info}.

The results are averaged over $1000$ random quantum channels. For the purpose of illustration, we explicitly showcase the case where the ansatz contains \texttt{CZ} and \texttt{CNOT} as entangler in~\figref{fig:TFB_cz} and in~\figref{fig:TFB_ CNOT} respectively. This helps us to note that for random 1-qubit quantum channels, it is better to consider \texttt{CNOT} gates in the ansatz than other entangling gates. Additionally, in both cases,  we observe that the introduced certification procedure provides a very good approximation to fidelity for low-rank quantum channels. At the same time, SSFB certification can provide a useful lower bound for fidelity between operations. However, the upper bound obtained in the SSFB case is unsuitable for providing a viable approximation of fidelity.

\begin{figure}[H]
	\centering
	\includegraphics[]{ 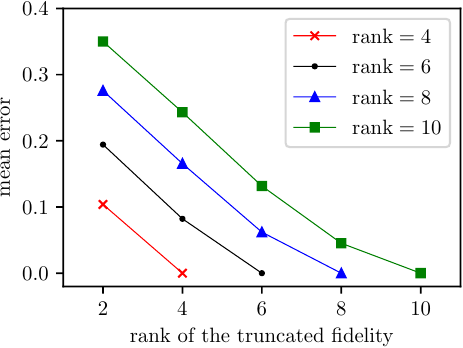}
	\caption{{Mean error for the approximation of fidelity by the truncated fidelities. Values are plotted for random quantum operations as a function of the rank of the truncated fidelity}. Each combination of color and shape corresponds to density matrices with a fixed rank. Each point was obtained by averaging the difference between the fidelity and the truncated fidelity on the sample of $10^5$ pairs of random dynamical matrices.}
	\label{fig:truncated-fid-error}
\end{figure}
In ~\figref{fig:truncated-fid-error}, where the dependency of the mean approximation error of the truncated fidelity is plotted for dynamical matrices with different ranks. As one can see, the bound obtained using the truncated fidelity can be easily tightened. Moreover, the mean for the given rank of the truncated fidelity decreases with the increasing rank of random dynamical matrices. The theoretical interpretation of this observation can be clearly seen through the~\eqref{eq:TFB-numericals} where the bounds in TFB reply on the eigenvalues of the matrix $T$. In~\eqref{eq:T-matrix-tfb} we also see that the primary building blocks of the $T$ matrix are the eigenvalues of the quantum channel that is to be diagonalized in the VQSD subroutine. Hence higher the rank of the quantum channel the more eigenvalues the VQSD subroutine can approximate which in turn gives a better approximation to the eigenvalues of the $T$ matrix, resulting in tighter truncated fidelity bound.

The exact ansatz construction that we use is a hardware efficient ansatz (HEA) of the form
\begin{equation}	
	U(\vec{\theta}) = U_\text{ent}\times\prod_{i=1}^{N}\texttt{RY}(\theta_i)^{\otimes i} \texttt{RZ}(\theta_i)^{\otimes i},\label{eq:ansatz}
\end{equation}
where $n$ is the size of the quantum channel.

\paragraph{Noisy scenario}
In this part, we illustrate through~\figref{fig:sim-noise-on-certification} the effect of depolarizing, amplitude damping, and random $X$ noise.
\begin{figure}[H]
	\centering \includegraphics[width=\textwidth]{ 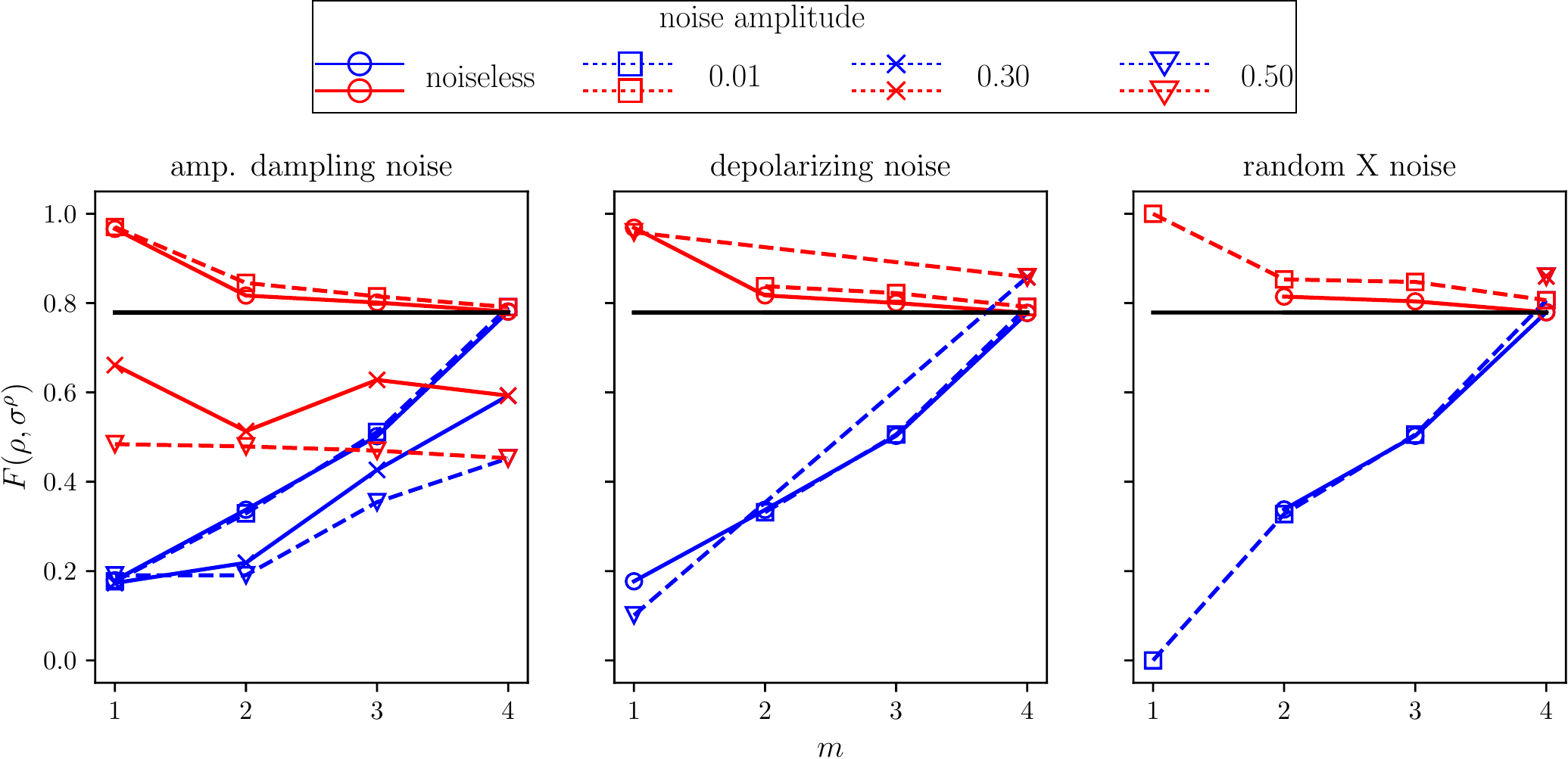}
	\caption{{Effect of noise in the estimation of truncated fidelity bound under amplitude damping, depolarizing, and random $X$ noise}. It can be seen that the effect of amplitude damping noise is more on fidelity estimation when we increase the noise above $30\%$. Additionally, for amplitude damping noise when the noise amplitude is within $10\%$, we do not observe any notable changes in the fidelity estimation, but in retrieving lower ranks, it performs poorly.}
	\label{fig:sim-noise-on-certification}
\end{figure}
A brief discussion of the three kinds of noise models is given in {Sec.~\ref{sec:noise-model-certification}}. Through our illustration, we see that the fidelity bound fails to converge to the true value as the noise amplitude increases more than $10\%$. The depolarizing and random $X$ noise elevates the fidelity at full rank compared to the true value, while due to the effect of amplitude damping noise, the estimated fidelity decreases. It is due to the fact that the amplitude damping noise reduces the value of the eigenvalues that are collected during the channel diagonalization process. As we already saw the fidelity calculation directly depends on the eigenvalues i.e. more exact the eigenvalues the better the estimation of fidelity hence a decrease in eigenvalue decreases the estimated fidelity. When the amplitude of noise increases more than $10\%$, the deviation of fidelity estimation from its true value becomes more prominent.

We can also notice that the higher the noise value the higher the error in the fidelity estimation due to failure in retrieval of eigenvalue for lower rank. This is more prominent in the case of random X noise where we only retrieve the fidelity eigenvalue corresponding to the highest rank i.e. rank 4.

In Fig.(\ref{fig:real-device-simulate}), we depict the results for fidelity estimation after running the algorithm in \texttt{Qiskit Aer} with real device backends provided by IBM. In the presence of noise, we expect a degradation of the performance of the algorithm which is evident from the figure. The fidelity achieved by using \texttt{CNOT} and \texttt{CZ} for the entangling gate is similar at the 3rd layer but the \texttt{CZ} gate helps to obtain a very close to optimal fidelity (which is less than the true fidelity) in just one layer. Whereas the ansatz with \texttt{CNOT} keeps on improving at each step and finally achieves a similar fidelity as the ansatz with \texttt{CZ} at the 3rd layer. This helps us infer that although \texttt{CNOT} gate two-qubit better learnability, hence beneficial for optimization in a noiseless scenario. But in a noisy case, it is profitable to use \texttt{CZ} to achieve higher accuracy with very few gates and depth. 

As the noise model for both devices is similar hence we observe similar characteristics in the variation of fidelity with respect to layers for the two kinds of noise. Additionally, as the amplitude-damping noise is more prominent in the real device we see a decrease in the fidelity value.
As the noise model for both devices is similar hence we observe similar characteristics in the variation of fidelity with respect to layers for the two kinds of noise. Additionally, as the amplitude-damping noise is more prominent in the real device we see a decrease in the fidelity value.
\begin{figure}[H]
	\centering
	\begin{subfigure}{.5\textwidth}
		\centering
		\includegraphics[width=\linewidth]{ 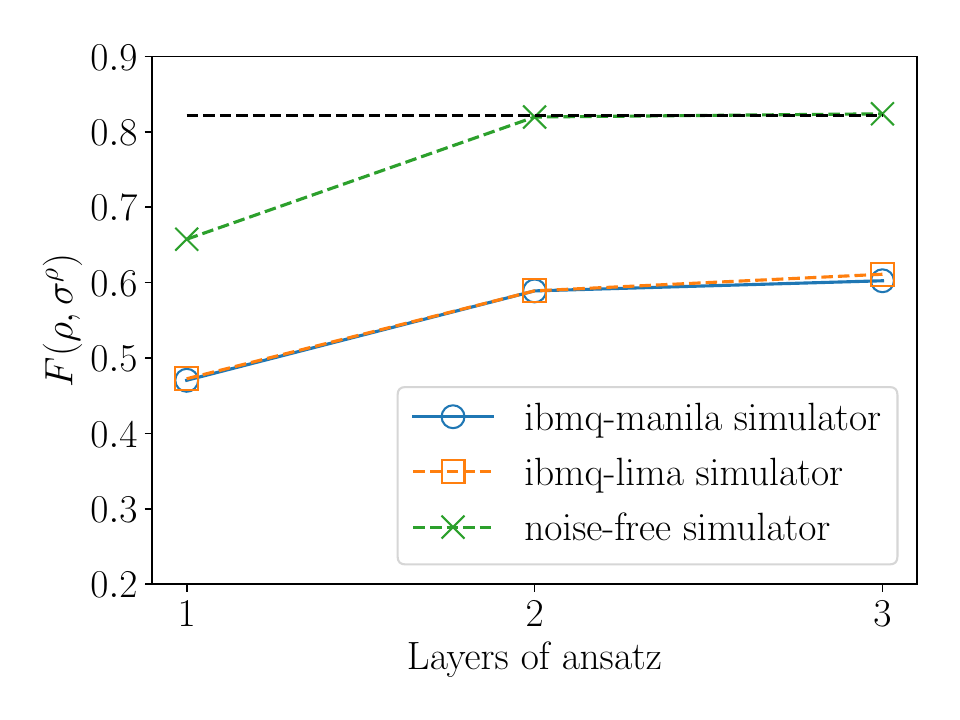}
		\caption{\texttt{CZ} as entangler.}
		\label{fig:qbit1_devicereal_layers3_ansatz_3}
	\end{subfigure}%
	\begin{subfigure}{.5\textwidth}
		\centering
		\includegraphics[width=\linewidth]{ 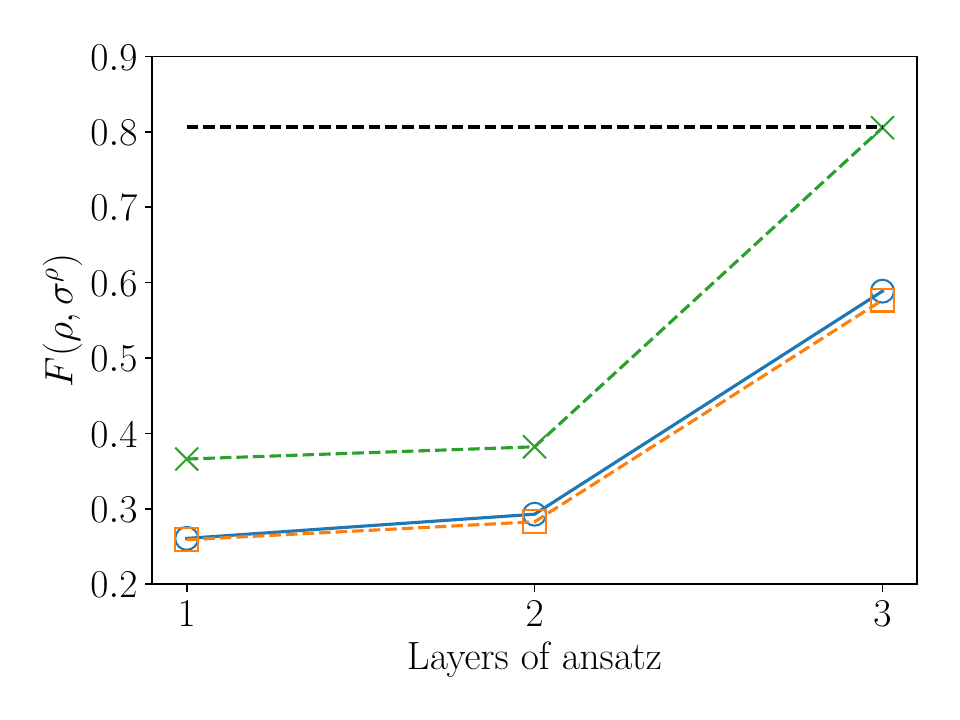}
		\caption{\texttt{CNOT} as entangler.}
		\label{fig:qbit1_devicereal_layers3_ansatz_1}
	\end{subfigure}
	\caption{{The convergence of the truncated fidelity for a single random one-qubit (rank $4$) channel in IBM's noise-free and \texttt{AerSimulator} from Qiskit package with real device backends}. The black dashed line depicts the true value of fidelity. See~\cite{kundu2021qiskit} for the implementation details.}
	\label{fig:real-device-simulate}
\end{figure}

\subsection{Two-qubit quantum channel}

In this section, we show the performance of the VQFE-based device certification with the scaling of the quantum channel. The result is illustrated in~\figref{fig:2-qubit-certification-performance} where we see that the average error in the estimation of the fidelity for one-qubit quantum channels is around 15\% whereas 100 channels are below 5\% error, which is lower compared to the one-qubit case where we saw 500 channels
goes below 5\% error.
\begin{figure}[H]
	\centering   
	\includegraphics[width=80mm]{ 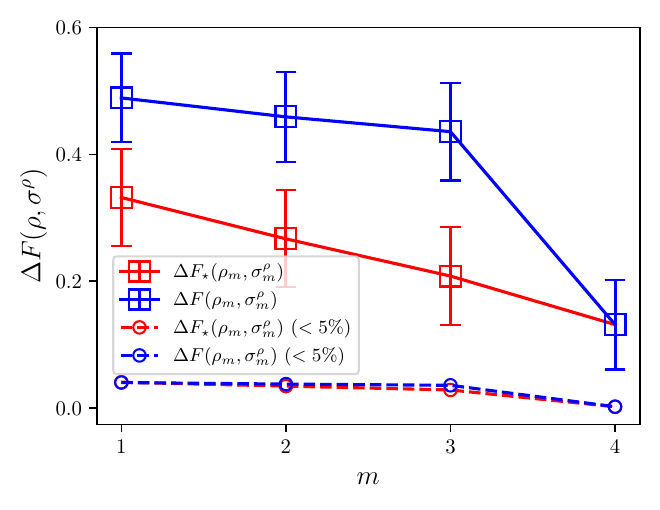}
	\caption{{The average error in fidelity estimation with respect to $m$, obtained by averaging over $1000$ random two-qubit quantum channels ($n = 2$) of rank~$4$}, using IBM quantum computer simulator from \texttt{Qiskit Aer} package.}
	\label{fig:2-qubit-certification-performance}
\end{figure}

\section{RL assisted variational channel certification}
In the previous chapter in {Sec.~\ref{sec:RL-VQSD-results}}, we elaborately discuss how we can harness a reinforcement learning agent to efficiently diagonalize a quantum state. From our elaborated investigation, it was evident that a carefully chosen encoding scheme for the quantum circuit, a dense reward function, and an $\epsilon$-greedy policy for the agent can help us achieve a diagonalizing unitary with very low depth and with a small number of gates. Additionally, with the efficient circuit, we were able to achieve an accuracy better than the existing algorithm.

Using the Jamio\l{}kowski-Choi isomorphism in~\eqref{eq:choi-jamil-iso} we can extend the applicability of RL-VQSD to RL-based quantum device certification. Finding the exact eigendecomposition of a quantum channel is an indispensable part of the VQFE-based certification process. A near-exact approximation of the eigenvalues and eigenvectors of the channel led us to the efficient certification of quantum devices. As we saw the RL can enhance the performance of the diagonalization and give us a better approximation to the eigendecomposition. This in turn enhances the VQFE-based device certification process. 

Meanwhile, till now the ansatz we were using in the VQFE-based certification process were all fixed structures LHEA that contain at least 4-8 parameters per layer and its depth grows rapidly as we increase the layers of the ansatz. This has two issues (1) A large number of parameters can cause trainability issues as we scale up the size of the quantum channel and (2) due to high depth and gate various kinds of quantum noise can impact the accuracy of the certification process. An RL-based certification will reduce the number of parameters in the quantum circuit by minimizing the number of gates and the depth of the circuit. On the other hand, the minimized ansatz proposed by an RL-agent will be more noise-resilient than the fixed structure LHEA.
\section{Conclusion}
In this chapter, as an application of the RL-VQSD method, discussed in the previous chapter, we present a novel approach to certify quantum devices based on the variational quantum fidelity estimation method. The introduced algorithm is primarily based on the combination of Jamio\l{}kowski-Choi isomorphism and variational quantum state diagonalization method. The algorithm is suitable for real-life cases when one needs to provide a convincing argument supporting their claim about a quantum device, keeping in mind the certified quantum device operates on a small number of qubits. This is exactly the case for the NISQ era.

The intrinsic benefit of the VQFE-based channel certification comes from the principle of variational paradigm~\cite{moll2018quantum}. The variational structure of quantum circuits provides us with the advantage of harnessing the strengths of a given quantum architecture. Which benefits avoiding some classes of quantum noises~\cite{mcclean2016theory}, making it more appealing for the NISQ devices.

This is more apparent during the noisy simulations of the VQFE-based certification where we show that amplitude-damping, depolarizing, and random X noise with $1\%$ intensity do not affect the performance of the algorithm. We also simulate our protocol on a real quantum device and show that the performance of the certification process decreases drastically compared to the noiseless scenario.

\section{Takeways}

This chapter introduced a protocol for quantum channel discrimination based on the variational quantum algorithm for fidelity estimation. 
\begin{itemize}
	\item \textbf{Choi-Jamio\l{}kowski isomorphism driven quantum device certification}
	In this chapter, we utilize the variational quantum state diagonalization (that is discussed in the previous chapter) with the Choi-Jamio\l{}kowski isomorphism as described in~\eqref{eq:choi-jamil-iso} to certify quantum channels. Where the approximation of the similarity measure, i.e. fidelity (see~\eqref{eq:fidelity-formula}), plays a very important role. During the construction of the certification scheme, we see that the state-of-the-art fidelity bounds in~\eqref{eq:super-fb} are replaced by the truncated fidelity bounds as given in~\eqref{eq:truncated-fidelity-bound} and in~\eqref{eq:generalized-truncated fidelity}. Hence, for two distinct quantum channels, if the truncated fidelity bounds converge to $1$, we can say that the channels are the same. The algorithm is elaborated in the {Sec.~\ref{sec:vqfe-based-vertification-salgorithm}}.
	
	\item \textbf{Better the approximation of VQSD, the better certification}
	In the~\figref{fig:truncated-fid-error}, we illustrate the dependency of the approximation error in the certification process with respect to the rank of dynamical matrices of quantum channels. As it can be seen from the illustration for 4, 6, 8, and 10 rank dynamical matrices the mean error in truncated fidelity decreases gradually and becomes tighter as we increase the rank in the dynamical matrix. This indicates that during the channel diagonalization process, if we can retrieve all the eigenvalues and eigenvectors of the channel, then the certification efficiency maximizes.
	
	\item \textbf{Little effect of noise and scaling}
	In the~\figref{fig:sim_device_certification_algo} and \textbf{\ref{fig:2-qubit-certification-performance}} we present the performance of the VQFE-based certification process for one- and two-qubit quantum channels averaged over 1000 random quantum channels. It can be seen that the error in fidelity does not get affected by a lot as we scale up the size of the quantum channel. But the number of states with less than $5\%$ error decreases from 480 for one-qubit to 99 in the case of two-qubit.
	
	On the other hand in~\figref{fig:sim-noise-on-certification} we show that depolarizing, amplitude-damping, and random X noise do not affect the performance of the certification scheme if the noise value is below $10\%$ but due to the decoherence error and connectivity constraints, the performance of the certification process underwhelms compared to a single type of noise scenario.
	
	\item \textbf{RL assisted channel certification}
	In the concluding section of the chapter, we briefly describe how a reinforcement learning agent can be utilized to enhance the performance of the VQFE-based certification process. To do this we propose to replace the fixed-structured quantum circuit given in~\eqref{eq:ansatz} by a variable ansatz, where each element (i.e. quantum gate and we refer to {Sec.~\ref{sec:rl-action}} for details) is decided by a policy and a dense reward function (see {Sec.~\ref{sec:rl-action}} for details). This promises to give us a short-depth ansatz containing comparatively very few one and two-qubit gates.
	
	Due to the compactness of the RL-agent proposed ansatz we guarantee to observe better performance of the certification scheme in the noisy scenario and also while running the scheme on the real quantum device.
\end{itemize}

\chapter{Discussions}\label{ch:discussions}

Current quantum hardware is limited in terms of the number of qubits and connectivity among them. Due to very small decoherence time, the qubits are noise-prone and their performance degrades rapidly as we evolve a quantum state through a unitary. Hence it is important to find small quantum circuits that can solve a problem with low error. To address this, in this thesis, we discuss a general framework that automates the search for efficient quantum circuits for variational quantum algorithms (VQAs) based on reinforcement learning (RL) techniques. A lot of research is dedicated to providing a definitive structure of the variational quantum circuit in VQAs. Due to the mystery behind the structure of these variations circuits, they are famously known as \textit{ansatz}. The problem of automating the search for an efficient ansatz is known as the quantum architecture (QAS) problem.

Before discussing the framework, we give a brief introduction to quantum states, gates, and quantum channels in the {Appendix~\ref{app:intro-quantum-computing}} which is followed by a brief discussion of variational quantum algorithms (VQAs) in {Chapter~\ref{sec:intro_vqa}}. The description of VQAs includes a brief introduction to two of the most crucial subroutines of VQAs (1) the construction of cost function and how its landscape varies with the number of qubits in~{Sec.~\ref{sec:cost_function_introduced}}, which is followed by a brief discussion of (2) the problem inspired and problem agnostic ways of constructing an ansatz in {Sec.~\ref{sec:ansatz_introduced}}. Afterwards, in {Sec.~\ref{sec:reinforcement_learning_introdced}} we provide an introduction to reinforcement learning. Keeping the vastness of RL we keep our discussion compact and introduce the concepts relevant to the thesis. Finally, in the {Sec.~\ref{fig:rlqas_ingredients}} we introduce the necessary ingredients for building the RL-assisted quantum architecture search framework that we use in the following chapters of the thesis. 

In the framework, the \textit{environment} is defined by the quantum-classical optimization loop of VQAs, the \textit{RL-state} is represented through the ansatz structure, the \textit{reward function} is itself a function of the cost function (that encodes the problem) and the \textit{action space} is composed of quantum gates. Utilizing this framework we consider two problems based on their relevancy in quantum physics, quantum chemistry, and condensed matter physics. These two problems are variational quantum state diagonalization (VQSD) and variational quantum eigensolver (VQE).

As the first application of the framework, in the {Chapter~\ref{ch:vqsd_using_rl}} we consider the reinforcement learning assisted variational quantum state diagonalization (VQSD) which we term as RL-VQSD~\cite{kundu2024enhancing} in a noiseless scenario. The algorithm focuses on identifying the unitary rotation under which a given quantum state becomes diagonal in the computational basis. We first benchmark the performance of linear hardware efficient ansatz (LHEA) in diagonalizing a two- and three-qubit quantum state. Later we use the RL-based diagonalization method to find the eigenvalues of two-, three-, and four-qubit quantum states. In the case of two-qubit states, we sample arbitrary quantum states from the Haar distribution. Whereas, for three-, and four-qubit cases we choose the reduced ground state of a six-, and eight-qubit Heisenberg model respectively, and diagonalize it. Finally, to demonstrate the hardness of the problem, we run a random-agent method where the RL-agent does not choose the action at a step from the cumulative reward collected from the previous steps but randomly.

In the next application, in {Chapter~\ref{ch:crlvqe}}, we use our framework to construct a curriculum reinforcement learning-based variational quantum eigensolver (CRLVQE) based on ref~\cite{rl-vqe-paper} algorithm where we find the ground state of molecules in a noiseless and noisy scenario. In the noiseless scenario, we consider a \texttt{LiH} molecule of four- and six-qubit and an eight-qubit $\texttt{H}_2\texttt{O}$. Next, we compare the performance of the CRLVQE with the previously proposed quantum architecture search (QAS) algorithms such as RL-VQE in ref.~\cite{ostaszewski2021reinforcement} and quantumDARTS in ref~\cite{wu2023quantumdarts}. The majority of research in QAS is focused on a noiseless scenario and the impact of noise on the QAS remains inadequately explored. So in this chapter, we extend the results in realistic noisy scenarios as well where the noise models are based on IBM quantum devices such as \texttt{ibmq\_mumbai} (with all-to-all connectivity) and \texttt{ibmq\_ourense} (with constrained connectivity). Under these noise models, we solve the two-, three-, and four-qubit $\texttt{H}_2$ molecule and the four-qubit \texttt{LiH} molecule.

While solving these two problems we consider the following settings:

\begin{enumerate}
	\item \textbf{The RL methods:} While dealing with VQSD the RL method we use is a \textit{vanilla curriculum} approach where the user needs to provide a predefined threshold. In the ideal case for the VQSD, the cost function should approach $0$, but in the realistic scenario, it is important to define a predefined threshold ($\zeta$) arbitrarily and completely depends on the number of qubits.
	But while dealing with ground state finding problems using VQE, the exact value of the ground state of molecules is not known; hence, the idea of fixing a predefined threshold is out of scope. That is why we utilize a \textit{feedback-driven curriculum RL} (namely CRL) in {Sec.~\ref{ch:crlvqe}}, which is independent of the prior knowledge regarding the true ground state and does not impose any specific constraints on the initial threshold value.
	
	\item \textbf{The RL-state representation:} The RL-state is defined through a tensor-based binary (TBE) encoding scheme elaborated in \textbf{Section~\ref{sec:binary-encoding-scheme}} in dealing with both the VQSD and VQE. In \textbf{Section~\ref{sec:encoding_comparison}}, we also show that for VQE problems, the TBE outperforms the previously proposed encodings for similar problems~\cite{ostaszewski2021reinforcement}.
	
	\item \textbf{The RL-action space encoding:} The action space is encoded in a one-hot-encoding manner where each action two-qubit a quantum gate defined by a list of four elements. The complete encoding scheme is elaborated in {Sec.~\ref{sec:rl-action}} with an example. The encoding of the action space remains invariant in the VQSD and VQE. Meanwhile,
	
	\item \textbf{The reward function construction:} For the VQSD algorithm the reward function is defined densely as given in \eqref{eq:log_reward}. Whereas, for VQE we utilize a sparse reward function \eqref{eq:old_reward}. In the {Sec.~\ref{sec:reward_comparison}} we show that for the VQE problem both the reward functions provide similar results in terms of ansatz depth, parameters, and accuracy but as the number of successful episodes using \eqref{eq:old_reward} is higher than \eqref{eq:log_reward} the former is more reliable in finding smaller ansatz.
	
	\item \textbf{The action space pruning:} Throughout the thesis we call the action space pruning method as \textit{illegal actions}. When adding a unitary to a qubit at a certain step $s$, if we append the same unitary to the same qubit in step $s+1$, these two operations negate, and the cumulative result is an identity matrix or an idling operation. To restrict redundant such operations and enhance the RL-agent’s exploration efficiency while dealing with VQE we utilize the \textit{illegal actions} technique, elaborated in {Sec.~\ref{sec:ill-actions}} with an example on four-qubits.
	
	\item \textbf{The random halting technique:} This technique is specifically designed to deal with the noise in quantum devices. As it is discussed in {Sec.~\ref{fig:random_halting_effect}}, whenever noise is applied to a quantum circuit, a CPTP channel is applied to the circuit, which not only reduces the performance of the circuit to achieve a task but increases the computation time by many times compared to the noiseless scenario. Hence to keep up with this, we make the total number of steps in an episode a variable by sampling it from a negative binomial distribution. This method is particularly used while finding the ground state of molecules using the VQE algorithm. In~\figref{fig:random_halting_effect} we show that the \textit{random halting} searches the optimal ansatz 3 times faster. 
\end{enumerate}

Finally, in the {Chapter~\ref{ch:vqcd_application_rl_vqsd}} we discuss a novel quantum channel certification~\cite{kundu2022variational} process based on the VQSD algorithm. The algorithm takes two devices as input – the standard device with the
operational capacity already confirmed, and the device for which its conformation
with the standard device to be confirmed. We benchmark the variational quantum channel algorithm in the noiseless scenario for one-, and two-qubit Haar random quantum channels. Next, we investigate how the algorithm is influenced under amplitude damping, depolarizing, and random \texttt{X} noise models. Later on, we benchmark the performance of the performance of the algorithm under real device noise imported from \texttt{ibmq\_manila} and \texttt{ibmq\_lima}. Finally, we discuss as a future work the possibility of constructing an RL-based quantum channel certification process.

\chapter{Conclusions}\label{ch:conclusions}
In this thesis, we introduce a simple yet effective infrastructure for a reinforcement learning-based quantum architecture search algorithm. To show the performance of the thesis, we propose an RL-VQSD~\cite{kundu2024enhancing} and CRLVQE~\cite{anonymous2023curriculum} algorithm that is built upon the framework introduced in this thesis. We solve these problems in the absence and the presence of noise and gather valuable insights. The results are summarized through the following points:

\begin{enumerate}
	\item \textbf{In noiseless scenario:} For both the RL-VQSD and the CRLVQE, the RL-agent provides us with an ansatz (that we call RL-ansatz), which contains a smaller number of parameters and depth while achieving lower error in cost function evaluation compared to state-of-the-art methods. For example, in the {Chapter~\ref{ch:vqsd_using_rl}} we compare the performance of the LHEA ansatz proposed in state-of-the-art VQSD algorithm with the RL-anatz provided by the agent in RL-VQSD. We show that the RL-ansatz is not only of smaller depth and number of gates, but it also helps us achieve a lower error in eigenvalue estimation which is represented through the \figref{fig:2_qubit_eig_conv_HEA_comparison}. Meanwhile, in the {Chater~\ref{ch:crlvqe}}, we compare the performance of our algorithm with previously proposed learning-based methods for quantum chemistry problems such as the RL-VQE~\cite{ostaszewski2021reinforcement}, quantumDARTS~\cite{wu2023quantumdarts} and the net-based QAS~\cite{du2022quantum} algorithms. Our results in {Tab.~\ref{tab:noiseless_rl_algorithm_compatison}} show that the CRLVQE algorithm outperforms the RL-VQE and the quantumDARTS and provides an ansatz with a smaller number of gates, depth, and parameters also keeping a very high accuracy in ground state energy estimation. {From the perspective of complexity theory, the gate count and the total number of parameters in the UCCSD ansatz scale as $1.45\times N^{3.45}\approx\mathcal{O}(N^{3.45})$ under the Jordan-Wigner transformation, where $N$ represents the number of qubits (i.e., spin orbitals).
	In contrast, the gate count in CRLVQE scales as $0.26\times N^3\approx\mathcal{O}(N^3)$\footnote{The curve fitting data was estimated using the \texttt{polyfit} module in \texttt{python}.}. Although the improvement over the number of gates and depth of the parameterized quantum circuits is not significant, CRLVQE demonstrates the potential to address quantum chemistry problems on NISQ devices with circuit containing number of gates $20.68$ times and depth up to $20$ times smaller than UCCSD. These promising results validate \textbf{Hypothesis~\ref{hyp:1st_hypothesis}}.}
	
	On the other hand,
	\item \textbf{In noise scenario:} The performance of the CRLVQE algorithm is investigated under real device noise. Where the noises are modeled from IBM devices such as \texttt{ibmq\_mumbai} (with all-to-all connectivity) and \texttt{ibmq\_ourense} (with constrained connectivity). Through {Tab.~\ref{tab:noisy_rl_algorithm_compatison}}, We show that under the framework of the CRLVQE in the presence of such noise and constrained connectivity, the RL-agent solves various quantum chemistry problems while utilizing a novel action space pruning (which we call \textit{illegal actions}) and variable episode length (calling it \textit{random halting}) techniques with curriculum reinforcement learning. This proves the
	\textbf{Hypothesis~\ref{hyp:2nd_hypothesis}} introduced in the introduction section.
\end{enumerate}

In the final chapter of the thesis, we discuss a variational quantum channel certification algorithm that uses the VQSD algorithm as a subroutine. The results on the \figref{fig:sim_device_certification_algo} and \figref{fig:2-qubit-certification-performance} benchmark the performance of the algorithm for one- and two-qubit quantum channels. Whereas in \figref{fig:sim-noise-on-certification} and \figref{fig:real-device-simulate} we show the performance of the algorithm under various simulated noise models. In the following, we discuss broadly first the strengths and advantages of the discussed framework which is followed by an elaboration of the limitations and accompanying future work.

\section{Strengths and Advantages}
In this section, we outline the merits of the introduced framework, specifically emphasizing the benefits of utilizing the existing infrastructure as a foundation for future research.

\paragraph{A straightforward framework of QAS for VQAs}
In this thesis, we introduce a very simple and efficiently implementable framework for quantum architecture search in VQAs. The strength of the framework lies in the fact that it can perform efficiently in the absence and presence of real device noise. The primary advantage of the framework is that it can be readily adapted to address a wide range of VQAs in deploying quantum architecture search methods.

\paragraph{Binary depth-based encoding scheme for ansatz}
The quantum circuit encoding scheme that is presented in the scope of a subroutine of RL scales quadratically with the number of qubits. In this scheme, each depth of the circuit is encoded in a block of dimension $\left((N+3)\times N\right)$, where the presence or absence of a gate is represented by either $1$ or $0$ respectively. The strength of this approach becomes evident when contrasted with the previously suggested integer gate-based encoding~\cite{ostaszewski2021reinforcement}, surpassing the performance of the prior encoding scheme. 
The advantage of the presented encoding is the fact that it can be easily modified (based on connectivity among qubits and the limitation of the connectivity in two-qubit gates) and adapted to any quantum architecture search method.

\paragraph{Action-space pruning}
The \textit{illegal actions} is an add-on technique presented in the thesis and is a very simple yet effective way to narrow the action space for the RL-ansatz. Which optimizes the search in the action space and provides a significant advantage in minimizing the time per episode while dealing with our framework. The primary strength of the technique lies in its ability to be easily toggled on or off, and seamlessly adapted to any quantum architecture search methods without difficulty.

\paragraph{Accelerating the discovery of efficient ansatz}
Noise in quantum devices increases the simulation longer. To cope with this increasing computational time, in this thesis, we introduce \textit{random halting} scheme and we show that using this method we can search an efficient ansatz three times faster, showcasing its strength.
The advantage of the technique lies in the fact it can be easily adaptable to any quantum architecture search method and real quantum hardware.

\paragraph{Curriculum reinforcement learning}
In VQAs, it is not possible to have prior knowledge of the solution. In the vanilla curriculum of RL, it is required to define prior a threshold value for the cost function. In the case of some VQAs, the threshold can be arbitrary (as can be seen in the RL-VQSD) but for quantum chemistry problems, it is necessary to set the threshold in such a way that the ansatz proposed by the agent can provide the ground state energy. In this thesis, we utilize curriculum reinforcement learning. The strength and advantage of this RL is that, without having prior knowledge of the ground energy, the algorithm iteratively shifts the threshold value after each episode and leads the agent to find the ground state energy.

\section{Limitations and Future Work}

\paragraph{Progress in RL methods}
Despite the promise of RL methods, they encounter challenges concerning sample efficiency, stability, and sensitivity regarding the learning ability of the agent. It is crucial to acknowledge these aspects in the scope of the constantly evolving RL field and mitigate such limitations in future work.

\paragraph{Computational requirements}
The computational demands of training the agent are substantial, primarily in the presence of real device noise. This imposes challenges in both evaluating quantum circuits on a quantum computer and training the algorithm on classical devices. A careful investigation is needed to find more efficient computational strategies.

\paragraph{Scalability to complex problems}
As a follow-up to the discussion regarding the computational challenges, it is crucial to address the scalability of the proposed framework to larger and highly correlated quantum chemistry problems, and in case of diagonalizing many-body Hamiltonians or diverse noise models.
The training of the RL-agent is done from scratch, which promotes the need for exploration to enhance scalability.

\paragraph{Real quantum hardware validation}
The thesis is constrained by the simulation of real device noise and its connectivity which lacks the lack of validation on real quantum hardware, primarily due to existing cost limitations. This can be seen as a motivation for future research endeavors.

\paragraph{Broadening the application scenario}
The application scenario of the proposed framework in the thesis is limited to just three possible problems that have an impact on quantum physics, chemistry, condensed matter physics, and quantum information theory.
A straightforward future work in this line can be broadening the scope of this framework to additional applicable scenarios. For example, in variational quantum linear solver~\cite{bravo2023variational} the authors propose a fixed structure layered hardware efficient ansatz which contains a large number of one- and two-qubit gates.
Our framework has the potential to automate the search for an efficient ansatz that is of shorter depth and gate. Following this motivation it is recently shown that using an Schr{\"od}inger-Heisenberg variational quantum algorithm~\cite{shang2023schrodinger} it is possible to scale find the expectation value of a complex chemistry or condense matter Hamiltonian with shallow ansatz. It is possible to automate the search for the Schr{\"od}inger and the Heisenberg circuits to further optimize the gates and depth of the overall ansatz.

\bibliographystyle{abbrvurl}
\bibliography{bibliography}
\addcontentsline{toc}{chapter}{Bibliography}

\begin{appendix}
\label{appendix}
\chapter{Basics of quantum computing}\label{app:intro-quantum-computing}

\section{Quantum states}
A quantum state encapsulates all the information about a quantum system, including its position, momentum, spin, and other relevant properties. Mathematically, quantum states are represented as vectors in a complex vector space known as a Hilbert space. For one qubit quantum state can be represented by a linear combination of the $\ket{0}$ and $\ket{1}$ as follows
\begin{equation}
	\ket{\psi} = \alpha\ket{0}+\beta\ket{1},
\end{equation}
where $\alpha$ and $\beta$ are complex numbers known as probability amplitudes and $\ket{0}$, $\ket{1}$ are orthonormal ($\bra{0}1\rangle=0$) are known as basis states. We can do a projective measurement on the basis and find the probability of finding the $\ket{0}$ to be $\bra{\psi}0\rangle=|\alpha|^2$ and finding the $\ket{1}$ to be $\bra{\psi}1\rangle=|\beta|^2$.
If the quantum state is a linear combination of $j$ basis vectors, then the quantum state looks as follows
\begin{equation}
	\ket{\chi} = \sum_jc_j\ket{\phi}_j.\label{eq:general_quantum_state}
\end{equation}
Another way to express a quantum state is through density matrices, where we take the outer product of the state \eqref{eq:general_quantum_state}, which results into
\begin{equation}
	\rho = \ket{\psi}\bra{\psi} = \sum p_j \ket{\phi}\bra{\phi}_j,\label{eq:general_density_matrix}
\end{equation}
where $p_j = |c_j|^2$. It is important to define states in this way because if we do not have the whole information, only a statistical mixture of states, then we can not get access to the probability amplitudes of corresponding basis states; hence, it is not possible to represent it with a single ket vector. As such, we see mixed states.

\paragraph{Pure state} If the $\rho$ can not be expressed as a mixture/convex combination of other states, it can be defined as a single ket vector, i.e. $\rho=\ket{\psi}\bra{\psi}$. The purity of a pure state, i.e. $\textrm{Tr}\rho^2=1$.

\paragraph{Mixed state} If the $\rho$ can be expressed as a statistical mixture/convex combination of other states, it is known as a mixed state. The purity of a mixed state, i.e. $\textrm{Tr}\rho^2<1$.

The purity mathematically helps us distinguish between a pure and a mixed state. It tells us that a pure state exhibits maximal coherence and lacks inherent uncertainty in its properties. However, under realistic consideration when, the environment can influence the quantum systems, leading to interactions that cause decoherence and introduce noise. Which results in a mixed state—a statistical ensemble of pure states.

\section{Quantum gates}
Quantum gates are fundamental elements in quantum computing that perform operations on qubits. These gates are represented by unitary operators in quantum mechanics, transforming the quantum state of one or more qubits. Each gate corresponds to a specific operation, such as changing the state of a qubit, entangling multiple qubits, or performing computations on quantum data. For the sake of the thesis, we only define three kinds of one-qubit parameterized quantum gates that rotate a qubit in X (\texttt{RX}), Y (\texttt{RY}) and Z (\texttt{RZ}) direction as follows
\begin{align}
	&\texttt{RX} = \begin{bmatrix} 
		\textrm{cos}(\theta/2) & -i\textrm{sin}(\theta/2)\\
		-i\textrm{sin}(\theta/2) &
		\textrm{cos}(\theta/2)
	\end{bmatrix}, 
	\;\;\; \texttt{RY} = \begin{bmatrix} 
		\textrm{cos}(\theta/2) & -\textrm{sin}(\theta/2)\\
		\textrm{sin}(\theta/2) &
		\textrm{cos}(\theta/2)
	\end{bmatrix},\\
	&\texttt{RZ} = \begin{bmatrix} 
		\textrm{exp}(-i\theta/2) & -0\\
		0&
		\textrm{exp}(i\theta/2)
	\end{bmatrix}.
\end{align}
Additionally, for a two-qubit gate, we use the \texttt{CNOT} gate, where one of the qubits works as a control and modifies the target qubit depending on the input in the control.
Without influencing the control qubit, the gate performs a Pauli X gate on the target qubit when the control qubit is in $\ket{1}$. 

\section{Quantum channels}
Quantum channels represent the mathematical formalism used to describe the evolution of quantum systems under the influence of various physical processes, including interactions with an external environment. A quantum channel is a completely positive and trace-preserving map that describes the evolution of a quantum state. It models how quantum information interacts with its surroundings, resulting in system state changes.

A quantum channel $\mathcal{E}$ can be expressed as a linear, completely positive trace-preserving (CPTP) map that acts on the density operators of the quantum states. Mathematically, the action of the quantum channel $\mathcal{E}$ on a quantum state $\rho$ is represented by
\[
\rho' = \mathcal{E}(\rho),
\]
where $\rho'$ is the density matrix after passing through the channel $\mathcal{E}$. To adhere to the complete positivity requirement, the channel $\mathcal{E}$ should leave all positive semidefinite operators invariant under its action, ensuring the preservation of the positivity of the density matrix. A fundamental property of quantum channels is their linearity, characterized by
\[
\mathcal{E}(\alpha \rho_1 + \beta \rho_2) = \alpha \mathcal{E}(\rho_1) + \beta \mathcal{E}(\rho_2),
\]
for any scalar $\lambda, \mu$ and input states $\rho_1$, and $\rho_2$. This property ensures that the transformation remains consistent and proportional to the input states.

Additionally, the trace-preserving property of the quantum channel dictates that the output state maintains a unit trace, capturing the conservation of probability: 
\[
\text{Tr}(\mathcal{E}(\rho_j)) = \text{Tr}(\rho_j) = 1,
\]
ensuring the preservation of the normalization of quantum states.

Quantum channels are pivotal in various quantum information processing tasks, including quantum communication, error correction, and computation. Understanding the properties and characteristics of quantum channels is crucial for designing efficient communication protocols, error-correcting codes, and quantum algorithms.

In this thesis, we utilize specific types of quantum channels, such as unitary (random \texttt{X}) channels and depolarizing channels. Each chapter where these channels are applied will include a brief explanation to aid conceptual clearance.

{\chapter{Overview of Quantum Architecture Search Methods}
\label{app:overview_of_architecture_search_methods}
As of the time of the final update of this thesis (October 1st, 2024), the following papers caught my attention for introducing innovative Quantum Architecture Search (QAS) methods. The list I present does not necessarily provide a comprehensive overview of these methods but can serve as a summary of various approaches to QAS, without any specific ordering.
\section{QAS with Unsupervised Representation Learning~\cite{sun2024quantum}}
\paragraph{Motivation} In the majority of QAS approaches, the algorithm combines the search space of all possible ansatz architectures with the search algorithm. A major drawback of this approach is that it generally requires evaluating a large number of quantum circuits during the search process, which is computationally demanding, limiting their application to large-scale problems. Predictor-based approaches can overcome this issue but require a large number of labeled quantum circuits, which is time-consuming. To address these challenges, the authors propose Unsupervised Graph Representation Learning (UGRL) for QAS, where the QAS process can be decoupled from UGRL itself.
\paragraph{Method} The authors utilize autoencoders as one of the approaches to UGRL due to their effectiveness in feature representation. Since autoencoders have shown promise in learning good features from graphs, the authors represent the ansatz structure as a Directed Acyclic Graph (DAG). As an encoder, they use Graph Isomorphism Networks to encode the DAG into a latent space. Subsequently, a decoder receives the sampled latent space as input to reconstruct the adjacency matrix. For the search strategy, reinforcement learning and Bayesian optimization are employed.
\section{Training-Free QAS~\cite{he2024training}}
\paragraph{Motivation} The motivation for the training-free approach stems from the fact that existing QAS methods require circuit training to assess an ansatz's performance for a specific problem, and this training process is time-consuming. To overcome the need for the training subroutine, the authors propose a training-free QAS method relying on two proxies.
\paragraph{Method} To bypass the expensive training subroutine, the authors propose two training-free proxies to rank ansatzes. Similar to the previous approach, the ansatz structures are represented as DAGs. The first proxy, called the zero-cost proxy, filters out unpromising ansatzes based on the number of paths in the DAG. The second proxy, based on the expressibility of the ansatz, assesses the performance of the ansatzes. One advantage of these proxies is that they do not involve the cost functions of Variational Quantum Algorithms (VQAs), making them easily transferable across different VQAs.
\section{A Meta-Trained Generator for QAS~\cite{he2024meta}}
\paragraph{Motivation} The primary motivation behind this research is to address the challenges associated with determining optimal circuit structures and gate configurations in QAS. Traditional methods are often computationally intensive, requiring extensive evaluations of numerous ansatz structures to identify high-performance designs. The authors propose a meta-learning framework that leverages prior knowledge from previous tasks, reducing the time and resources needed for QAS.
\paragraph{Method} The authors propose a two-step methodology involving meta-training and meta-testing. In the first step, a Variational Autoencoder (VAE) is trained on a diverse set of tasks. This training allows the VAE to learn optimal circuit structures through amortized inference. The generator can then produce high-performance quantum circuits tailored for new tasks based on the learned meta-knowledge. In the next step, a meta-predictor is developed to filter out suboptimal circuit designs from the generated candidates, ensuring that only the most promising structures are considered for further evaluation.
\section{Self-Attention Enhanced Differentiable QAS~\cite{sun2024sa}}
\paragraph{Motivation} Traditional approaches to differentiable QAS often struggle with inefficiencies due to the discrete nature of ansatz sampling and evaluation. In this paper, the authors utilize self-attention to optimize the search process by enabling a better understanding of the relationships between different components of the ansatz structure.
\paragraph{Method} Self-attention-based differentiable QAS aims to improve gradient-based optimization. Initially, a sampling process begins with architecture weights typically set to zero at the start of training. These weights are updated iteratively during training to refine the circuit design. The self-attention mechanism helps capture relationships among various components of the ansatz. Then, a stable softmax strategy is used to evaluate the probabilities for selecting operations based on their architecture weights. Finally, the architecture parameters are updated using gradient descent techniques.
\section{ZX-Calculus for QAS~\cite{ewen2024application}}
\paragraph{Motivation} The primary motivation behind this research is to address the difficulty of finding the optimal ansatz as the complexity of quantum circuits specific to a problem increases. To tackle this, the authors leverage the ZX-calculus and genetic programming.
\paragraph{Method} The ansatzes are first transformed into ZX-diagrams, which consist of Z and X spiders connected by edges. This representation allows for a more intuitive understanding of quantum computations compared to traditional circuit diagrams. A deterministic procedure for extracting ansatzes from the simplified diagrams is introduced. This process ensures that the resulting diagrams maintain the same functionality as their original counterparts while potentially reducing their size and complexity. Finally, by applying transformation rules from the ZX-calculus, the authors produce shallower circuits with more uniformly allocated gates, enhancing efficiency in QAS.
\section{Surrogate Modeling for QAS~\cite{soloviev2024trainability}}
\paragraph{Motivation} The primary motivation of this research is to address the issues of vanishing gradients and barren plateaus associated with QAS. These issues hinder the effective training of ansatzes. The authors aim to improve the efficiency of QAS by reducing the number of measurements required during the evaluation process and avoiding training in regions where barren plateaus are present.
\paragraph{Method} The paper first develops an online surrogate model to predict the performance of various ansatz structures, allowing for the elimination of poorly performing designs without extensive measurement. This approach reduces the number of evaluations needed, making the search process more efficient. Additionally, a new metric called "information content" is used, which requires only a small set of parameter measurements to estimate the gradient. This helps the algorithm avoid training ansatzes in regions where the circuit's gradient is likely to vanish.
\section{Quantum Information Enhanced QAS~\cite{sadhu2024quantum}}
\paragraph{Motivation} The primary motivation behind this research is to leverage tools from quantum information theory to enhance the performance of reinforcement learning-based QAS. The authors utilize tools such as quantum entanglement and conditional entropy to improve the performance of the RL agent by a factor of two in QAS methods.
\paragraph{Method} In the first step, the authors gain quantum information-theoretic insights into the properties of the ansatzes by evaluating the concurrence and conditional entropy of different ansatz structures. By leveraging these insights, the RL agent can prioritize certain architectural features that are likely to provide better performance. For example, it may focus on circuits that maximize entanglement or optimize fidelity. The reward function is crucial in RL, as it drives the learning process. The authors propose a reward function rooted in quantum information metrics, such as the entanglement measure, to enhance the training process of the agent.
\section{Kolmogorov-Arnold Network for QAS~\cite{kundu2024kanqas}}
\paragraph{Motivation} The main motivation behind this research is to enhance the efficiency and effectiveness of QAS by employing Kolmogorov-Arnold networks as predictors instead of traditional multilayer perceptrons. Multilayer perceptrons, such as feed-forward neural networks, face challenges in interpretability and performance due to their complexity and high parameter counts. The authors aim to demonstrate that Kolmogorov-Arnold networks can outperform multilayer perceptrons by achieving higher probabilities of success in quantum state preparation in both noiseless and noisy environments. Moreover, this approach helps design more efficient quantum circuits for applications in quantum chemistry.
\paragraph{Method} Initially, the authors integrate Kolmogorov-Arnold networks into curriculum reinforcement learning to automate the search for optimal ansatzes. This helps reduce the number of 2-qubit gates compared to multilayer perceptron-based QAS algorithms. In the next step, the authors incorporate two different action space descriptions, observing that in both cases, the Kolmogorov-Arnold network is able to learn the optimal structure of ansatzes for a task by learning the activation functions.}
\chapter{Proof of Truncated Fidelity Bound~\ref{eq:truncated-fidelity-bound}}
\label{app:truncated-bound}

\begin{prop}
		The truncated fidelity bounds are as follows:
		\begin{equation}
			F(\rho_m, \sigma^\rho_m)\leq F(\rho, \sigma)\leq F_*(\rho_m, \sigma_m^\rho).
			\label{eq:trun-bound-appndx}
		\end{equation}
\end{prop}

\begin{proof}
	The $F_*$ two-qubit the generalized fidelity and the RHS of the~\eqref{eq:trun-bound-appndx} follows from the monotonous nature of $F_*$ under the completely positive trace non-increasing maps~\cite{tomamichel2010duality,tomamichel2015quantum}. Such a map can be defined as
	\begin{equation}
		M_m(\rho) = \Pi_m^\rho\rho\Pi_m^\rho,
	\end{equation}
	where $\Pi_m^\rho$ is the projector on the subspace spanned by the eigenvectors for the $m$-largest eigenvalues. This helps us write
	\begin{equation}
		F(\rho,\sigma) = F_*(\rho,\sigma)\leq F_*(M_m(\rho),M_m(\sigma)) = F_*(\rho_m, \sigma_m^\rho) .
	\end{equation}
	Meanwhile, the lower bound i.e. $F(\rho,\sigma)\geq F(\rho_m, \sigma^\rho_m)$ can be derived using the strong concavity of fidelity~\cite{nielsen2010quantum} as follows
	\begin{align}
		F(\rho,\sigma) = F(M_m(\rho) + \Bar{M}_m(\rho),\sigma)&\geq
		\sqrt{p_m}F\left( \frac{M_m(\rho)}{p_m} ,\rho\right) + \sqrt{1-p_m}F\left( \frac{\bar{M}_m(\rho)}{1-p_m},\sigma \right)\nonumber\\
		& \geq \sqrt{p_m}F\left( \frac{\rho_m}{p_m} ,\rho\right) = F(\rho_m,\sigma) = F(\rho_m, \sigma_m^\rho),
	\end{align}
	where $p_m=\text{Tr}\rho_m$ and we have expressed
	\begin{equation}
		\rho = M_m(\rho) + \bar{M}_m(\rho)
	\end{equation}
	where $\bar{M}_m(\rho) = \bar{\Pi}_m^\rho\rho\bar{\Pi}_m^\rho$ and $\bar{\Pi}_m^\rho$ is the orthogonal compliment of $\Pi_m^\rho$.
\end{proof}

\chapter{Illegal actions elaboration}\label{apndix:illegal_action_elaboration}
Here we delve into a comprehensive discussion of the concept of illegal actions. For every $N$-qubit system, there can be a maximum of $N$ actions applied to individual qubit(s). We start with an empty list, representing an initially empty circuit, which can accommodate up to four sub-lists, each corresponding to a specific illegal action.
\begin{equation}
	A_\text{illegal} = [ \underbrace{[]}_{\text{N times}} ].
\end{equation}
If the agent selects an action represented as $[i,j,N,N]$, signifying a $\texttt{CNOT}$ operation with control on qubit $i$ and target on qubit $(i+j)~\pmod{N}$, the illegal action list is adjusted as follows:
\begin{equation}
	A_\text{illegal} = [ [i,j,N,N], \underbrace{[]}_{\text{$(N-1)$ times}} ].
\end{equation}

For the selection of the next action, we employ a selection rule that encourages the agent to opt for the action with the highest estimated action value. To prevent the agent from repeating the action $[i,j, N, N]$ in the next action, we assign a corresponding estimate of $-\infty$ to the illegal action. Therefore, during the greedy action selection process, the agent automatically excludes the illegal action from the available list of actions, pruning the action space.

In the next time step if the agent decides on an action $[N, N,l,m]$ which corresponds to a rotation towards $m$ direction and the rotation to be applied on $l$-th qubit then if $i\neq l$, and $(i+j)$~$\pmod{N}\neq l$, then the illegal actions updated as:
\begin{equation}
	A_\text{illegal} = [ [i,j,N,N], [N,N,l,m],\underbrace{[]}_{\text{$(N-2)$ times}} ],
\end{equation}
else if $i= l$ or $(i+j)$~$\pmod{N}= l$ then the update of illegal actions follows
\begin{equation}
	A_\text{illegal} = [[N,N,l,m],\underbrace{[]}_{\text{$(N-1)$ times}} ].
\end{equation}
In either case, we set the estimate corresponding to the illegal actions to $-\infty$, and the agent does not choose the action or the list of actions.

\chapter{Implementation of components of RLQAS}\label{app:codes}

\section{RL-state implementation}\label{app:rlstate_code}
Here, we give the code for the tensor-based encoding for the ansatz that is used as the observable for the RL-agent. The state is represented as a \texttt{torch tensor} of dimension $T\times((N+3)\times N)$. Where in the code block \texttt{qubits} $=N$, \texttt{num\_step} $=T$. To keep track of the gates and the operations, we use the notion of \texttt{moments} in a quantum circuit. The moment of a quantum circuit is defined by a collection of quantum gates that all act during the same abstract time slice~\cite{cirqmoment}. The RL-state gets filled up by the rule presented through \figref{fig:encoding} for each action. Meanwhile, the action space encoding is straightforward and is presented in \figref{fig:encoding-action-elaborate}. In~\ref{algo:rl_state_code}, we provide a python code example for RL-state (presented by \texttt{state}) encoding where the action space is denoted by the \texttt{action\_list} and for each action, the \texttt{state} gets updated. In order to keep track of the gates and the qubits, we make use of the notion of \texttt{moments}~\cite{cirqmoment}.
\begin{algorithm}[H]
	\caption{RL-state construction and update}
	\begin{algorithmic}[1]
		\Require $qubits$ (the number of qubits in the problem)
		\Require $num\_step$ (number of steps in an episode)
		\Require $action\_list$ (list of possible actions)
		\Ensure Updated RL-state ($state$)
		\Ensure Updated ansatz moments ($moments$)
		
		\State $state \gets \text{torch.zeros}((num\_step, qubits + 3 + 3, qubits))$
		\State $moments \gets [0] * qubits$
		
		\For {each $action$ in $action\_list$}
		\State Extract $ctrl$, $targ$, $rot\_qubit$, and $rot\_axis$ from $action$
		
		\If {$rot\_qubit < qubits$}
		\State $gate\_tensor \gets moments[rot\_qubit]$
		\ElsIf {$ctrl < qubits$}
		\State $gate\_tensor \gets \text{np.max}(moments[ctrl], moments[targ])$
		\EndIf
		
		\If {$ctrl < qubits$}
		\State $state[gate\_tensor][targ][ctrl] \gets 1$
		\ElsIf {$rot\_qubit < qubits$}
		\State $state[gate\_tensor][qubits + rot\_axis - 1][rot\_qubit] \gets 1$
		\EndIf
		
		\If {$rot\_qubit < qubits$}
		\State $moments[rot\_qubit] \gets moments[rot\_qubit] + 1$
		\ElsIf {$ctrl < qubits$}
		\State $max\_of\_two\_moments \gets \text{np.max}(moments[ctrl], moments[targ])$
		\State $moments[ctrl] \gets max\_of\_two\_moments + 1$
		\State $moments[targ] \gets max\_of\_two\_moments + 1$
		\EndIf
		\EndFor
		\label{algo:rl_state_code}
	\end{algorithmic}
\end{algorithm}

\section{Illegal actions implementation}\label{app:ill_action_code}
The following code block provides an elaborated implementation of the \textit{illegal actions} technique. As a quantum physicist, I must concede that the code is entangled to a degree beyond understanding, and I am confident that someone can render it significantly and make it more decoherent.

\begin{lstlisting}[language=Python, caption= A python code example for illegal action (presented by \texttt{illegal\_actions}) technique where the action space is denoted by the \texttt{action\_list} and for each action the \texttt{illegal\_actions} list gets updated. As a final product we get \texttt{illac\_decode}.]
	qubits = #the number of qubits in the problem
	current_action = [qubits]*4
	illegal_actions = [[]]*qubits
	action_list = #list of actions
	
	for action in action_list:
	action=current_action
	ctrl,targ = action[0],(action[0]+action[1])%qubits
	rot_qubit,rot_axis = action[2],action[3]
	
	if ctrl < qubits:
	are_you_empty=sum([sum(l) for l in illegal_actions])
	if are_you_empty!=0:
	for ill_ac_no,ill_ac in enumerate(illegal_actions):
	if len(ill_ac) != 0:
	ill_ac_targ=(ill_ac[0]+ill_ac[1])%qubits
	if ill_ac[2]==qubits:
	
	if ctrl==ill_ac[0] or ctrl==ill_ac_targ:
	illegal_actions[ill_ac_no]=[]
	for i in range(1, qubits):
	if len(illegal_actions[i])==0:
	illegal_actions[i]=action
	break
	elif targ==ill_ac[0] or targ==ill_ac_targ:
	illegal_actions[ill_ac_no]=[]
	for i in range(1, qubits):
	if len(illegal_actions[i])==0:
	illegal_actions[i]=action
	break
	else:
	for i in range(1, qubits):
	if len(illegal_actions[i])==0:
	illegal_actions[i]=action
	break
	
	else:
	if ctrl==ill_ac[2]:
	illegal_actions[ill_ac_no]=[]
	for i in range(1, qubits):
	if len(illegal_actions[i])==0:
	illegal_actions[i]=action
	break
	elif targ==ill_ac[2]:
	illegal_actions[ill_ac_no]=[]
	for i in range(1, qubits):
	if len(illegal_actions[i])==0:
	illegal_actions[i]=action
	break
	else:
	for i in range(1, qubits):
	if len(illegal_actions[i])==0:
	illegal_actions[i]=action
	break                          
	else:
	illegal_actions[0]=action
	
	if rot_qubit<qubits:
	are_you_empty=sum([sum(l) for l in illegal_actions])
	if are_you_empty!=0:
	for iac_no, iac in enumerate(illegal_actions):
	
	if len(iac)!=0:
	ill_ac_targ=(iac[0]+iac[1])%qubits
	if iac[0]==qubits:
	
	if rot_qubit==iac[2] and rot_axis!=iac[3]:
	illegal_actions[iac_no]=[]
	for i in range(1, qubits):
	if len(illegal_actions[i])==0:
	illegal_actions[i]=action
	break
	elif rot_qubit!=iac[2]:
	for i in range(1, qubits):
	if len(illegal_actions[i])==0:
	illegal_actions[i]=action
	break
	else:
	if rot_qubit==iac[0]:
	illegal_actions[iac_no]=[]
	for i in range(1, qubits):
	if len(illegal_actions[i])==0:
	illegal_actions[i]=action
	break       
	elif rot_qubit==iac_targ:
	illegal_actions[iac_no]=[]
	for i in range(1, qubits):
	if len(illegal_actions[i])==0:
	illegal_actions[i]=action
	break
	else:
	for i in range(1, qubits):
	if len(illegal_actions[i])==0:
	illegal_actions[i]=action
	break 
	else:
	illegal_actions[0]=action
	
	for indx in range(qubits):
	for jndx in range(indx+1, qubits):
	if illegal_actions[indx]==illegal_actions[jndx]:
	if jndx!=indx +1:
	illegal_actions[indx]=[]
	else:
	illegal_actions[jndx]=[]
	break
	
	for indx in range(qubits-1):
	if len(illegal_actions[indx])==0:
	illegal_actions[indx]=illegal_actions[indx+1]
	illegal_actions[indx+1]=[]
	
	illac_decode=[]
	for key, contain in dictionary_of_actions(qubits).items():
	for ill_action in illegal_actions:
	if ill_action==contain:
	illac_decode.append(key)
\end{lstlisting}

\section{3-stage Adam-SPSA pseudocode and hyperparameters setting}\label{algo:adamspsa}
In this section we provide the pseudocode for multistage Adam-SPSA optimization algorithm and the parameters that are fixed while running the algorithm.
\begin{figure}[H]
	\includegraphics{ 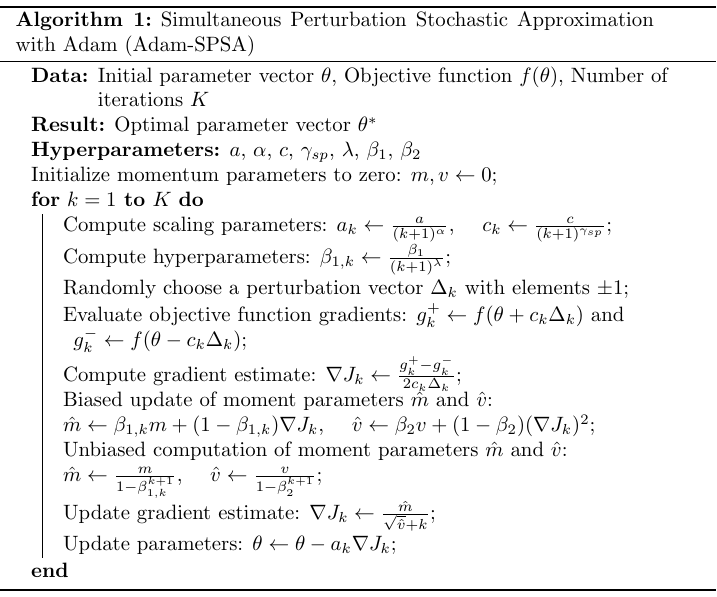}
\end{figure}

\begin{table}[H]
	\caption{\small The hyperparameters of Adam-SPSA optimizer used during the noisy simulations.
		In the noisy simulation of two-, and three-qubit problems we used $1$-stage version of the algorithm, therefore only single maximum function evaluation hyperparameters are given.
		The parameters within the curly brackets denote the maximum number of function evaluations in the $3$-stage version of the algorithm. We provide Max fevals both for $1$-stage and their $3$-stage equivalents. }
	\label{tab:spsa_hyperparams}
	\centering
	\small
	\resizebox{1\textwidth}{!}{%
		\begin{tabular}{@{}cccccccccc@{}}
			\hline
			\makecell{Molecule} & $a$ & $\alpha$ & $\beta_1$ & $\beta_2$ & $c$ & $\gamma_{sp}$ & $\lambda$ & Max fevals & Shots \\
			\hline
			\texttt{H}$_2$-2   &1.2104 &0.9531 &0.9414 & 0.9983& 0.1039 & 0.0984& 0.9277& 500& $10^{3}$   \\
			\texttt{H}$_2$-3   &0.5188& 0.9859&0.716&0.6265&0.0938  &0.0974&0.6483  &500 & $10^{4}$      \\
			\texttt{LiH}-4   & $1.2324$ &  $0.9709$  &    $0.6114$     &   $0.9326$        & $0.2215$ &  $0.1485$ &  $0.9772$    & \makecell{$1600$\\ \{$1191$, $357$, $119$\}}  \makecell{$3300$\\ \{$2383$, $715$, $238$\}} & $10^{6}$ \\
			\texttt{LiH}-6   &1.7564  &0.8365   &0.6841         &0.9048          &0.1068  &0.1549   &0.1223 &  \makecell{$2000$\\ \{$1430$, $429$, $143$\}}  & $10^{8}$ \\
			\hline
	\end{tabular}}
\end{table}

\chapter{Details on the training time}\label{apndix:time-resource}
In this section we elaborate on the time it takes per episode in RL-VQSD (see chapter~\ref{ch:vqsd_using_rl}) and CRLVQE (see chapter~\ref{ch:crlvqe}) algorithms. In the case of CRLVQE we note time in the abscence of the Random Halting (wo-RH) and in the presence of the Random Halting (RH) mechanism.
\begin{table}[h!]
	\centering
	\begin{tabular}{c|cc}
		\hline
		Problem & qubits & Avg. time per episode (in seconds) \\
		\hline
		\multirow{2}{6em}{RL-VQSD} & two & 9.21 \\
		& three & 37.74\\
		\cline{2-3}
		\multirow{4}{6em}{CRLVQE} & four (\texttt{H}$_2$, wo-RH) & 0.75   \\
		& four (\texttt{LiH}, wo-RH) & 1.56\\
		& four (\texttt{LiH}, RH) & 1.45\\
		& six (\texttt{LiH}, wo-RH) & 3.90  \\
		& six (\texttt{LiH}, RH) & 2.76  \\
		& six ($\texttt{H}_2\texttt{O}$, wo-RH) & 45.02  \\
		\hline
	\end{tabular}
	\caption{The record of the training time it takes for the RL-agent to complete the same number of episodes in diagonalizing a two, three-qubit quantum state and finding the ground state of four, six and eight qubit chemical Hamiltonians.}
	\label{tab:training-time-record}
\end{table}
The times are recorded in a CPU and GPU device with specifications described in~\ref{tab:gpu-cpu-details} for RL-VQSD and CRLVQE algorithms.
\begin{table}[h!]
	\centering
	\begin{tabular}{c|c|c}
		\hline
		Algorithm & Device & Specifications \\
		\hline
		\multirow{2}{6em}{RL-VQSD} & CPU & \texttt{Intel(R) Core(TM) i7-10700KF CPU @ 3.80GHz}  \\
		& GPU & \texttt{NVIDIA GA102 64 bits}   \\
		\cline{2-3}
		\multirow{2}{6em}{CRLVQE} &CPU & \texttt{Intel(R) Core(TM) i7-10700KF CPU @ 3.80GHz}  \\
		&\multirow{2}{2em}{GPU} & \texttt{NVIDIA A100 64 bits}   \\
		&                        & \texttt{NVIDIA GA102 64 bits}\\
		\hline
	\end{tabular}
	\caption{{The details of GPU and CPU resources utilized to record the training time for CRLVQE algorithm elaborated in chapter~\ref{ch:crlvqe}}.}
	\label{tab:gpu-cpu-details}
\end{table}

\end{appendix}


\backmatter


\end{document}